%% file: AFree.tex
\documentclass{article}

\input{TeX/CEpreamble}

\input{TeX/CEmacros}

\input{TeX/macros.JohYde.tex}

\title{Cut elimination for the alternation-free modal mu-calculus
	\thanks{Funded by the Dutch Research Council grants [617.001.857, OCENW.M20.048] and the Knut and Alice Wallenberg Foundation [2020.0199].}
}

\author[1]{Bahareh Afshari
}
\author[2]{Johannes Kloibhofer
}
\affil[1]{Department of Philosophy, Linguistics and Theory of Science, University of Gothenburg, Sweden}
\affil[2]{Institute for Logic, Language and Computation, University of Amsterdam, The Netherlands}

\begin{document}
	\maketitle
	
	\begin{abstract}
We present a syntactic cut-elimination procedure for the alternation-free fragment of the modal $\mu$-calculus.
Cut reduction is carried out within a {\em cyclic proof system}, where proofs are finitely branching but may be non-wellfounded.
The structure of such proofs is exploited to directly transform a cyclic proof with cuts into a cut-free one, without detouring through other logics or relying on intermediate machinery for regularization.
Novel ingredients include the use of {\em multicuts} and results from the {\em theory of well-quasi-orders}, the latter used in the termination argument.

\medskip

    
	\end{abstract}
	
	\input{sec.intro.tex}

	\input{sec.prelim.tex}

	\input{app.mathPrelim.tex}
	
	\input{sec.proofsystem.tex}
	
	\input{sec.cutStrat.tex}

\input{sec.cutImportant.tex}

\input{sec.cutUnimportant.tex}
	
	\input{sec.contractions.tex}
	
	\input{sec.cutElim.tex}

	\input{sec.conc.tex}
	
	\bibliographystyle{plain}
	\bibliography{checked.bib}

	\appendix
	\input{app.constrImp.tex}
	\input{app.cutRed.tex}

	\input{app.contrRed.tex}

\end{document}

%% file: TeX/CEpreamble.tex
\usepackage[english]{babel}

\usepackage{authblk}

\usepackage[a4paper,top=2cm,bottom=2cm,left=3cm,right=3cm,marginparwidth=1.75cm]{geometry}

\usepackage{amsmath,amsthm,amssymb,latexsym}
\usepackage{graphicx}
\usepackage[colorlinks=true, allcolors=blue]{hyperref}
\usepackage{xcolor} 
\usepackage{xspace} 
\usepackage{xargs} 
\usepackage{ebproof} 
\usepackage{enumitem} 
\usepackage{mathtools} 
\usepackage{ stmaryrd } 
\usepackage{mdframed} 
\usepackage{outlines} 
\usepackage{subcaption} 
\usepackage{tikz} 

\theoremstyle{plain}
\newtheorem{theorem}{Theorem}
\newtheorem{lemma}[theorem]{Lemma}
\newtheorem{corollary}[theorem]{Corollary}

\newtheorem{proposition}[theorem]{Proposition}
\theoremstyle{definition}

\newtheorem{remark}[theorem]{Remark}

\newtheorem{definitioninner}[theorem]{Definition}

\newenvironment{definition}
{%
	\pushQED{\qed}%
	\begin{definitioninner}
	}
	{%
		\popQED
	\end{definitioninner}
}

\usepackage{fixme}
\fxsetup{
	status=draft,
	author=,
	layout=footnote,
	theme=color
}

\newcommand{\claim}[1]{\medskip\noindent\textbf{\underline{Claim #1}:}}
\newcommand{\claimproof}[1]{\smallskip\noindent\textbf{{Proof of Claim #1}:}}
\newcommand{\claimproofend}{\hfill$\dashv$\medskip}

\newcommand{\new}[1]{{\color{black} #1}}

%% file: TeX/CEmacros.tex
\newcommand{\calA}{\mathcal{A}} 

\newcommand{\calD}{\mathcal{D}}
\newcommand{\calG}{\mathcal{G}}

\newcommand{\calM}{\mathcal{M}}

\newcommand{\calS}{\mathcal{S}}
\newcommand{\calT}{\mathcal{T}}
\newcommand{\calQ}{\mathcal{Q}}

\newcommand{\sfC}{\mathsf{C}} 
 
\newcommand{\sfG}{\mathsf{G}} 
\newcommand{\sfM}{\mathsf{M}} 
\newcommand{\sfN}{\mathsf{N}}

\newcommand{\sfR}{\mathsf{R}}
\newcommand{\sfS}{\mathsf{S}}

\newcommand{\Tau}{\mathsf{T}}

\renewcommand{\|}{~|~}

\renewcommand{\phi}{\varphi}

\newcommand{\mybar}[1]{\overline{#1}}
\newcommand{\Nat}{\mathbb{N}}

\renewcommand{\emptyset}{\varnothing}
\newcommand{\multi}{\sigma}
\newcommand{\lDM}{\ensuremath{<_{\mathrm{DM}}}}
\newcommand{\geDM}{\ensuremath{\geq_{\mathrm{DM}}}}
\newcommand{\norm}[1][\cdot]{\ensuremath{\llbracket #1 \rrbracket}}
\newcommand{\lenf}[2]{\ensuremath{{L}[#1,#2]}\xspace}

\newcommand{\subst}[2]{[{#1}/{#2} ]} 

\newcommand{\Lor}{\bigvee}
\renewcommand{\land}{\wedge}
\renewcommand{\lor}{\vee}

\newcommand{\Clos}{\mathsf{Clos}}
\newcommand{\tracestep}{\to_{\mathsf{C}}}
\newcommand{\trace}{\twoheadrightarrow_{\mathsf{C}}}
\newcommand{\sub}{\trianglelefteq}

\newcommand{\twoheadleftrightarrow}{%
	\mathrel{\twoheadleftarrow\mkern-14mu\twoheadrightarrow}%
}
\newcommand{\equic}{\twoheadleftrightarrow_{\mathsf{C}}}
\newcommand{\Set}{\mathsf{Set}}
\newcommand{\seteq}{=_{\Set}}

\newcommand{\traversedTrans}{traversed leaf reduction algorithm\xspace}

\newcommand{\RuWeak}{\ensuremath{\mathsf{w}}\xspace}
\newcommand{\RuContr}{\ensuremath{\mathsf{c}}\xspace}
\newcommand{\RuCut}{\ensuremath{\mathsf{cut}}\xspace}
\newcommand{\RuCutMulti}{\ensuremath{\mathsf{multicut}}\xspace}
\newcommand{\RuMix}{\ensuremath{\mathsf{mix}}\xspace}
\newcommand{\RuU}{\ensuremath{\mathsf{u}}\xspace}
\newcommand{\RuF}{\ensuremath{\mathsf{f}}\xspace}
\newcommand{\RuFp}[1]{\ensuremath{{#1}}\xspace}
\newcommand{\RuMu}{\RuFp{\mu}}
\newcommand{\RuNu}{\RuFp{\nu}}
\newcommand{\RuEta}{\RuFp{\eta}\xspace}
\newcommand{\Ru}{\ensuremath{\mathsf{R}}\xspace}
\newcommand{\RuDischarge}[1][\dx]{\ensuremath{\mathsf{D}^{#1}}\xspace}

\newcommand{\Tokens}{\mathcal{D}}
\newcommand{\disA}{\ensuremath{\dagger}}
\newcommand{\disB}{\ensuremath{\ddagger}}

\usepackage{stackengine,scalerel} 

\newcommand\DDDag{%
	\sbox0{\ddag}\stretchrel*{%
		\stackengine{-.6\ht0}{\ddag}{\ddag}{O}{c}{F}{F}{S}}{\ddag}%
}
\newcommand{\disC}{\ensuremath{\DDDag}}

\newcommand{\dx}{\disA}
\newcommand{\dy}{\disB}
\newcommand{\dz}{\disC}
\newcommand{\extraValue}{\ensuremath{\mathsf{o}}\xspace}

\newcommand{\Prop}{\ensuremath{\mathsf{Prop}}\xspace}
\newcommand{\Var}{\ensuremath{\mathsf{Var}}\xspace}

\newcommand{\PDL}{\ensuremath{\mathsf{PDL}}\xspace}
\newcommand{\PDLf}{\ensuremath{\mathsf{PDL}_f}\xspace}

\newcommand{\GKe}{\ensuremath{\mathsf{GKe}}\xspace}

\newcommand{\CTL}{\ensuremath{\mathsf{CTL}}\xspace}
\newcommand{\LTL}{\ensuremath{\mathsf{LTL}}\xspace}

\newcommand{\sccr}{\twoheadrightarrow}
\newcommand{\cyclicPT}{\calT_\pi^C}

\newcommand{\depth}{\ensuremath{\mathsf{depth}}\xspace}
\newcommand{\rank}{\ensuremath{\mathsf{rank}}\xspace}

\newcommand{\comp}{\ensuremath{\mathsf{comp}}\xspace}

\newcommand{\unf}[2][]{#2^{*{#1}}}
\renewcommand{\merge}[3]{\ensuremath{[#1]#2[#3]}\xspace}

\newcommand{\del}[3]{\ensuremath{#1(#2,#3)}}

\newcommand{\nmf}{\ensuremath{\mathsf{nmf}}}

\newcommand{\RC}{\ensuremath{\mathsf{eFT}}\xspace}
\newcommand{\FT}{\ensuremath{\mathsf{FT}}\xspace}
\newcommand{\uptoM}[1]{#1_{\geq m}}
\newcommand{\uptoMC}[1]{#1_{\geq m}^-}

\newcommand{\ldia}{ \scalebox{0.8}{\rotatebox[origin=c]{45}{$\square$}}}
\newcommand{\lbox}{\scalebox{0.8}{$\square$}}

\newcommand{\lm}{\mbox{\ooalign{\ld \cr \hidewidth\raise.05ex\hbox{$* \mkern3.1mu$}\cr}}} 
\newcommand{\lbm}{\ooalign{$\lbox$ \cr \hidewidth\raise.05ex\hbox{$* \mkern5.5mu$}\cr}\hspace{-0.1cm}} \newcommand{\lcm}{\ooalign{$\lbox$ \cr \hidewidth\raise.05ex\hbox{$\cdot \mkern2.9mu$}\cr}}

\newcommand{\scs}{\mathsf{scst}}
\newcommand{\Leaves}[1]{\mathsf{Leaves}(#1)}

\newcommand{\origin}{\ensuremath{\mathsf{origin}}}
\newcommand{\oriPath}{\ensuremath{\mathsf{hist}}}
\newcommand{\oriNode}{\ensuremath{\mathsf{origin}}}
\newcommand{\depthT}{\ensuremath{\mathsf{oriDepth}}\xspace} 
\newcommand{\leqT}{\ensuremath{\leq_t}}


%% file: TeX/macros.JohYde.tex
\newcommand{\dia}{\Diamond}

\newcommand{\atneg}[1]{\overline{#1}}

\newcommand{\muML}{\mathcal{L}_{\mu}}
\newcommand{\muMLca}{\mathcal{L}_{\mu}^{\mathit{ca}}}
\newcommand{\AFMC}{\muML^{\mathit{af}}}

\newcommand{\ol}[1]{\overline{#1}}   

\newcommand{\Focus}{\ensuremath{\mathsf{Focus}}\xspace}
\newcommand{\FocusM}{\ensuremath{\mathsf{Focus^c}}\xspace}
\newcommand{\FocusS}{\ensuremath{\mathsf{Focus}^*}\xspace}

\newcommand{\AxLit}{\ensuremath{\mathsf{Ax}}\xspace}
\newcommand{\RuOr}{\ensuremath{{\lor}}\xspace}
\newcommand{\RuAnd}{\ensuremath{{\land}}\xspace}
\newcommand{\RuBox}{\ensuremath{{\lbox}}\xspace}

\newcommand{\last}{\mathsf{last}}

\newcommand{\isdef}{\mathrel{:=}}

\renewcommand{\phi}{\varphi}



%% file: sec.intro.tex
\section{Introduction}\label{sec.Intro}

The modal $\mu$-calculus $\muML$ is an elegant logic for reasoning about properties of transition systems and programs. By extending modal logic with least and greatest fixpoint operators, it combines good computational behavior with high expressive power: the $\mu$-calculus enjoys the finite-model property \cite{Kozen1988}, its model checking problem can be solved in quasi-polynomial time \cite{Calude2017}, and it captures exactly the bisimulation-invariant fragment of monadic second-order logic \cite{Janin1996}.

Although the modal $\mu$-calculus is defined in a general way, many important concepts can already be captured within specific fragments of the logic. In fact, a range of dynamic and temporal logics -- such as \PDL, \LTL, and \CTL{} -- can be expressed as fragments of the modal $\mu$-calculus by appropriately restricting the use of fixpoint operators.
In particular, many of these logics fall within the \emph{alternation-free modal $\mu$-calculus} $\AFMC$, the fragment of $\muML$ where least and greatest fixpoints are not interleaved.
While $\AFMC$ is strictly less expressive than $\muML$ \cite{Bradfield1998}, it attains the same expressive power over certain classes of frames \cite{Alberucci2009,Gutierrez2014}.  
Moreover, $\AFMC$ enjoys Craig interpolation \cite{DAgostino2018}, further reinforcing its importance as a logic worthy of independent study.

Kozen’s seminal work~\cite{Kozen1983} introduced a Hilbert-style axiomatization of the modal $\mu$-calculus.
Although the system is natural and elegant, proving its completeness turned out to be challenging and Kozen was able to establish completeness only for a fragment of the logic.
Much later, Walukiewicz~\cite{Walukiewicz2000} succeeded in proving completeness for the full modal $\mu$-calculus.
However, his proof, recast in sequent-calculus style, crucially relies on the cut rule, and it remains a major open problem whether Kozen’s system remains complete without cut.

An alternative characterization of the set of validities of $\mu$-calculus formulas was proposed by Niwiński and Walukiewicz \cite{Niwinski1996} in the form of a two-player game, in which  the winning positions of one player correspond precisely to valid $\muML$-formulas. This validity checking game can be readily utilized to present a cut-free, non-wellfounded proof system.

Only more recently have \emph{non-wellfounded proof systems} emerged as the preferred proof-theoretic framework for studying logics with inductive or fixpoint behavior, and they have since been explored in a variety of settings~\cite{Santocanale2002,Brotherston2006,Baelde2016,Simpson2017,Kuperberg2021}.
In such systems, proofs are finitely branching but may contain infinite branches within the proof tree.
To ensure soundness, non-wellfounded proofs must satisfy a \emph{global soundness condition}, which enforces a notion of `progress' along infinite paths.
The precise formulation of this condition determines the key subtleties and characteristics of the proof system.
 
For the modal $\mu$-calculus, Jungteerapanich~\cite{Jungteerapanich2010} and Stirling~\cite{Stirling2014} introduced annotated proof systems that enrich sequents with additional information, enabling a finite cyclic representation of otherwise infinite proof trees. In this setting, the global soundness condition reduces to a constraint on paths between non-axiomatic leaves and their corresponding `companion' nodes.

In \cite{Marti2021}, Marti and Venema present a so-called ``focus" cyclic proof system for the alternation-free $\mu$-calculus based on the annotated proof system of Jungteerapanich and Stirling. Their system, denoted \Focus and employed in this article, features a simpler and more elegant management of annotations, owing to the absence of fixpoint alternations in $\AFMC$.
In the same work, the authors established the admissibility of the cut rule, whereas in the present article we provide a syntactic cut-elimination procedure for the alternation-free $\mu$-calculus. 

Providing syntactic cut-elimination procedures is of both theoretical and practical significance. Even when working within a cut-free proof system, composing cut-free proofs typically introduces cuts. Although such cuts can sometimes be removed by semantic means, \emph{syntactic cut elimination} gives a direct normalised proof suitable for immediate applications.
In \cite{Afshari2025}, for example, partial cut elimination for a non-wellfounded proof system is used to reconstruct Walukiewicz' completeness proof for $\muML$  proof-theoretically.

With the growing popularity of non-wellfounded proofs, it is not surprising that cut elimination has been investigated across a range of non-wellfounded proof systems \cite{Fortier2013, Savateev2020, Shamkanov2025,Afshari2024, Acclavio2024, Borja2024, Baelde2016,Baelde2022,Saurin23,Bauer2025, Afshari2025}. 
Many of these studies, however, deal with very simple forms of global soundness conditions.
Cut-elimination procedures for systems with complex global soundness conditions have so far been developed primarily in the context of linear logic~\cite{Baelde2016,Baelde2022,Saurin23}.
In~\cite{Bauer2025}, Bauer and Saurin extend this line of work to the modal $\mu$-calculus by encoding modalities in linear logic via super-exponentials. 
 Since this approach provides only an indirect method, a \emph{direct} proof of cut elimination for the modal $\mu$-calculus remains of interest.

With the exception of~\cite{Afshari2024}, existing cut-elimination procedures operate on non-wellfounded proof systems.
When working within an infinitary system, the main challenges for achieving cut elimination are twofold: first, ensuring the termination of cut-reduction steps, which are applied along infinite branches; and second, verifying that the zipping of these infinite branches preserves the system’s global soundness condition.
In contrast, working with cyclic proofs introduces an additional complication: the resulting cut-free derivation trees are often no longer regular, and therefore fail to constitute valid proofs within the cyclic proof system itself.
  
 In this article, we present a cut-elimination procedure for the cyclic proof system \Focus for the alternation-free modal $\mu$-calculus. Compared to existing work, our result is noteworthy for the following reasons.

\begin{enumerate}
	\item \textbf{Directness.} Our method applies to a cyclic proof and outputs a cyclic cut-free proof without appealing to intermediate machinery for regularizing the end proof. Working on the cyclic proof allows us to employ induction invariants utilizing the structure of cyclic proof trees and to eliminate cuts depending on where they are located. 
	\item  \textbf{Expressiveness.} Previous cut-elimination strategies have been developed for fragments of the modal $\mu$-calculus such as $\mathsf{Grz}$ \cite{Savateev2020, Borja2024} or modal logic with transitive closure \cite{Afshari2024, Shamkanov2025}. Here we address a larger fragment with a more complex global soundness condition.
	\item \textbf{Transparency.} Bauer \& Saurin \cite{Bauer2025} established cut elimination for the full modal $\mu$-calculus by encoding modalities in linear logic. In contrast, our approach avoids any detours through other proof systems. This is preferable from a practical, as well as from a theoretical point of view, as it provides a more transparent explanation of why cut-elimination holds in a certain system.
\end{enumerate}

In our proof strategy we make essential use of the structure of cyclic proofs. 
We distinguish between cuts inside cycles, which we call \emph{unimportant}, and cuts outside cycles, referred to as \emph{important}, and handle them differently.
We show that the cut formulas of unimportant cuts do not interfere with the global soundness condition.
Consequently, such cuts can be pushed upwards, away from the root, allowing successful repeats to be identified below them.

The treatment of important cuts is more intricate.
Our approach builds on the strategy developed for modal logic with the eventually operator in~\cite{Afshari2024}.
Let $\pi_L$ and $\pi_R$ be the left and right subproofs rooted at the respective left and right premise of an important cut.
The key idea is to push the important cut upwards while retaining annotations on formulas in $\pi_L$ and removing annotations on formulas in $\pi_R$. Progress on repeat-paths in $\pi_L$ is preserved in the resulting proof, enabling the identification of successful repeats below the cuts.

For the alternation-free $\mu$-calculus this procedure becomes more involved, particularly because conjunctions and disjunctions may occur in the scope of fixpoints. 
To handle this complexity, we introduce \emph{multicuts}. 
This, however, further complicates the elimination of important cuts, as it requires determining which premises of a multicut should retain annotations on their formulas.

The introduction of multicuts requires working in a system where sequents are defined over multisets of formulas.
Consequently, an explicit contraction rule becomes necessary, which in turn poses additional challenges for the elimination of important cuts. 
To overcome these challenges, we first eliminate contractions from cut-free proofs, thereby reducing the problem to eliminating important cuts in proofs without contraction.
We then establish the termination of this procedure using known results on well-quasi-orders.

\paragraph{Overview of the paper.}
In Section~\ref{sec.prelim} we introduce the alternation-free modal $\mu$-calculus, and state some facts about multisets and well-quasi-orders that are used later on. The cyclic proof system \Focus is defined in Section~\ref{sec.proofsystem}. 
In Section~\ref{sec.ceStrategy}, we lay the groundwork for the cut-elimination procedure: we provide a high-level overview of the setup, introduce the notions of important and unimportant cuts, and define a normal form for \Focus-proofs.
We deal with important cuts in Section~\ref{sec.importantCuts}, with unimportant cuts in Section~\ref{sec.unimportantCuts} and eliminate contractions in Section~\ref{sec.contractions}.
Results of these sections are combined in the proof of the cut-elimination theorem in Section~\ref{sec.CEtheorem}. 
In Section~\ref{sec.conclusion} we discuss possible directions for further work.

%% file: sec.prelim.tex
\section{Preliminaries}\label{sec.prelim}

\subsection{The alternation-free modal $\mu$-calculus}\label{sec.afMuCalculus}

Let \Prop be a fixed set of propositions and \Var be a fixed set of variables such that $\Prop \cap \Var = \emptyset$. The \emph{formulas} of the modal $\mu$-calculus are generated by the grammar
\[
\phi ::= 
   p \| \mybar{p} \| x \| \phi\lor\phi \| \phi\land\phi \|
   \ldia\phi \| \lbox\phi 
   \| \mu x. \phi \|  \nu x. \phi,
\]
where $p \in \Prop$ and $x \in \Var$. We write
$\muML$ for the set of formulas in the modal $\mu$-calculus. 
Note that for simplicity our language only contains one modality. However, all results of this paper can be readily extended to the \emph{polymodal case}, where a set of modalities is considered.

Formulas of the form $\mu x . \phi$ ($\nu x . \phi$) are called 
\emph{$\mu$-formulas} (\emph{$\nu$-formulas}, respectively); formulas 
of either kind are called \emph{fixpoint formulas}. We write $\eta$ to denote an arbitrary fixpoint operator $\mu$ or $\nu$ and define $\mybar{\mu} = \nu$ and $\mybar{\nu} = \mu$.
Formulas that are of the form $\lbox \phi$ or $\ldia \phi$ are 
called \emph{modal formulas} and formulas of the form $\phi \lor \psi$ and $\phi \land \psi$ are called \emph{boolean formulas}.

We say that an occurrence of a variable $x$ in a formula $\phi$ is \emph{bound} if it is in the scope of an $\eta x$-fixpoint operator and \emph{free} otherwise. We call a formula $\phi$ \emph{closed} if all variable occurrences in $\phi$ are bound. 

We call a formula $\phi \in \muML$ \emph{alternation-free} if it satisfies the following: for any subformula $\eta x. \psi$ of $\phi$ no free occurrence of $x$ in $\psi$ is in the scope of an $\ol{\eta}$-operator in $\psi$.
An inductive definition of the set of alternation-free formulas is given in \cite{Marti2021}. 

A bound variable $x$ in $\eta x. \psi$ is called \emph{guarded}, if it is in the scope of a modality. A formula $\phi$ is called \emph{guarded}, if all variables $x$ in subformulas of the form $\eta x. \psi$ are guarded. 
It is well known that every $\mu$-calculus formula is equivalent to a guarded one \cite{Demri2016}, in particular, every alternation-free $\mu$-calculus formula is equivalent to a guarded alternation-free $\mu$-calculus formula.
Note that the transformation of a formula into a guarded equivalent formula can cause an exponential blowup in the size of the formula \cite{Bruse2015}.

We use standard terminology and notation for the binding of variables by 
the fixpoint operators and for substitutions, and make sure only to 
apply substitution in situations where no variable capture will occur. 
An important use of the substitution operation concerns the \emph{unfolding}
$\chi[\xi/x]$ of a fixpoint formula $\xi = \eta x . \chi$.

We write $\sub$ for the subformula relation on formulas. Given two formulas $\phi,\psi \in \muML$ we write $\phi \tracestep \psi$ if 
$\psi$ is either a direct Boolean or modal subformula of $\phi$, or else 
$\phi$ is a fixpoint formula and $\psi$ is its unfolding. We let $\trace$ be the reflexive and transitive closure of $\tracestep$. We write $\phi \equic \psi$ if $\phi \trace \psi$ and $\psi \trace \phi$.

\noindent 
The \emph{closure} $\Clos(\Phi) \subseteq \muML$ of $\Phi \subseteq \muML$
is the least superset of $\Phi$ that is closed under the relation $\trace$.
It is well known that $\Clos(\Phi)$ is finite if and only if $\Phi$ is finite. Moreover, all formulas in $\Phi$ are closed and guarded, if and only if all formulas in $\Clos(\Phi)$ are closed and guarded. 
A \emph{trace} is a sequence $(\phi_{n})_{n<\kappa}$, with $\kappa \leq
\omega$, such that $\phi_{n} \tracestep \phi_{n+1}$, for all $n
+ 1 < \kappa$.
\begin{definition}
	Let $\rank$ be the minimal function from the set of formulas to $\Nat$, such that
	\begin{enumerate}
		\item $\rank(p) =\rank(\mybar{p}) = 1$,
		\item $\rank(\phi) = \rank(\psi)$ if $\phi \equic \psi$ and
		\item  $\rank(\phi) > \rank(\psi)$ if $\phi \trace \psi$ and $\psi \not\trace \phi$.\qedhere
	\end{enumerate}%
\end{definition}
\noindent%
We extend the operation $p \mapsto \mybar{p}$ to a negation on all formulas $\phi \in \muML$ by the following induction:
\[\begin{array}{lllclllclllclll}
	\mybar{\phi \land \psi} & \isdef & \mybar{\phi} \lor \mybar{\psi} 
	&& \mybar{\mu x. \phi} & \isdef & \nu x. \mybar{\phi}
	&& \mybar{\lbox \phi} & \isdef & \ldia \mybar{\phi}
	&& \mybar{x} & \isdef & x 
	\\ 
	\mybar{\phi \lor \psi} & \isdef & \mybar{\phi} \land \mybar{\psi}
	&& \mybar{\nu x. \phi} & \isdef & \mu x. \mybar{\phi}
	&& \mybar{\ldia \phi} & \isdef & \lbox \mybar{\phi}  
	&& \mybar{\mybar{p}} & \isdef & p  	
\end{array}\]
Note that $\mybar{\mybar{\phi}} = \phi$ for every $\phi \in \muML$.

The semantics of the modal $\mu$-calculus will only play an indirect role and will therefore be omitted. An introduction can be found in \cite{Bradfield2007}.

\begin{definition}
	We call a formula $\phi$ \emph{magenta}, if there is $\mu x. \psi$ such that $\phi \equic \mu x.\psi$, and we call $\phi$ \emph{navy}, if there is $\nu x. \psi$ such that $\phi \equic \nu x.\psi$. 
\end{definition}

The following proposition summarizes key properties of $\AFMC$ that we will need later. 
\begin{proposition}\label{prop.aFree}
	Let $\phi$ be an alternation-free formula. Then 
	\begin{enumerate}
			\item its negation $\mybar{\phi}$ is alternation-free,
			\item every subformula of $\phi$ is alternation-free,
			\item every formula in $\Clos(\phi)$ is alternation-free, \label{prop.aFreeClos}
			\item there is a  guarded formula $\phi'$ which is also alternation-free and logically equivalent to $\phi$,
			\item  $\phi$ is not both magenta and navy.
	\end{enumerate}
\end{proposition}

%
%
\begin{proof}
Items 1--3 are immediate and item 4 follows from the standard translation of a $\mu$-calculus formula to a guarded one (see for instance \cite[Theorem 8.3.23.]{Demri2016}). 
We prove item 5.
	Towards a contradiction assume that $\xi \in \AFMC$ is both magenta and navy. Then there is a pair of formulas $\eta x. \phi, \mybar{\eta} y. \psi \in \Clos(\xi)$ with $\eta x. \phi \trace  \mybar{\eta} y. \psi$ and $\mybar{\eta} y. \psi \trace \eta x. \phi$.
	
	For any trace $\chi_1...\chi_n$ there is $i \in \{1,...,n\}$, such that $\chi_i$ is a subformula of every formula on the trace and such that $\chi_i$ is a fixpoint-formula if $\chi_n$ is one. This can be shown by induction on the length of the trace.
	
	 Therefore we can find such formulas $\eta x. \phi, \mybar{\eta} y. \psi$ that also satisfy $\eta x. \phi \sub \mybar{\eta} y. \psi$. Moreover we may assume that all fixpoint formulas occurring on the trace $\tau: \eta x. \phi ... \mybar{\eta} y. \psi$ are $\eta$-formulas (apart from $\mybar{\eta} y. \psi$). Otherwise we could replace $\mybar{\eta} y. \psi$ with such an $\mybar{\eta}$-formula.
	
	Due to Proposition \ref{prop.aFree}.\ref{prop.aFreeClos} all formulas on $\tau$ are alternation-free. Thus there is no free occurrence of $z$ in $\psi$ for any $\eta$-formula $\eta z. \delta$ on $\tau$. As all fixpoint formulas occurring on $\tau$ are $\eta$-formulas it follows inductively that $\mybar{\eta} y. \psi$ is a subformula of every formula on $\tau$, in particular $\mybar{\eta} y. \psi \sub \eta x. \phi$. Yet this contradicts $\eta x. \phi \sub \mybar{\eta} y. \psi$.
\end{proof}

In this paper we will simply write \emph{formulas} for guarded, closed and alternation-free formulas.

%% file: app.mathPrelim.tex

\subsection{Multisets}
Sequents in our proof system will consist of multisets of formulas. We define multisets slightly differently than usual, but all intuitions about multisets remain the same.

Let $X$ be a set, a \emph{multiset over $X$} is a set of indexed elements, meaning that it consists of pairs $(x,n)$, where $x \in X$ and $n\geq 1$ such that $(x,n) \in A$ implies $(x,m) \in A$ for all $m \leq n$. We write $\calM_X$ for the set of all finite multisets over $X$. We only mention indices if they are of importance and otherwise denote a multiset $A = \{(x_1,n_1),...,(x_k,n_k)\}$ just by $\{x_1,...,x_k\}$. For simplicity we also sometimes omit the brackets and write $A = x_1,...,x_k$.

If $A$ is a multiset, we define the \emph{multiplicity} $\multi_A(x)$ of $x$ in $A$ as the maximal $n$ such that $(x,n) \in A$ and define it to be $0$ if no such $n$ exists. This definition agrees with the number of occurrences of $x$ in $A$. 
For example, we denote the multiset $A= \{(x,1),(x,2),(y,1)\}$ by $\{x,x,y\}$ and have that $\multi_A(x) = 2$, $\multi_A(y) = 1$ and $\multi_A(z) = 0$ for all $z \not\in \{x,y\}$. 

For two multisets $A$ and $B$ over a set $X$ we write $A \subseteq B$ if $\multi_A(x) \leq \multi_B(x)$ for all $x \in X$.
We define $A_{\Set} \isdef \{a \| (a,n) \in A \text{ for some }n\}$ for the \emph{underlying set of $A$} and write $A \seteq B$ if $A_{\Set} = B_{\Set}$.
We write $A,B$ for the union of the multisets $A$ and $B$, defined as expected.

The reason for this choice of definition lies in the need to talk about specific elements $x$ of a multiset $A$. As $A$ is a set of indexed elements we can then choose $(x,n)$ for some specific $n$.

\medskip
Let $(X,<_X)$ be a well-ordered set and $\calM_X$ be the set of all finite multisets over $X$. We define the \emph{Dershowitz-Manna ordering} $\lDM$ on $\calM_X$ as follows: Let $A,B$ be in $\calM_X$, then $A \lDM B$ iff there exists $x \in X$ such that
\begin{enumerate}
	\item $\multi_A(x) < \multi_B(x)$ and
	\item  $\multi_A(y) = \multi_B(y)$ for all $y >_X x$.
\end{enumerate}
The Dershowitz-Manna ordering was introduced in \cite{Dershowitz1979}. That work also establishes that the ordering is well-founded.

\begin{proposition}\label{prop.DerMannaWellfounded}
	Let $(X,<_X)$ be a well-ordered set and $\calM_X$ be the set of all finite multisets over $X$. Then $(\calM_X,\lDM)$ is a well-order.
\end{proposition}

\subsection{Well-quasi-orders}\label{sec.sub.wqo}

 We briefly introduce well-quasi-orders to the extent used in this paper, for a more extensive treatment we refer to \cite{Schuster2020} and for examples of applications of well-quasi-orders to proof-theory we refer to \cite{Galatos2025}. We will use well-quasi-orders in Section \ref{sec.contractions} to show that the procedure of eliminating contractions terminates.

Let $\calQ = (Q,\leq_Q)$ be a \emph{quasi-order}, meaning that $\leq_Q$ is a reflexive and transitive relation on a non-empty set $Q$. 
Let $\kappa \leq \omega$, a \emph{bad sequence} of length $\kappa$ over a quasi-order $\calQ$ is a sequence $(q_n)_{n < \kappa}$ such that for all $m < n$, $q_m \not \leq_Q q_n$.
A quasi-order $\calQ$ is a \emph{well-quasi-order}, in short \emph{wqo}, if every bad sequence over $\calQ$ is finite.

Let $(\Nat^k, \leq)$ be the set of $k$-tuples of natural numbers ordered with the natural product order: $(m_1,...,m_k) \leq (n_1,...,n_k) :\Leftrightarrow m_i \leq n_i$ for all $i =1,...,k$. Clearly, $(\Nat^k, \leq)$ is a quasi-order. Dickson's Lemma \cite{Dickson1913} states that it is in fact a wqo.

\begin{lemma}[Dickson's Lemma]\label{lem.wellQuasiOrderMultiset}
	 For every $k \in \Nat$,  $(\Nat^k,\leq)$ is a well-quasi-order.
\end{lemma}
\begin{proof}
	By induction on $k$. The base case is trivial. For the inductive step assume that $(\Nat^k,\leq)$ is a wqo, we need to show that $(\Nat^{k+1}, \leq)$ is a wqo. 
	Towards a contradiction assume that $(a_n)_{n \in \omega}$ is an infinite sequence of $(k+1)$-ary tuples such that $a_m \not\leq a_n$ for all $m < n$. For $n \in \omega$ let $a_n = (a_n^1,...,a_n^k, a_n^{k+1})$ and define $b_n \isdef (a_n^1,...,a_n^{k})$ and $s_n \isdef a_n^{k+1}$.  Because $(s_n)_{n\in \omega}$ is an infinite sequence of natural numbers there are increasing indices $(n(i))_{i \in \omega}$ such that $s_{n(i)} \leq s_{n(j)}$ for all $i < j$. But then $(b_{n(i)})_{i \in \omega}$ is an infinite sequence of $k$-tuples such that $b_{n(i)} \not\leq b_{n(j)}$ for all $i < j$. This contradicts the fact that $(\Nat^{k}, \leq)$ is a wqo.
\end{proof}

We use well-quasi-orders as a tool for showing termination of our cut-elimination algorithm. This resembles their use in showing termination of proof search algorithms in substructural logics \cite{Galatos2025}. For establishing the complexity of our algorithm we do not only need the non-existence of infinite bad sequences, but moreover a bound on the length of finite bad sequences. Such a bound may not always be found for wqos, for example consider the wqo $(\Nat,\leq)$, where we can easily find bad sequences of arbitrary length. 

We therefore move to the concepts of \emph{normed well-quasi-orders} and \emph{controlled bad sequences.}

\begin{definition}
	A \emph{normed well-quasi-order}, in short \emph{nwqo}, is a triple $\calQ = (Q,\leq_Q,\norm)$, where $(Q,\leq_Q)$ is a wqo and $\norm : Q \to \Nat$ is a \emph{proper norm}, meaning that for every $n \in \Nat$, the set $\{q \in Q \| \norm[q] \leq n\}$ is finite.
 \end{definition}
 
 \begin{definition}
 	A \emph{control function} is a map $f: \Nat \to \Nat$ that is strictly increasing, that is, $f(m) > f(n)$ for all $m>n$.
 	
 	Given an nwqo $\calQ$, a control function $f$ and $t \in \Nat$, an \emph{$(f,t)$-controlled bad sequence over $\calQ$} is a bad sequence $(q_n)_{n < \kappa}$ over $\calQ$ where\footnote{Here $f^n(t) = f(\cdots (f(t)) \cdots )$ stands for the $n$-th iterate of $f$ applied to $t$.} $\norm[q_n] \leq f^n(t)$ for all $n < \kappa$.
 \end{definition}
 
 \begin{lemma}\label{lem.wqoMaximalLength}
 	Given an nwqo $\calQ$, a control function $f$ and $t \in \Nat$, $(f,t)$-controlled bad sequences over $\calQ$ have a maximal length.
 \end{lemma}
 \begin{proof}
 	The idea is to construct a finitely branching tree $T$ of all possible $(f,t)$-controlled bad sequences over $\calQ$ and then use König's Lemma. The root of $T$ will be unlabeled, and at level $1$ we add all elements $q \in Q$ such that $\norm[q]\leq t$. Now let $q_1,...,q_i$ be a path to a node $q_i$ at level $i$, as children of $q_i$ we add all elements $q \in Q$ such that $\norm[q] \leq f^{i+1}(t)$ and such that $q_1,...,q_i,q_{i+1}$ is a bad sequence over $\calQ$. This constructs a tree $T$ of all possible $(f,t)$-controlled bad sequences over $\calQ$. As $\norm$ is a proper norm, $T$ is finitely branching and, because $Q$ is a wqo, it does not have an infinite branch. Therefore König's Lemma yields that $T$ is finite. In particular, $T$ has a maximal depth which corresponds to the maximal length of an $(f,t)$-controlled bad sequences over $\calQ$.
 \end{proof}
 
\begin{definition}
	 Given an nwqo $\calQ$ and a control function $f$ we define the \emph{length function} $\lenf{\calQ}{f} : \Nat \to \Nat$ that maps each $t \in \Nat$ to the maximal length of $(f,t)$-controlled bad sequences over $\calQ$.
\end{definition}
 
 Let $\norm_{\infty}$ be the infinity norm on $(\Nat^k,\leq)$ defined as $\norm[(n_1,...n_k)]_{\infty}  \isdef \max\{n_i \| i = 1,...,k\}$ and define $\sfN^k \isdef (\Nat^k,\leq,\norm_{\infty})$. A thorough investigation of the complexity of $\lenf{\sfN^k}{f}$ can be found in \cite{Figueira2011}. We will not go into more detail as we are not dealing with complexity issues in this paper. Let us note though, that for a primitive recursive $f$ and fixed $k$ the function $\lenf{\sfN^k}{f}$ is primitive recursive as well. If $k$ is added as a part of the input, the function $(k,t) \mapsto \lenf{\sfN^k}{f}(t)$ is not primitive recursive and its growth is comparable to the Ackermann function.

 \medskip
 In this paper we will be working with the following nwqo. Let $X$ be a finite set. We let $\sfM_X \isdef (\calM_X,\subseteq, \norm_{\infty})$ be the nwqo consisting of the set of all multisets over $X$ ordered by inclusion together with the infinity norm $\norm[A]_{\infty} \isdef \max\{\multi_A(x) \| x \in X\}$. Let $X = \{x_1,...,x_k\}$. It is easy to see that $\sfM_X$ is isomorphic to $\sfN^k$: consider the map $A \mapsto (\multi_A(x_1),...,\multi_A(x_k))$. Therefore, $\sfM_X$ is an nwqo due to Lemma \ref{lem.wellQuasiOrderMultiset} and we can use Lemma \ref{lem.wqoMaximalLength} for $\sfM_X$.

%% file: sec.proofsystem.tex
\section{The \Focus system}
\label{sec.proofsystem}

An \emph{annotated formula} is a pair $(\phi,a)$, usually denoted as $\phi^a$, where $\phi$ is a formula and $a \in \{f,u\}$. We call annotated formulas of the form $\phi^f$ \emph{in focus} and of the form $\phi^u$ \emph{out of focus}. We use $a,b$ as variables ranging over $\{f,u\}$.
A finite \emph{multiset} of annotated formulas is called a \emph{sequent}. 
We define the following  operations on sequents $\Sigma$:
\begin{align*}
	\Sigma^u &\isdef \{\phi^u \mid \phi^a \in \Sigma \} \qquad &
	\Sigma^- &\isdef \{\phi \mid \phi^a \in \Sigma \}  \\
	\Sigma^f &\isdef \{\phi^f \mid \phi^a \in \Sigma \} \qquad &
	\ldia\Sigma~~ &\isdef \{\dia \phi^a \mid \phi^a \in \Sigma\}
\end{align*}

In Figure \ref{fig.RulesFocus} the rules of the \Focus system are depicted. The axiom \AxLit and the rules \RuOr,\RuAnd,\RuBox, \RuWeak, \RuContr are standard. In the fixpoint rules \RuMu and \RuNu we prove a fixpoint formula $\eta x. \phi$ by showing its unfolding $\phi\subst{\eta x. \phi}{x}$. Note that $\phi\subst{\eta x. \phi}{x}$ might be syntactically bigger than $\eta x. \phi$ and therefore branches in a derivation might be infinitely long. In the \emph{focus rules} \RuF and \RuU we can put formulas in focus and out of focus, together with \RuMu these are the only formulas changing annotations. The rule \RuDischarge[] marks repeats; every \RuDischarge[] rule is labeled with a unique discharge token taken from a fixed infinite set $\Tokens = \{\dx, \dy, \dz,...\}$. The rule \RuCut is formulated as usual. Notably, the cut formula is always out of focus. It will be the goal of this paper to show that we can eliminate \RuCut rules.

\begin{definition}[Derivation]
	A \emph{\Focus-derivation} $\pi = (T,P,\sfS,\sfR)$ is a quadruple such that
	$(T,P)$ is a, possibly infinite, tree with nodes $T$ and parent relation $P \subseteq T \times T$; 
	$\sfS$ is a function that maps every node $v \in T$ to a sequent $\sfS_v$;
	$\sfR$ is a function that maps every node $v \in T$ to either 
	(i) the name of a rule in Figure \ref{fig.RulesFocus},
	(ii) a discharge token or (iii) an extra value \extraValue (signaling an open assumption) such that 
	\begin{enumerate}
		\item the specifications of the rules in Figure \ref{fig.RulesFocus} are satisfied,
		\item every node labeled with a discharge token or \extraValue is a leaf, and
		\item for every leaf $l$ that is labeled with a discharge token $\dx \in \Tokens$ there is an ancestor $c(l)$ of $l$ that is labeled with \RuDischarge[\dx] and such that $l$ and $c(l)$ are labeled with the same sequent. In this case we call $l$ a \emph{repeat leaf} and $c(l)$ the \emph{companion} node of $l$.
	\end{enumerate}
	A \Focus-derivation of a sequent $\Sigma$ is a \Focus-derivation where the root is labeled with $\Sigma$.
\end{definition}

We will read proof trees ``upwards", so nodes labeled with premises are viewed as children of nodes labeled with the conclusion of a rule. A node $v$ is called an \emph{ancestor}  of a node $u$ if there are nodes $v=v_0$, \dots , $v_n=u$ with $v_{i+1}$ being a child of $v_i$ for $i = 0,...,n-1$ and $n > 0$.


\begin{figure}[tbh]
\begin{mdframed}[align=center]
\begin{minipage}{\textwidth}
	\begin{minipage}{0.2\textwidth}
		\begin{prooftree}
		 \hypo{\phantom{X}}
		 \infer1[\AxLit]{p^a, \atneg{p}^b}
		\end{prooftree}
	\end{minipage}
	\begin{minipage}{0.25\textwidth}
		\begin{prooftree}
		 \hypo{\phi^a,\psi^a,\Sigma}
		 \infer1[\RuOr]{(\phi \lor \psi)^a,\Sigma}
		\end{prooftree}
	\end{minipage}
	\begin{minipage}{0.25\textwidth}
		\begin{prooftree}
			\hypo{\phi[\mu x . \phi / x]^u, \Sigma}
			\infer1[\RuMu]{\mu x . \phi^a, \Sigma}
		\end{prooftree}
	\end{minipage}
	\begin{minipage}{0.20\textwidth}
		\begin{prooftree}
			\hypo{\Sigma}
			\infer1[\RuWeak]{\phi^a, \Sigma}
		\end{prooftree}
	\end{minipage}
	
\end{minipage}

\bigskip

\begin{minipage}{\textwidth}
	\begin{minipage}{0.2\textwidth}
		\begin{prooftree}
			\hypo{\phi^a,\Sigma}
			\infer1[\RuBox]{\Box \phi^a, \dia \Sigma}
		\end{prooftree}
	\end{minipage}
	\begin{minipage}{0.25\textwidth}
		\begin{prooftree}
			\hypo{\phi^a, \Sigma}
			\hypo{\psi^a,\Sigma}
			\infer2[\RuAnd]{(\phi \land \psi)^a,\Sigma}
		\end{prooftree}
	\end{minipage}
	\begin{minipage}{0.25\textwidth}
		\begin{prooftree}
		 \hypo{\phi[\nu x . \phi / x]^a, \Sigma}
		 \infer1[\RuNu]{\nu x . \phi^a, \Sigma}
		\end{prooftree}
	\end{minipage}
	\begin{minipage}{0.2\textwidth}
		\begin{prooftree}
			\hypo{\phi^a,\phi^a,\Sigma}
			\infer1[\RuContr]{\phi^a,\Sigma}
		\end{prooftree}
	\end{minipage}
\end{minipage}

\bigskip

\begin{minipage}{\textwidth}
	\begin{minipage}{0.01\textwidth}
		\phantom{x}
	\end{minipage}
	\begin{minipage}{0.19\textwidth}
		\begin{prooftree}
			\hypo{[\Sigma]^\dx}
			\infer[no rule]1{\vdots}
			\infer[no rule]1{\Sigma}
			\infer1[\RuDischarge[\dx]]{\Sigma}
		\end{prooftree}
	\end{minipage}
	\begin{minipage}{0.25\textwidth}
		\begin{prooftree}
			\hypo{\Delta^f,\Sigma}
			\infer1[\RuF]{\Delta^u,\Sigma}
		\end{prooftree}
	\end{minipage}
	\begin{minipage}{0.25\textwidth}
		\begin{prooftree}
			\hypo{\Delta^u,\Sigma}
			\infer1[\RuU]{\Delta^f,\Sigma}
		\end{prooftree}
	\end{minipage}	
	\begin{minipage}{0.2\textwidth}
		\begin{prooftree}
			\hypo{\Sigma_0, \phi^u}
			\hypo{\mybar{\phi}^u, \Sigma_1}
			\infer2[\RuCut]{\Sigma_0,\Sigma_1}
		\end{prooftree}
	\end{minipage}
\end{minipage}
\end{mdframed}

\caption{Proof rules of the \Focus system}
\label{fig.RulesFocus}
\end{figure}

\newpage
Let $\pi = (T,P,\sfS,\sfR)$ be a derivation. We will be working with the following two trees associated to $\pi$.
\begin{enumerate}
	\item[(i)] The usual proof tree $\calT_\pi \isdef (T,P)$.
	\item[(ii)] The \emph{proof tree with back edges} $\cyclicPT \isdef (T,P^C)$, where $P^C$ is the parent relation plus back-edges for each repeat leaf, meaning that
	\( P^C \isdef P \cup \{(l,c(l))\mid l \text{ is a repeat leaf}\}.\) 
\end{enumerate}
A \emph{path} in a \Focus-derivation $\pi = (T,P,\sfS,\sfR)$ is a path in $\cyclicPT$.

\begin{definition}[Successful path]
	A path $\tau$ in a \Focus-derivation is called \emph{successful} if
	\begin{enumerate}
		\item every sequent on $\tau$ has a formula in focus,
		\item there is no application of \RuF on $\tau$ and
		\item  $\tau$ passes through an application of \RuBox.
	\end{enumerate}
	Let $l$ be a repeat leaf in a \Focus-derivation  $\pi = (T,P,\sfS,\sfR)$ with companion  $c(l)$, and let $\tau_l$ denote the \emph{repeat path} of $l$  in $\calT_\pi$ from $c(l)$ to $l$. We call $l$ a \emph{discharged leaf} if the path $\tau_l$ is successful.		
	A leaf is called \emph{closed} if it is either a discharged leaf or labeled with an axiom and is called \emph{open} otherwise.
\end{definition}

\begin{definition}[Proof]
	A \emph{\Focus-proof} $\pi$ is a finite \Focus-derivation, where every leaf is closed.
\end{definition}

\noindent
Note that we make some adaptations to the \Focus system compared to the presentation in \cite{Marti2021}:
	\begin{enumerate}
		\item Sequents are \emph{multisets} of formulas, compared to sets in \cite{Marti2021}. Therefore we add the contraction rule \RuContr.
		\item We change the focus rules \RuF and \RuU to apply to multisets of formulas compared to single formulas,
		\item On successful paths we allow \RuU rules.
	\end{enumerate}
	It can easily be seen that the adaptations 1 and 2 are harmless. Proposition \ref{prop.FocusAdaptations} deals with the third adaptation and thus shows the equivalence of the two presentations.  Consequently we obtain Soundness and Completeness from \cite{Marti2021}.

\begin{proposition}\label{prop.FocusAdaptations}
	Let $\pi$ be a \Focus-proof. Then we can obtain a \Focus-proof $\pi'$ from $\pi$ without applications of \RuU rules on repeat paths by only adding and deleting focus rules. 
\end{proposition}
\begin{proof}
	Let $\pi$ be a \Focus-proof, where \RuU rules might be applied on repeat paths $\tau_v$ for discharged leaves $v$.
	Let $\pi'$ be the \Focus-proof, where all \RuU rules on repeat paths $\tau_v$  are deleted and focus annotations are inductively propagated upwards in $\cyclicPT$. As we are only putting formulas in focus, i.e. changing formulas $\phi^u$ to $\phi^f$, this terminates. It remains to adjust nodes, that do not lie on a repeat path $\tau_v$ for any $v$, by putting formulas in focus and adjusting \RuF and \RuU rules. This results in a \Focus-proof of the same sequent without \RuU rules on repeat paths.	
\end{proof}

\begin{theorem}[Soundness and Completeness, \cite{Marti2021}]
	There is a \Focus-proof of $\Sigma$ iff $\Lor \Sigma^-$ is valid.
\end{theorem}

\noindent
We end this section with a few definitions that will be of importance later on.

The \emph{size} of a \Focus-proof $\pi$ is the number of nodes in $\pi$.
For a path $\tau$ in $\pi$ we call a path $\tau'$ a \emph{subpath} of $\tau$, if $\tau$ can be written as $\tau = \sigma_0 \tau' \sigma_1$ for some paths $\sigma_0,\sigma_1$ in $\pi$. 

Let $\Gamma$ and $\Sigma$ be sequents. We write $\Gamma \seteq \Sigma$ if the underlying sets of the multisets $\Gamma$ and $\Sigma$ coincide. Importantly, if $\Gamma \seteq \Sigma$, then we may apply weakening and contraction rules to derive $\Gamma$ from $\Sigma$:
	\begin{align*}
		\begin{prooftree}
			\hypo{\Sigma}
			\infer1[\RuWeak,\RuContr]{\Gamma}
		\end{prooftree}
	\end{align*}

\begin{definition}
	The \emph{rank} of a cut with cut formula $\phi$ is $\rank(\phi)$. The \emph{cut-rank} of a \Focus-derivation $\pi$ is the maximal rank of a cut in $\pi$ and is $0$ if there is no cut in $\pi$.
\end{definition}

%% file: sec.cutStrat.tex
\section{Cut elimination strategy}\label{sec.ceStrategy}

We present a cut elimination procedure for the \Focus proof system. Our approach builds on the strategy developed for the \GKe proof system for modal logic with the eventually operator presented in \cite{Afshari2024}. The method is based on reductive cut elimination adjusted to cyclic proofs, where the cut reductions we employ are the expected ones and can be found in Appendix \ref{app.CutReductions}. We start with an informal explanation of the cut-elimination strategy.

\subsection{Main ideas}
One way to prove cut elimination for finitary proofs is by first proving \emph{cut admissibility}, in other words eliminating a cut at the root of a proof. In the context of cyclic proofs the notion of cut admissibility has to be extended, such that we first eliminate cuts that are \emph{in the root-cluster} -- those nodes from which there is a path to the root in the proof tree with back edges. If the root-cluster only consists of one node we retrieve the usual notion. 

Cut admissibility is shown by an induction on the \emph{rank} of the cut formulas, which is a linearization of the trace relation $\trace$. Importantly $\rank(\phi) > \rank(\psi)$ if $\phi \trace \psi$ and $\psi \not\trace \phi$.

At the core of our strategy is the need to isolate the applications of \RuCut that present the greatest challenges. We thus split applications of \RuCut into two categories: cuts that are located inside a cycle are called \emph{unimportant} and cuts that do not are called \emph{important}. We reduce unimportant cuts to important ones of the same rank and reduce the rank of important cuts. 

As the name suggests, unimportant cuts are easier to deal with. Cut reductions on unimportant cuts do not affect formulas in focus, hence those can be pushed upwards and we can find successful repeats ``below'' the cuts. All remaining cuts will be important and of the same rank.

The treatment of important cuts is more complicated, as descendants of the cut formula might be in focus. Pushing up those cuts might put formulas out of focus and consequently undermine successful paths. In order to still find successful repeats we use a property of $\AFMC$: exactly one of $\phi$ or $\mybar{\phi}$ is not navy for any $\AFMC$-formula $\phi$. Recall that $\phi$ is navy if $\phi \equic \nu x. \chi$ for some $\nu$-formula $\nu x. \chi$. Assume that $\phi$ is not navy and consider the following important cut:
\[\begin{prooftree}
	\hypo{\pi_0}
	\infer[no rule]1{\Sigma_0,\phi}
	\hypo{\pi_1}
	\infer[no rule]1{\mybar{\phi},\Sigma_1}
	\infer2[\RuCut]{\Sigma_0,\Sigma_1}
\end{prooftree}\] 
Then this property implies that no descendant of $\phi$ in $\pi_0$ of the same rank is a $\nu$-formula. As we may assume that only navy formulas are in focus, all descendants of $\phi$ in $\pi_0$ of the same rank are out of focus. We carry on by deleting all descendants of $\phi$ of the same rank in $\pi_0$ and all descendants of $\mybar{\phi}$ of the same rank in $\pi_1$ and ``merge'' those two proofs. This process is similar to pushing cuts upwards, unfolding cycles whenever necessary and introducing cuts for descendants of $\phi$ of lower rank. In the resulting proof $\rho$ we can find successful repeats, as all formulas in focus in $\pi_0$ are transferred over and therefore successful paths in $\pi_0$ are projected to successful paths in $\rho$.

The main difficulty compared to the system \GKe for modal logic with the eventually operator \cite{Afshari2024} are occurrences of conjunctions and disjunctions in the scope of fixpoint operators. Applying cut reductions leads to multiple cut formulas in sequents and multiple sequents connected by cuts. To deal with these situations we employ a multi-cut rule. Because the multi-cut may increase in size one extra difficulty in the termination proof is to show that pushing up multi-cuts is productive. 

As it is often the case, \emph{contractions} pose one of the main problems to cut elimination. For finitary proof systems there are two approaches to deal with contractions. In the first approach a generalisation of the cut rule is added to the system -- the \emph{mix rule}. This rule allows to introduce the cut-formula multiple times in the premises of its rule and therefore functions as a combination of cut and contraction. All cut rules can then be replaced with mix rules and henceforth all mix rules are eliminated.
In the second approach the contraction rule is first shown to be admissible in the proof system without an explicit contraction rule and then cut rules are eliminated from the system without contractions. 

We take inspiration from both of these approaches. In order to eliminate unimportant cuts we introduce a mix rule. The proof is then partitioned into subproofs not containing modal rules -- on these finitary subproofs we can eliminate the mix rules as for finitary proofs. Before eliminating important cuts  first the subproofs rooted at the premises of the cut rule are preprocessed, such that those subproofs do not contain contractions anymore -- this elimination of contractions is done in Section \ref{sec.contractions}.

In the next subsection we introduce the necessary notions to make the definitions of important and unimportant cuts formal.

\subsection{Important and unimportant cuts}

Let $(G,E)$ be a directed graph. A \emph{strongly connected subgraph} of $G$ is a set of nodes $A \subseteq G$, such that from every node of $A$ there is a path to every other node in $A$. A \emph{cluster} of $G$ is a maximally strongly connected subgraph of $G$. A cluster is called \emph{trivial} if it consists of only one node and \emph{proper} otherwise.

Let $\pi$ be a \Focus-proof. A \emph{cluster} of $\pi$ is a cluster of the proof tree with back edges $\cyclicPT$. Let $\calS_\pi$ be the set of proper clusters of $\pi$. We define a relation $\sccr$ on $\calS_\pi$ as follows: $S_1 \sccr S_2$ if $S_1 \neq S_2$ and there are nodes $v_1 \in S_1, v_2 \in S_2$ such that there is a path from $v_1$ to $v_2$ in $\cyclicPT$. The relation $\sccr$ is irreflexive, transitive and antisymmetric. We write $\depth(S)$ for the length of the longest path in $(\calS_\pi,\sccr)$ starting from the cluster $S$.

For a node $v$ in a proof $\pi$, we define the \emph{depth} of $v$ to be
\[
\depth(v)=\max\{\depth(S)\| S \in \calS_\pi \text{ and there is a path from }v \text{ to some } u \in S\}
\] 
where $\max\emptyset=0$. The \emph{depth} of a proof is defined as the depth of its root.

The \emph{component} of $v$, written $\comp(v)$ is the set of nodes $u \in \pi$, that are reachable from $v$ in $\cyclicPT$ with $\depth(u) = \depth(v)$. Note that $\comp(v)$ does not have to coincide with a cluster in $\pi$, but may contain multiple clusters. The component of the root is called the \emph{root-component} and the cluster of the root is called the \emph{root-cluster}.

\begin{definition}[Important cut]	Let $\sfC$ be an occurrence of a \RuCut rule at a node $v$ in a \Focus-proof $\pi$. 
	We call $\sfC$ \emph{important} if $v$ is in a trivial cluster of $\pi$ and \emph{unimportant} otherwise. 
\end{definition}

Let $v'$ be a child node of $v$ in a \Focus-derivation $\pi$. We call a formula $\phi'$ at $v'$ an \emph{immediate descendant} of $\phi$ at $v$ if either (i) $\phi$ and $\phi'$ are designated formulas in the rule description or (ii) $\phi = \phi'$ and they are not designated.
A formula $\psi$ at a node $u$ is called a \emph{descendant} of a formula $\phi$ at a node $v$ if there is a path $v=v_0$, \dots, $v_n=u$ containing respective formulas $\phi=\phi_0$, \dots, $\phi_n=\psi$ such that $\phi_{i+1}$ at $v_{i+1}$ is an immediate descendant of $\phi_i$ at $v_i$ for $i = 0,...,n-1$.
A descendant $\psi$ at a node $u$ of a formula $\phi$ at a node $v$ is called a \emph{component descendant} of $\phi$ if $u \in \comp(v)$.

The next lemma justifies the definition of unimportant cuts. It implies that pushing unimportant cuts upwards does not alter successful paths.

\begin{lemma}\label{lem.unimporantCompDescendent}
	Let $\sfC$ be an unimportant occurrence of a \RuCut with cut-formula $\phi$. Then all component descendants of $\phi$ are out of focus.
\end{lemma}
\begin{proof}
	If $\sfC$ is unimportant then it is in a proper cluster. The cut-formula $\phi$ is always out of focus and in proper clusters there are no applications of \RuF. Thus, a descendant of $\phi$ can only be put in focus in a trivial cluster, which has to be in a different component. 
\end{proof}

\new{
	\subsection{Unfoldings and subproofs}
	One of our main tools to transform \Focus-proofs is to unfold discharged leaves and consider subproofs; we define those constructions here.
	
\begin{definition}
		Let $\pi$ be a \Focus-proof. We call a leaf $l$ in $\pi$ \emph{outermost}, if $c(l)$ is the root of some proper cluster in $\pi$. The \emph{unfolding} of an outermost leaf $l$ in $\pi$ is the derivation obtained from $\pi$ by replacing $l$ with the maximal subderivation $\pi_{c(l)}$ of $\pi$ rooted at $c(l)$.\footnote{In order to guarantee that \RuDischarge{} rules are labeled with unique discharge tokens, in $\pi_{c(l)}$ discharge tokens may be replaced with fresh discharge tokens not occurring in $\pi$.}
		
		Let the root $r$ of $\pi$ be labeled with $\RuDischarge[\disA]$. The \emph{unfolding of $\pi$}\index{unfolding!of a proof} is obtained from $\pi$ by replacing every discharged leaf labeled with $\disA$ with $\pi$, and removing the node $r$.\footnote{We replace discharge tokens $\disB$ with fresh discharge tokens, whenever a \RuDischarge[\disB] rule is duplicated.} 
\end{definition}
	
	The \emph{strongly connected subtree} $\scs(u)$ of a companion node $u$ in $\pi$ is defined as the maximal strongly connected subgraph $A$ of $\pi$ such that $u$ is the root of $A$.

	\begin{definition}[Subproof]\label{def.generatedSubproof}
		Let $v$ be a node in a \Focus-proof $\pi$. We let $\tilde{\pi}_v$ be the \Focus-proof obtained from $\pi$  by recursively unfolding outermost leaves $l$, where $v \in \scs(c(l))$. The \emph{subproof of $\pi$ rooted at $v$}, written as $\pi_v$, is the maximal subderivation of $\tilde{\pi}_v$ rooted at $v$.
	\end{definition}
	Note that $\tilde{\pi}_v$ and $\pi_v$ are well-defined and that $\pi_v$ is a \Focus-proof. This holds, as for every $v$ the set $\Leaves{v} \isdef \{l \text{ leaf in } \pi \| v \in \scs(c(l))\}$ is finite; one of those leaves $l \in \Leaves{v}$ is outermost; and after unfolding $l$, the size of $\Leaves{v}$ gets reduced.

	\begin{figure}[h]
		\begin{subfigure}{0.23\linewidth}
			\centering
			\begin{tikzpicture}[scale = 0.9]
				\tikzset{dot/.style={circle, fill, inner sep=1.5pt}}
				\node[dot] (rsprout)     at ( 0, 0) {};
				\node[dot] (gsprout)     at ( 1, 1.4) {};
				\node[dot] (rbud)        at ( 2, 2.8) {};
				\node[dot] (gbud)        at ( 0, 2.8) {};
				\node[dot,label=right:$v$] (v) at (0.5,0.7) {};
				\node[dot] (ax1)        at ( -1, 1.4) {};
				
				\path (rsprout) edge (gsprout);
				\path (gsprout) edge (rbud);
				\path (gsprout) edge (gbud);
				\path (rsprout) edge (ax1);
				\path [->,out=-20,in=0,>=latex,line width=0.25mm,dashed,looseness=.8,shorten >=2pt,shorten <=2pt] (rbud) edge (rsprout);
				\path [->,out=200,in=180,>=latex,line width=0.25mm,dashed,looseness=1.3, shorten >=2pt,shorten <=2pt] (gbud) edge (gsprout);
			\end{tikzpicture}
			\caption*{$\pi$}
		\end{subfigure}
		\begin{subfigure}{0.07\textwidth}
			\centering
			\begin{tikzpicture}[scale = 0.9]
				\node (invisible) at (0,0) {};
				\node (arrow) at (0,2.8) {\Large$\rightsquigarrow$};
			\end{tikzpicture}
		\end{subfigure}
		\begin{subfigure}{0.28\textwidth}
			\centering
			\begin{tikzpicture}[scale = 0.9]
				\tikzset{dot/.style={circle, fill, inner sep=1.5pt}}
				\node[dot] (rsprout)     at ( 0, 0) {};
				\node[dot] (gsprout)     at ( 1, 1.4) {};
				\node[dot] (rbud)        at ( 2, 2.8) {};
				\node[dot] (gbud)        at ( 0, 2.8) {};
				\node[dot,label=right:$v$] (v) at (0.5,0.7) {};
				\node[dot] (ax1)        at ( -1, 1.4) {};
				
				\path (rsprout) edge (gsprout);
				\path (gsprout) edge (rbud);
				\path (gsprout) edge (gbud);
				\path (rsprout) edge (ax1);
				\path [->,out=200,in=180,>=latex,line width=0.25mm,dashed,looseness=1.3, shorten >=2pt,shorten <=2pt] (gbud) edge (gsprout);
				
				\node[dot] (rsprout)     at ( 2, 2.8) {};
				\node[dot] (gsprout)     at ( 3, 4.2) {};
				\node[dot] (rbud)        at ( 4, 5.6) {};
				\node[dot] (gbud)        at ( 2, 5.6) {};
				\node[dot] (ax1)        at ( 1, 4.2) {};
				
				\path (rsprout) edge (gsprout);
				\path (gsprout) edge (rbud);
				\path (gsprout) edge (gbud);
				\path (rsprout) edge (ax1);
				\path [->,out=-20,in=0,>=latex,line width=0.25mm,dashed,looseness=.8,shorten >=2pt,shorten <=2pt] (rbud) edge (rsprout);
				\path [->,out=200,in=180,>=latex,line width=0.25mm,dashed,looseness=1.3, shorten >=2pt,shorten <=2pt] (gbud) edge (gsprout);
			\end{tikzpicture}
			\caption*{$\tilde{\pi}_v$}
		\end{subfigure}
		\begin{subfigure}{0.05\textwidth}
			\centering
			\begin{tikzpicture}[scale = 0.9]
				\node (invisible) at (0,0) {};
				\node (arrow) at (0,2.8) {\Large$\rightsquigarrow$};
			\end{tikzpicture}
		\end{subfigure}
		\begin{subfigure}{0.2\textwidth}
			\centering
			\begin{tikzpicture}[scale = 0.9]
				\tikzset{dot/.style={circle, fill, inner sep=1.5pt}}
				\node[dot] (gsprout)     at ( 1, 1.4) {};
				\node[dot] (rbud)        at ( 2, 2.8) {};
				\node[dot] (gbud)        at ( 0, 2.8) {};
				\node[dot,label=right:$v$] (v) at (0.5,0.7) {};
				
				\path (v) edge (rbud);
				\path (gsprout) edge (gbud);
				\path [->,out=200,in=180,>=latex,line width=0.25mm,dashed,looseness=1.3, shorten >=2pt,shorten <=2pt] (gbud) edge (gsprout);
				
				\node[dot] (rsprout)     at ( 2, 2.8) {};
				\node[dot] (gsprout)     at ( 3, 4.2) {};
				\node[dot] (rbud)        at ( 4, 5.6) {};
				\node[dot] (gbud)        at ( 2, 5.6) {};
				\node[dot] (ax1)        at ( 1, 4.2) {};
				
				\path (rsprout) edge (gsprout);
				\path (gsprout) edge (rbud);
				\path (gsprout) edge (gbud);
				\path (rsprout) edge (ax1);
				\path [->,out=-20,in=0,>=latex,line width=0.25mm,dashed,looseness=.8,shorten >=2pt,shorten <=2pt] (rbud) edge (rsprout);
				\path [->,out=200,in=180,>=latex,line width=0.25mm,dashed,looseness=1.3, shorten >=2pt,shorten <=2pt] (gbud) edge (gsprout);
			\end{tikzpicture}
			\caption*{$\pi_v$}
		\end{subfigure}
		\caption{A proof $\pi$ and a subproof $\pi_v$. Note that the depth of $\pi_v$ is bigger than the depth of $\pi$.}
	\end{figure}
}

\subsection{Minimally focused proofs}
In the operations we perform on proof-trees we need a good handle on the shape of the proof-trees we are dealing with. We therefore introduce a normal form on proofs that align proper clusters with sequents that have formulas in focus. 

Any node $v$ in a proper cluster of a \Focus-proof $\pi$ has formulas in focus, as it is on the path $\tau_l$ of a discharged leaf $l$ to its companion. For nodes in a trivial cluster this is not necessarily the case. We can apply \RuF and \RuU rules in a certain way to minimize nodes with formulas in focus. By doing so, nodes with formulas in focus resemble the proper clusters of the proof tree with back edges: any node with formulas in focus is either in a proper cluster or is labeled with \RuU. 

Moreover, we can minimize the number of focused formulas at every node in a cluster. Without loss of generality we may also assume that all formulas in focus are navy and of the same rank, since this can be ensured by only focusing navy formulas of the same rank when a \RuF-rule is applied, and applying a \RuU-rule whenever a focused formula of lower rank is derived.

\begin{definition}[Minimally focused]
	A \Focus-proof is called \emph{minimally focused} if the following conditions are satisfied:
	\begin{enumerate}
		\item if $v$ is labeled with \RuF, then its child is labeled with \RuDischarge[];
		\item  if $\depth(v') < \depth(v)$ for a child $v'$ of $v$, then $v'$ is labeled with \RuU, where all formulas in its premise are out of focus;
		\item in any rule application of \RuF all formulas in $\Delta$ are navy formulas with the same rank; and
		\item for any node $v$ in a proper cluster $S$, where $k$ is the maximal rank of a formula in focus in $S$: if $v$ is labeled with $\Sigma,\phi^f$ where $\rank(\phi)<k$, than $v$ is labeled with \RuU with premise $\Sigma, \phi^u$. These are the only applications of \RuU in proper clusters.\qedhere
	\end{enumerate}
\end{definition}
\begin{lemma}
	Let $\pi$ be a \Focus-proof. Then we can obtain a minimally focused \Focus-proof $\pi'$ from $\pi$ by only adding and deleting focus rules. 
\end{lemma}
\begin{proof}
	Annotations only matter on repeat paths, therefore we may employ focus rules in such a way that all formulas on nodes in trivial clusters are out of focus, hence satisfying conditions 1 and 2. Now assume that there is a proper cluster $S$ that does not satisfy 3 or 4. Let the parent of the root of $S$ be labeled with~
	$$\begin{prooftree}
		\hypo{\Delta^f,\Sigma}
		\infer1[\RuF]{\Delta^u,\Sigma}
	\end{prooftree}.$$
	Let $\Delta_m$ be the subset of $\Delta$ consisting of all navy formulas in $\Delta$ of maximal rank $k$ and let $\Delta_r = \Delta \setminus \Delta_m$. We change the \RuF rule to~ $$\begin{prooftree}
		\hypo{\Delta_m^f,\Delta_r,\Sigma}
		\infer1[\RuF]{\Delta_m^u, \Delta_r,\Sigma}
	\end{prooftree}$$ and propagate the annotations upwards accordingly, where we apply \RuU rules, whenever formulas of rank lower than $k$ are in focus. It remains to show that discharged leaves remain discharged leaves. A formula in focus of rank $k$ can only originate from a formula in focus of the same rank, as there are no applications of \RuF on repeat paths and $k$ is the maximal rank of formulas in focus. Therefore, all navy formulas of maximal rank in focus in the original proof remain in focus in the adapted proof. Thus it holds that all discharged leaves are translated to discharged leaves, meaning that they are still repeat leaves and that on every sequent on the repeat path there is a formula of rank $k$ in focus. Doing so we satisfy conditions 3 and 4.
\end{proof}
\new{
As every proof can be transformed to a minimally focused proof of the same sequent by only rearranging focus rules and changing annotations, we may always implicitly transform \Focus-proofs to ones that are minimally focused. 

\begin{mdframed}[align=center,userdefinedwidth=0.8\textwidth]
	\begin{center}
		We assume that every \Focus-proof is minimally focused.
	\end{center}
	\end{mdframed}
}

\begin{proposition}
	If $\pi$ is minimally focused, then a \RuCut rule is important iff all formulas in the conclusion of the cut are out of focus.
\end{proposition}

%% file: sec.cutImportant.tex
\section{Elimination of important cuts}\label{sec.importantCuts}

In this section we develop the required technical machinery to eliminate important cuts. In particular, we will be able to prove the following key lemma.

\begin{lemma}[Main Lemma]\label{lem.cutsImportant}
	Let $\pi$ be a \Focus-proof of the form
	\begin{align*}
		\begin{prooftree}
			\hypo{\hat{\pi}}
			\infer[no rule]1{\Sigma_0, \phi^u}
			\hypo{\hat{\tau}}
			\infer[no rule]1{\mybar{\phi}^u, \Sigma_1}
			\infer2[\RuCut]{\Sigma_0, \Sigma_1 }
		\end{prooftree}
	\end{align*}	
	where $\hat{\pi}$ and $\hat{\tau}$ are cut-free and contraction-free and $\phi$ is a $\mu$-formula. Then we can construct a \Focus-proof $\pi'$ of $\Sigma_0,\Sigma _1$ with cut-rank $< \rank(\phi)$.
\end{lemma}
\noindent
We will obtain the proof of Lemma \ref{lem.cutsImportant} by the following approach:
	\begin{enumerate}
		\item In Subsection \ref{subsec.traversedProofs} we introduce \emph{traversed proofs}; these will be the intermediate objects in the elimination of important cuts.
		\item We proceed with defining a traversed proof $\rho_I$ from $\pi$ in Definition \ref{def.traversedInitial}.
		\item Then we define a construction transforming traversed proofs, that stops if a proof of lower cut-rank is obtained. [Definition \ref{def.traversedConstruction}]
		\item Finally, in Subsection \ref{subsec.importantTermination} we prove that this construction applied to $\rho_I$ terminates, meaning that it produces a \Focus-proof $\pi'$ of cut-rank $<\rank(\phi)$. 
	\end{enumerate}

\subsection{Traversed proofs}\label{subsec.traversedProofs}

We will utilize a \emph{multicut rule} -- a derivable generalisation of the ordinary cut rule -- to avoid the nuisance of cut reductions with a cut rule, which might lead to commuting cut rules without any progress. This is a common way to deal with this technicality, see for instance \cite{Fortier2013}. The multicut compresses several cut rules to one rule with multiple premises. For example the following proof would be expressed by a multicut as follows:

\begin{align*}
	\begin{minipage}{0.35\textwidth}
		\begin{prooftree}
			\hypo{\Sigma, \phi}
			\hypo{\mybar{\phi},\psi,\Sigma,}
			\infer2[\RuCut]{\psi,\Sigma}
			\hypo{\mybar{\psi}, \Sigma}
			\infer2[\RuCut]{\Sigma}
		\end{prooftree}
	\end{minipage}
	\longrightarrow \qquad
	\begin{minipage}{0.4\textwidth}
		\begin{prooftree}
			\hypo{\Sigma, \phi}
			\hypo{\mybar{\phi},\psi,\Sigma,}
			\hypo{\mybar{\psi}, \Sigma}
			\infer3[\RuCutMulti]{\Sigma}
		\end{prooftree}
	\end{minipage}
\end{align*}

In the multicut rule we employ, annotations might vary between premises of the rule and its conclusion. The distinctness in our setting is that we have to keep track of the already constructed part of the proof below the multicuts -- there we find successful repeats. We thus define \emph{traversed proofs}: proofs that are traversed by multicuts, meaning that on every branch of the proof there is at most one multicut. These will be our central technical objects in the elimination of important cuts.

A \emph{colored graph} is a graph $(G,E)$, where every edge $e \in E$ is labeled with a color $c$. We write $E_c(v,w)$ if there is an edge between $v$ and $w$ labeled with $c$.
\begin{definition}[Multicut]\label{def:multi}
	A \emph{multicut} $\calM = (\Pi, \Psi, \Tau, \sfG)$ is a quadruple such that $\Pi = \pi_1,...,\pi_m$ and $\Tau = \tau_1,...,\tau_n$ are multisets of \Focus-proofs; $\Psi= \psi_1,...,\psi_k$ is a multiset of formulas; and $\sfG$ is an undirected colored graph with nodes $\Pi \cup \Tau$ and edges colored by formulas in $\Psi$; where $\Psi$ has two decompositions in multisets $\Psi = \Psi_1,...,\Psi_m$ and $\Psi = \Phi_1,...,\Phi_n$ such that the following conditions are satisfied:
	\begin{enumerate}
		\item $\pi_i$ is a proof of $\Gamma_i,\Psi_i^u$ for $i = 1,...,m$, 
		\item $\tau_j$ is a proof of $\Delta_j, \mybar{\Phi_j'}$, where $(\Phi_j')^- = \Phi_j$ for $j = 1,...,n$ and
		\item the graph $\calG$ is \emph{connected}, \emph{acyclic} and
		each $\psi \in \Psi$ is associated a unique edge $E_\psi(\pi_i,\tau_j)$ for some $i = 1,...,m$ and $j= 1,...,n$ such that $\psi \in \Psi_i$ and $\psi \in \Phi_j$.
	\end{enumerate}
	The sequent $\Gamma_1,...,\Gamma_m,\Delta_1,...,\Delta_n$ is called the \emph{conclusion of} $\calM$.
\end{definition}

We call $\sfG$ the \emph{cut-connection graph of $\calM$} and call $\pi$ and $\tau$ \emph{cut-connected via $\psi$} if $E_{\psi}(\pi,\tau)$. An edge $E_{\psi}(\pi_i,\tau_j)$ corresponds to a cut with cut-formula $\psi$ and premises $\pi$ and $\tau$.
If no confusion arises we will denote a multicut $\calM = (\Pi, \Psi, \Tau, \sfG)$ by $\merge{\Pi}{\Psi}{\Tau}$ and treat the cut-connection graph $\sfG$ implicitly.
If $m,n$ and $k$ denote the sizes of $\Pi$, $\Tau$ and $\Psi$, respectively, then the cut-connection graph $\calG$ consists of $m+n$ nodes and $k$ edges. As $\calG$ is connected and acyclic we have $m+n = k +1$. 

We can now define a proof-like object built around the multicut's structure, a \emph{formula-traversed proof}. Fix a formula $\phi$. Intuitively a $\phi$-traversed proof is a derivation (not necessarily a proof) with undischarged leaves of a special form: all non-axiomatic leaves in a $\phi$-traversed proof are conclusions of some multicut rule with cut formulas from $\Clos(\phi)$.
\begin{definition}[Traversed proof]\label{def.traversed}
	A \emph{$\phi$-traversed proof} $\rho$ of a sequent $\Sigma$ is a finite derivation of $\Sigma$, where all leaves $v$ are either closed or \emph{traversed leaves}, meaning that they are labeled with a sequent $\sfS_v$ together with a multicut $\calM_v = (\Pi, \Psi, \Tau, \sfG)$ where $\Psi \subseteq \Clos(\phi)$; and, if $\Gamma$ is the conclusion of $\calM_v$, then $\sfS_v^- = \Gamma^-$.

	If $\phi$ is clear from the context we will just write \emph{traversed proof}.
	We define the \emph{depth} of a traversed leaf $v$ as $\depth(v) = \max\{\depth(\pi) \| \pi \in \Pi\}$. Note that we only consider proofs in $\Pi$ in this definition and not the ones in $\Tau$.
	
	\smallskip\noindent
	A traversed leaf is called \emph{tidy} if 
	\begin{enumerate}
		\item $\Psi \neq \emptyset$ and
		\item $\phi \equic \psi$ for all $\psi \in \Psi$.
	\end{enumerate}
	Given the notation of  Definition \ref{def:multi}, it follows that for any tidy traversed leaf $\Psi_i \neq \emptyset$  and
	$\Phi_j \neq \emptyset$  for all indices $i,j$.
	A traversed proof is called \emph{tidy} if all its traversed leaves are tidy. 
\end{definition}

We will denote a traversed leaf $v$ labeled with a sequent $\Sigma$ and a multicut $\calM_v = (\Pi, \Psi, \Tau, \sfG)$ by 
\[
\begin{prooftree}
	\hypo{\calM_v}
	\infer[no rule]1{\Sigma}
\end{prooftree}\]
and, if we do not want to deal with the cut-connection graph explicitly, by
\[
\begin{prooftree}
	\hypo{\merge{\Pi}{\Psi}{\Tau}}
	\infer[no rule]1{\Sigma}
\end{prooftree}.\]

Ignoring the annotations for a moment, the traversed leaves of a traversed proof can also be interpreted as a multicut rule of the form 
\begin{align*}
	\begin{prooftree}
		\hypo{\pi_1}
		\infer[no rule]1{\Gamma_1, \Psi_1}
		\hypo{\cdots}
		\hypo{\pi_m}
		\infer[no rule]1{\Gamma_m, \Psi_m}
		\hypo{\tau_1}
		\infer[no rule]1{\Delta_1,\mybar{\Phi_1}}
		\hypo{\cdots}
		\hypo{\tau_n}
		\infer[no rule]1{\Delta_n,\mybar{\Phi_n}}
		\infer6[\RuCutMulti]{\Gamma_1, ...,\Gamma_m,\Delta_1,...,\Delta_n}
	\end{prooftree}.
\end{align*}
Additionally, in premises of the multicut, formulas might be in focus that are out of focus in its conclusion.
In this sense, every tidy $\phi$-traversed proof $\rho$ corresponds to a \Focus-proof $\pi$, where on every branch of the proof there is at most one multicut of rank $\rank(\phi)$. Hence, transforming a $\phi$-traversed proof to a traversed proof without traversed leaves corresponds to eliminating multicuts of rank $\rank(\phi)$.

Given a multicut $\calM$, we need an operation that removes an edge labeled with $\psi$ from the cut-connection graph. This might be necessary because a cut of lower rank in the proof is applied or one of the cut formulas is weakened. The multicut $\del{\calM}{\pi}{\psi}$ then consists of the remaining nodes connected to $\pi$.
\begin{definition}
	Let $\calM = (\Pi, \Psi, \Tau, \sfG)$ be a multicut, $\pi \in \Pi$, $\psi \in \Psi$, $\tau \in \Tau$ and let $E_{\psi}(\pi,\tau)$ be an edge in $\sfG$. 
	We define $\del{\calM}{\pi}{\psi}$ to be the multicut $(\Pi', \Psi', \Tau', \sfG')$ obtained as follows: Remove $E_{\psi}(\pi,\tau)$ from $\sfG$ and let $\sfG'$ be the subgraph of $\sfG$ of nodes connected to $\pi$. Let $\Pi' \cup \Tau'$ be the multiset of nodes of $\sfG'$ such that $\Pi'\subseteq \Pi$ and $\Tau'\subseteq \Tau$ and let $\Psi'\subseteq \Psi$ be the multiset of colors of edges occurring in $\sfG'$. Note that $\psi$ is in the conclusion of $\del{\calM}{\pi}{\psi}$.
	
	The multicut $\del{\calM}{\tau}{\psi}$ is defined analogously by replacing $\pi$ with $\tau$.	
\end{definition}

\begin{lemma}\label{lem.traversedTidy}
	Let $\rho$ be a $\phi$-traversed proof with cut-rank $< \rank(\phi)$. Then $\rho$ can be transformed to a tidy $\phi$-traversed proof $\rho'$ with cut-rank $< \rank(\phi)$ without introducing extra \RuF rules.
\end{lemma}
\begin{proof}
	Let $v$ be a traversed leaf in $\rho$ labeled with a sequent $\calS$ and a multicut $\calM_v$ that is not tidy. 
	If $\Psi$ is empty, then we replace $v$ with $\pi_1$ if $m = 1$ or by $\tau_1$ if $n=1$.

	If $\Psi = \Psi',\psi$ and $\psi \not\equic \phi$, then let $\pi$ and $\tau$ be proofs that are cut-connected via $\psi$. Let $\calS^l,\psi$ be the conclusion of $\del{\calM}{\pi}{\psi}$ and $\calS^r,\mybar{\psi}$ be the conclusion of $\del{\calM}{\tau}{\psi}$. Then we replace $v$ with
	\begin{align*}
		\begin{prooftree}
			\hypo{\del{\calM}{\pi}{\psi}}
			\infer[no rule]1{\calS^l, \psi}
			\hypo{\del{\calM}{\tau}{\psi}}
			\infer[no rule]1{\calS^r, \mybar{\psi}^u}
			\infer[]2[\RuCut]{\calS}
		\end{prooftree}.
	\end{align*}
	As $\psi \not\equic \phi$, this cut has rank lower than $\rank(\phi)$ and we obtain a $\phi$-traversed proof with cut-rank $< \rank(\phi)$.
\end{proof}

\new{
In the next subsection we will define a construction transforming traversed proofs to traversed proofs without traversed leaves.
In this construction it is necessary to keep track of the dynamics of the proofs occurring in the multicuts. When transforming a traversed proof with only one traversed leaf $v$ labeled with a multicut $\merge{\hat{\pi}}{\phi}{\hat{\tau}}$, we obtain multicuts of the form $\merge{\Pi}{\Psi}{\Tau}$ where we can relate any proof $\pi$ in $\Pi$ to some node in $\hat{\pi}$, called the \emph{origin} of $\pi$; and we can combine these nodes to a path in $\hat{\pi}$, called the \emph{history} of the proof $\pi$. Analogously for proofs in $\Tau$. 

\begin{definition}
	A \emph{multicut $\calM$ with origin $(\hat{\pi},\hat{\tau})$}\index{multicut!with origin} is a triple $(\calM,(\hat{\pi},\hat{\tau}),\oriPath)$ such that $\calM = (\Pi, \Psi, \Tau, \sfG)$ is a multicut; $(\hat{\pi},\hat{\tau})$ is a pair of \Focus-proofs; and the \emph{history map}\index{history map} $\oriPath$\index{0Symbols@\textbf{Symbols}!$\Hist$@$\oriPath$} is a map with domain $\Pi \cup \Tau$ that maps proofs in $\Pi$ to paths in $\hat{\pi}$ and proofs in $\Tau$ to paths in $\hat{\tau}$. 
	The \emph{origin map}\index{origin!map} of such a multicut with origin is defined as\index{0Symbols@\textbf{Symbols}!$\Origin$@$\oriNode$}
	\begin{align*}
		\oriNode: \Pi \cup \Tau &\to \hat{\pi}\cup \hat{\tau},\\
		\pi &\mapsto \last(\oriPath(\pi)).
	\end{align*}
	For such a triple to qualify as a multicut with origin, the origin map $\oriNode$ has to satisfy the following conditions for all $i =1,...,m$ and $j=1,...,n$, where we use notations as in the definition of a multicut, Definition \ref{def:multi}:
	\begin{enumerate}
		\item $\pi_i = \hat{\pi}_{\oriNode(\pi_i)}$ and $\tau_j = \hat{\tau}_{\oriNode(\tau_j)}$,
		\item $\sfS(\oriNode(\pi_i)) = \Gamma_i', \Psi_i^u$, where $\Gamma_i' = \Gamma_i$ or $(\Gamma_i')^u = \Gamma_i$ and 
		\item  $\sfS(\oriNode(\tau_j)) = \Delta_j', {\Phi_j'}$, where $(\Delta_j')^u = \Delta_j$ and $(\Phi_j')^u = \Phi_j$. 
	\end{enumerate}
	
	We define the \emph{origin depth}\index{origin!depth} of a multicut $\calM = (\Pi, \Psi, \Tau, \sfG)$ with a $(\hat{\pi},\hat{\tau})$-history as $\depthT(\calM) \isdef \max\{\depth(\oriNode(\pi)) \| \pi \in \Pi\}$\index{0Symbols@\textbf{Symbols}!$\OriDepth$@$\depthT$}. Note that we only consider proofs in $\Pi$ in this definition and not the ones in $\Tau$.
\end{definition}

Given a node $v$ in a \Focus-proof $\pi$, the subproof $\pi_v$ of $\pi$ rooted at $v$ may have a greater depth than $\pi$; it is only guaranteed that $\depth(\pi) \leq \depth(\pi_v)$. Therefore, $\depth(\oriNode(\pi))$ may differ from $\depth(\pi)$.

\begin{definition}
	A \emph{$\phi$-traversed proof $\rho$ with origin $(\hat{\pi},\hat{\tau})$}\index{traversed proof!with origin} is defined analogously to a $\phi$-traversed proof where traversed leaves are labeled with multicuts with origin $(\hat{\pi},\hat{\tau})$.
	
	We define the \emph{origin depth}\index{leaf!traversed!origin depth} of a traversed leaf $v$ labeled with $\calM_v = (\Pi, \Psi, \Tau, \sfG)$ as $\depthT(v) \isdef \max\{\depth(\oriNode(\pi)) \| \pi \in \Pi\}$.\index{0Symbols@\textbf{Symbols}!$\OriDepth$@$\depthT$} Note that we only consider proofs in $\Pi$ in this definition and not the ones in $\Tau$.
	\end{definition}
}

\subsection{Proof transformations}\label{subsec.constrImportant}

\begin{definition}\label{def.traversedInitial}
Let $\pi$ be a \Focus-proof as given in Lemma \ref{lem.cutsImportant}. We define the \emph{initial traversed proof} $\rho_I$ \new{with origin $(\hat{\pi},\hat{\tau})$ to be the $\phi$-traversed proof of $\Sigma_0,\Sigma _1$ consisting of a traversed leaf labeled with $\Sigma_0,\Sigma _1$ together with $\merge{\hat{\pi}}{\phi}{\hat{\tau}}$, and where $\oriPath(\hat{\pi})$  is the path consisting of the root of $\hat{\pi}$ and $\oriPath(\hat{\tau})$ is the paths consisting of the root of $\hat{\tau}$. We denote the traversed proof $\rho_I$ by}
\begin{align*}
	\begin{prooftree}
		\hypo{\merge{\hat{\pi}}{\phi}{\hat{\tau}}}
		\infer[no rule]1{\Sigma_0,\Sigma_1}
	\end{prooftree}.
\end{align*}
\end{definition}

The high level strategy to transform $\rho_I$ to a traversed proof without traversed leaves is as follows. We start by pushing up traversed leaves and unfold proofs whenever a companion node is reached. This is done similarly as one would push up multicuts. 
We continue pushing up the traversed leaves in the traversed proof until we find successful repeats below traversed leaves. This check will be done whenever a modal rule gets introduced.

In order to guarantee that we find such a successful repeat we have to be very careful about which formulas we put in focus. Let $v$ be a traversed leaf labeled with 
\[
\begin{prooftree}
	\hypo{\merge{\Pi}{\Psi}{\Tau}}
	\infer[no rule]1{\Gamma_1,...,\Gamma_m, \Delta_1,...,\Delta_n}
\end{prooftree}.
\]
We have to decide on which formulas in $\Gamma_1,...,\Gamma_m, \Delta_1,...,\Delta_n$ we keep the annotations as in the proofs in $\Pi$ and $\Tau$. \new{Our strategy is as follows: formulas in $\Delta_1,...,\Delta_n$ will always be out of focus; and formulas in $\Gamma_i$ keep the same annotation as in $\pi_i$ if $\depth(\oriNode(\pi_i)) = \depthT(v)$ and will be out of focus otherwise for $i = 1,...,m$. Recall that $\depthT(v)$ is the maximal depth of the nodes $\oriNode(\pi_1),...,\oriNode(\pi_m)$ in $\hat{\pi}$.}

The reason for this asymmetry stems from the following observation: the formula $\phi$ is a $\mu$-formula, therefore all formulas in $\Psi$ are magenta and all formulas in $\mybar{\Psi}$ are navy. As $\pi$ is minimally focused only navy formulas are in focus. This means that in the proofs $\pi_1,...,\pi_m \in \Pi$ formulas from $\Psi$ are out of focus, whereas in the proofs $\tau_1,...,\tau_n \in \Tau$ formulas from $\mybar{\Psi}$ might be in focus. By deleting the formulas from $\mybar{\Psi}$ in the proofs $\tau_1,..,\tau_n$ we cannot ensure that successful paths are still successful. Deleting formulas from $\Psi$ in the proofs $\pi_1,...,\pi_m$ on the other hand never removes formulas in focus. 

We thus only keep annotations on formulas originating from the proofs $\pi_1,...,\pi_m$. If we keep the annotations from all those proofs this could also lead to trouble -- we also add applications of \RuF potentially destroying the success-condition on paths. \new{We therefore opt to only keep annotations coming from those proofs in $\pi_1,...,\pi_m$, where $\oriNode(\pi_1),...,\oriNode(\pi_m)$ has maximal depth. This guarantees that at some point no \RuF rules are applied anymore. In the case that all formulas become out of focus. This ensures that $\depthT(v)$ got reduced and hence we can employ induction on $\depthT(v)$ in our termination argument.}

In the next definition we will give a formal description of these intuitions.


\begin{definition}\label{def.traversedConstruction}
We define the \emph{\traversedTrans}; it transforms a traversed proof with traversed leaves preserving the cut-rank.

Let $\rho$ be a $\phi$-traversed proof.  We may always assume that $\rho$ is tidy (see Lemma \ref{lem.traversedTidy}). If all leaves are closed we are done. Otherwise consider the leftmost traversed leaf $v$ labeled with 
\begin{align*}
	\begin{prooftree}
		\hypo{\merge{\Pi}{\Psi}{\Tau}}
		\infer[no rule]1{\Gamma_1,...,\Gamma_m, \Delta_1,...,\Delta_n}
	\end{prooftree}.
\end{align*}
We transform $\rho$ by a case distinction on the last applied rules in $\Pi$ and $\Tau$. 

\begin{itemize}
	\item \textbf{\RuBox rule.} If the last applied rule is \RuBox in $\pi_i$  for all $i = 1,...,m$ and in $\tau_j$ for all $j = 1,...,n$, we make the following case distinction.
\begin{itemize}
	\item If there is a node $c$ in $\rho$ that is an ancestor of $v$ such that $\sfS_c =_{Set} \Gamma_1,...,\Gamma_m, \Delta_1,...,\Delta_m$ and such that the  path from $c$ to $v$ is successful, then insert a \RuDischarge[\dx] rule at $c$ and replace  $v$ with
	\[
	\begin{prooftree}
		\hypo{[\sfS_c]^{\dx}}
		\infer[]1[\RuWeak, \RuContr]{\Gamma_1,...,\Gamma_m, \Delta_1,...,\Delta_n}
	\end{prooftree}
	\] with fresh discharge token $\dx$. If there is such an ancestor that is already labeled with $\RuDischarge[\dy]$, then let the new leaf be discharged by $\dy$ and do not insert an extra \RuDischarge[\dx] rule.
	\item  Otherwise we apply a \RuBox rule in $\rho$ and delete the rule in all proofs $\pi_i$ and $\tau_j$. This is always possible: Because $v$ is tidy there is one less cut formula in $\Psi$ than proofs in $\Pi, \Tau$. Every formula $\psi$ in $\Psi$ is modal, thus either $\psi$ or $\mybar{\psi}$ is of the form $\lbox \chi$. Therefore, there is exactly one formula of the form $\lbox \chi$ in $\Gamma_1,...,\Gamma_m,\Delta_1,...,\Delta_n$ and the rule $\RuBox$ is applicable.
	
	\new{
		We define $\oriPath({\pi}_i')$ as the path $\oriPath(\pi_i)$ extended with the child of $\oriNode(\pi_i)$ for $i=1,...,m$ and define $\oriPath({\tau_j'})$ as the path $\oriPath(\tau_j)$ extended with the child of $\oriNode(\tau_j)$ for $j=1,...,n$.
		}
\end{itemize}
\end{itemize}
Otherwise we pick $i \in \{1,...,m\}$ or $j \in \{1,...n\}$ and reduce $\pi_i$ or $\tau_j$. We let $\Pi = \Pi',\pi_i$ and $\Tau = \Tau',\tau_j$.
\begin{itemize}
	\item \textbf{\RuDischarge[] rule in $\Pi$.} If there is an $i$ such that the last applied rule in $\pi_i$ is \RuDischarge[], then $\pi_i$ has the form
\begin{align*}
	\begin{prooftree}
		\hypo{\pi_i'}
		\infer[no rule]1{\Gamma_i, \Psi_i^u}
		\infer1[\RuDischarge]{\Gamma_i, \Psi_i^u}
	\end{prooftree}.
\end{align*}

We unfold $\pi_i$, meaning that we let $\tilde{\pi}_i$ be the proof obtained from $\pi_i'$ by replacing every discharged leaf, that is labeled with $\dx$, with $\pi_i$. \footnote{Discharge tokens $\dy$ are replaced with fresh discharge tokens, whenever a \RuDischarge[\dy] rule is duplicated.}

\new{
	We perform the following check and focus the sequent $\Gamma_i$ if the origin depth of $\pi_i$ is maximal:
\begin{itemize}
	\item if $\depth(\oriNode(\pi_i)) = \depthT(v)$ and $\Gamma_i'\neq \Gamma_i$, then $\Gamma_i$ is unfocused and we replace $v$ with 
	\begin{align*}
		\begin{prooftree}
			\hypo{\merge{\Pi',\tilde{\pi}_i}{\Psi}{\Tau}}
			\infer[no rule]1{\Gamma_1,...,\Gamma_i',...,\Gamma_m, \Delta_1,...,\Delta_n}
			\infer[]1[\RuF]{\Gamma_1,...,\Gamma_i,...,\Gamma_m, \Delta_1,...,\Delta_n}
		\end{prooftree}
	\end{align*}
	\item otherwise we replace $v$ with
	\begin{align*}
		\begin{prooftree}
			\hypo{\merge{\Pi',\tilde{\pi}_i}{\Psi}{\Tau}}
			\infer[no rule]1{\Gamma_1,...,\Gamma_m, \Delta_1,...,\Delta_n}
		\end{prooftree}
	\end{align*}
\end{itemize}

We define $\oriPath(\tilde{\pi}_i)$ as the path $\oriPath(\pi_i)$ extended with the child of $\oriNode(\pi_i)$.
Note that in this case we have $\tilde{\pi}_i = \hat{\pi}_{\oriNode(\tilde{\pi}_i)}$ because the subproof $\hat{\pi}_{\oriNode(\tilde{\pi}_i)}$ of $\hat{\pi}$ rooted at $\oriNode(\tilde{\pi}_i)$ is defined as the unfolding of the subproof $\hat{\pi}_{\oriNode({\pi}_i)}$.
}

\new{
	\item \textbf{\RuDischarge[] rule in $\Tau$.} If there is a $j$ such that the last applied rule \Ru in $\tau_j$ is \RuDischarge, then $\tau_j$ has the form
\begin{align*}
	\begin{prooftree}
		\hypo{\tau_j'}
		\infer[no rule]1{\Delta_j, \mybar{\Phi_j}}
		\infer1[\RuDischarge]{\Delta_j, \mybar{\Phi_j}}
	\end{prooftree}.
\end{align*}
 We let $\tilde{\tau}_j$ be the proof obtained from $\tau_j'$ by replacing every discharged leaf, that is labeled with $\dx$, with $\tau_j$. 
We replace $v$ with 
\begin{align*}
	\begin{prooftree}
		\hypo{\merge{\Pi}{\Psi}{\Tau',\tilde{\tau}_j}}
		\infer[no rule]1{\Gamma_1,...,\Gamma_m, \Delta_1,...,\Delta_n}
	\end{prooftree}.
\end{align*}
We define $\oriPath(\tilde{\tau}_j)$ as the path $\oriPath(\tau_j)$ extended with the child of $\oriNode(\tau_j)$.
}

\item \textbf{\RuF rule in $\Pi$.} If there is an $i$ such that the last applied rule in $\pi_i$ is \RuF, then $\pi_i$ has the form
\begin{align*}
	\begin{prooftree}
		\hypo{\pi_i'}
		\infer[no rule]1{\Gamma_i', \Psi_i^u}
		\infer1[\RuF]{\Gamma_i, \Psi_i^u}
	\end{prooftree}.
\end{align*}
Note that all formulas in $\Psi_i$ are magenta, thus due to Proposition \ref{prop.aFree} no formula in $\Psi_i$ is navy. As $\pi$ is minimally focused it follows that no formula in $\Psi_i$ is put in focus in \RuF. 
We make a case distinction:
\begin{itemize}
	\item If \new{$\depth(\oriNode(\pi_i)) = \depthT(v)$}, then replace $v$ with
	\begin{align*}
		\begin{prooftree}
			\hypo{\merge{\Pi',{\pi}_i'}{\Psi}{\Tau}}
			\infer[no rule]1{\Gamma_1,...,\Gamma_i',...,\Gamma_m, \Delta_1,...,\Delta_n}
			\infer1[\RuF]{\Gamma_1,...,\Gamma_i,...,\Gamma_m, \Delta_1,...,\Delta_n}
		\end{prooftree}.
	\end{align*}
	
	\item Otherwise replace $\pi_i$ with $\pi_i'$ without applying an \RuF rule.
\end{itemize}
\new{In both cases we define $\oriPath({\pi}_i')$ as the path $\oriPath(\pi_i)$ extended with the child of $\oriNode(\pi_i)$.}

\item \textbf{\RuU rule in $\Pi$.} If there is an $i$ such that the last applied rule in $\pi_i$ is \RuU, then $\pi_i$ has the form
\begin{align*}
	\begin{prooftree}
		\hypo{\pi_i'}
		\infer[no rule]1{\Gamma_i^u, \Psi_i^u}
		\infer1[\RuU]{\Gamma_i', \Psi_i^u}
	\end{prooftree}.
\end{align*}
We make a case distinction:
\begin{itemize}
	\item If there are formulas in focus in $\Gamma_i$, then replace $v$ with
	\begin{align*}
		\begin{prooftree}
			\hypo{\merge{\Pi',{\pi}_i'}{\Psi}{\Tau}}
			\infer[no rule]1{\Gamma_1,...,\Gamma_i^u,...,\Gamma_m, \Delta_1,...,\Delta_n}
			\infer1[\RuU]{\Gamma_1,...,\Gamma_i,...,\Gamma_m, \Delta_1,...,\Delta_n}
		\end{prooftree},
	\end{align*}
	\item Otherwise replace $\pi_i$ with $\pi_i'$ without applying an \RuU rule.
\end{itemize}
\new{In both cases we define $\oriPath({\pi}_i')$ as the path $\oriPath(\pi_i)$ extended with the child of $\oriNode(\pi_i)$.}

\item \textbf{\RuF rule or \RuU rule in $\Tau$.} If there is a $j$ such that the last applied rule \Ru in $\tau_j$ is \RuF or \RuU, then $\tau_j$ has the form
\begin{align*}
	\begin{prooftree}
		\hypo{\tau_j'}
		\infer[no rule]1{\Delta_j', \mybar{\Phi_j}^b}
		\infer1[\Ru]{\Delta_j, \mybar{\Phi_j}^a}
	\end{prooftree}.
\end{align*}
	\new{We replace $\tau_j$ with $\tau_j'$, and define $\oriPath({\tau_j'})$ as the path $\oriPath(\tau_j)$ extended with the child of $\oriNode(\tau_j)$.}

\item \textbf{Other cases.} In the rest of the cases we push the traversed leaf upwards. Those transformations resemble the expected cut reductions for the multicut rule. As they are standard we leave those cases out here -- they can be found in Appendix \ref{app.constrImportant}. Note that we assumed that $\pi$ does not contain contractions and therefore proofs in $\Pi$ or $\Tau$ do not contain contractions either.\qedhere
\end{itemize}
\end{definition}

\begin{remark}
	It may seem that the construction is formulated in a non-deterministic way, yet this is only superficially so. All choices can be made canonical, depending on an arbitrary but fixed order on proof rules and proofs in $\Pi \cup \Tau$. For example, we could give priority to cases where a formula $\psi \in \Psi$ is principal on both sides and take an arbitrary order on $\Pi \cup \Tau$, where proofs in $\Pi$ are of higher priority than proofs in $\Tau$. Importantly, the particular choice of orders does not matter in the termination proof.
\end{remark}

\subsection{Proof of termination}\label{subsec.importantTermination}

We prove that the \traversedTrans given in Subsection \ref{subsec.constrImportant} yields the desired proof. First we show that the transformation only terminates if a traversed proof without traversed leaves is reached. In Lemma \ref{lem.traversedTerminates} we then show that the algorithm terminates when applied to $\rho_I$. 

\begin{lemma}\label{lem.traversedApplicable}
	If $v$ is a tidy traversed leaf in a traversed proof $\rho$, then one of the cases in the case distinction in Definition \ref{def.traversedConstruction} is applicable.
\end{lemma}
\begin{proof}
	Let $v$ be labeled with $\merge{\Pi}{\Psi}{\Tau}$. If there is a proof in $\Pi \cup \Tau$, where the last applied rule is different from a rule with principal formula in $\Psi$ and different from \RuBox, then we can transform that proof. 
	Otherwise for all $i = 1,...,m$ the last applied rule in $\pi_i$ is either \RuBox or a rule with principal formula in $\Psi_i$ and analogously for all  $j = 1,...,n$ the last applied rule in $\tau_j$  is either \RuBox or a rule with principal formula in $\mybar{\Phi_j}$.
	If the last applied rule in all those proofs is \RuBox we are in the first case of Definition \ref{def.traversedConstruction}. 
	Otherwise let $\Psi'$ be the non-empty subset of $\Psi$ consisting of all non-modal formulas in $\Psi$. Let $\Pi' \subseteq \Pi$ and $\Tau' \subseteq \Tau$ be the respective subset of proofs of $\Pi$ and $\Tau$, where the last applied rule is different from \RuBox. Let $\sfG'$ be the subgraph of the cut-connection graph $\sfG$ with nodes $\Pi'\cup \Tau'$ and edges labeled with formulas in $\Psi'$. Then $\sfG'$ is non-empty and acyclic. Moreover, we may assume that $\sfG'$ is connected, otherwise continue with a maximally connected subgraph of $\sfG'$. 
	Let $k'= |\Psi'|$, $m'= |\Pi'|$ and $n' = |\Tau'|$, then $m'+ n' = k' + 1$. At every node in $\Pi'\cup \Tau'$ the principal formula of the last applied rule in the proof is in $\Psi'$ or in $\mybar{\Psi'}$. 
	As $k' < m' + n'$ there is $\psi \in \Psi'$ and an edge $E_\psi(\pi_i,\tau_j)$ in $\calG_v'$  such that $\psi$ is principal in the last applied rule in $\pi_i$ and $\mybar{\psi}$ is principal in the last applied rule in $\tau$.
\end{proof}

\begin{lemma}\label{lem.traversedTerminates}
	The \traversedTrans given in Definition \ref{def.traversedConstruction} applied to the initial traversed proof $\rho_I$ terminates and yields a \Focus-proof $\rho_T$.
\end{lemma}
\begin{proof}	
	\new{
	In this proof we will simply write traversed proofs for traversed proofs with origin $(\hat{\pi},\hat{\tau})$ and treat the $\oriPath$ and $\oriNode$ maps implicitly.

	Let $\rho_k$ and $\rho_l$ be traversed proofs. We write $\rho_k \leqT \rho_l$ if $\rho_l$  can be obtained from $\rho_k$ by the traversed leaf reduction algorithm. It holds that $\leqT$ is a partial order.
	Moreover, if $\rho_k \leqT \rho_l$, then $\rho_k$ is a sub-traversed proof of $\rho_l$, in the sense that $\rho_l$ can be obtained from $\rho_k$ by replacing some traversed leaves in $\rho_k$ with traversed proofs and inserting nodes labeled with \RuDischarge[]. Thus, $\rho_l$ consists of at least the nodes in $\rho_k$ and we can identify nodes in $\rho_k$ with nodes in $\rho_l$.
	
	\begin{mdframed}[align=center,userdefinedwidth=0.8\textwidth]
			\begin{center}
				Fix an arbitrary traversed proof $\rho$ with $\rho_I \leqT \rho$.
			\end{center}
	\end{mdframed}
}

\new{
For every node $v$ in $\rho$ we can find $\rho' \leqT \rho$, where $v$ is the leftmost open leaf. If there are multiple ones, then choose the minimal. Note that we intentionally overuse $v$ to denote the node in $\rho$ and the traversed leaf in $\rho'$. Using the fact that we can identify the node $v$ in $\rho$ with the traversed leaf $v$ in $\rho'$, we define $\depthT(v)$ for a node $v$ in $\rho$ to be the depth $\depthT(v)$ of the traversed leaf $v$ in $\rho'$.

For the termination argument we need some measures on the proof $\pi$. 
Let $n_l \isdef |\Clos(\Sigma_0^-,\phi)|$ and $n_r \isdef |\Clos(\Sigma_1^-,\mybar{\phi})|$, that is, the sizes of the closures of the roots of $\hat{\pi}$ and $\hat{\tau}$, respectively.
Let $m_l$ be the number of nodes in $\hat{\pi}$.

\claim{1}
Nodes in $\rho$ are labeled with at most $2^{2\cdot n_l + n_r}$ many distinct sequents up to $\seteq$. 

\claimproof{1}
For every node $v$ in $\rho$ there is $\rho' \leqT \rho$ such that $v$ is an open leaf in $\rho'$. Let $v$ be labeled with
\begin{align*}
	\begin{prooftree}
		\hypo{\merge{\Pi}{\Psi}{\Tau}}
		\infer[no rule]1{\Gamma_1,...,\Gamma_m, \Delta_1,...,\Delta_n}
	\end{prooftree}.
\end{align*}
Because of the definition of the $\oriNode$ map, for each $i = 1,...,m$ there is a node $w$ in $\hat{\pi}$ such that $\sfS_w = \Gamma_i', \Psi_i^u$ where $\Gamma_i' = \Gamma_i$ or $(\Gamma_i')^u = \Gamma_i$. 
Similarly, for each $j= 1,...,n$ there is a node $w$ in $\hat{\tau}$ such that $\sfS_w = \Delta_j', \Phi_j'$ where $(\Delta_j')^u = \Delta_j$.

Therefore the sequent $\Gamma_1,...,\Gamma_m, \Delta_1,...,\Delta_n$ consists of the union of sequents in $\{\Gamma \subseteq \sfS_w \| w \in \hat{\pi}\}$, $\{\Gamma^u  \| \Gamma \subseteq \sfS_w \text{ and } w \in \hat{\pi}\}$ and $\{\Delta^u  \| \Delta \subseteq \sfS_w \text{ and } w \in \hat{\tau}\}$. Because $\hat{\pi}$ and $\hat{\tau}$ are cut-free, all formulas at nodes $w \in \hat{\pi}$ are in $\Clos(\Sigma_0^-,\phi)$ and formulas at nodes $w \in \hat{\tau}$ are in $\Clos(\Sigma_1^-,\mybar{\phi})$. This implies that only $2 \cdot n_l + n_r$ many annotated formulas occur in $\sfS_v$. Hence, $v$ is labeled with at most $2^{2\cdot n_l + n_r}$ many distinct sequents up to $\seteq$. 
\claimproofend
}
	
	\medskip
	Let $\alpha = a_0...a_d$ be a path in $\rho$ from the root of $\rho$ to a traversed leaf. For every node $a_k$ on $\alpha$ there is $\rho' < \rho$, where $a_k$ is the leftmost open leaf.  If there are multiple ones, then choose the minimal. Note that we intentionally overuse $a_k$ to denote the node in $\rho$ and the traversed leaf in $\rho'$. To fix notation we let $a_k$ in $\rho'$ be labeled with 
	\begin{align*}
		\begin{prooftree}
			\hypo{\merge{\Pi^k}{\Psi^k}{\Tau^k}}
			\infer[no rule]1{\Gamma^k_1,...,\Gamma^k_{m'}, \Delta^k_1,...,\Delta^k_{n'}}
		\end{prooftree}.
	\end{align*}

	\new{
	Let $\Pi = \pi_1,...,\pi_m$ and $\Tau=\tau_1,...,\tau_n$.
	We define corresponding paths $\alpha_i \isdef \oriPath(\pi_i)$ in $\hat{\pi}$ for $i = 1,...,m$ and $\beta_j \isdef \oriPath(\tau_j)$ in $\hat{\tau}$ for $j = 1,...,n$. We call $\alpha_i$ the \emph{$i$-th projection} of $\alpha$ to $\hat{\pi}$.
	}

	\bigskip 
	
	Next we want to show that for every $n$: if a path $\alpha$ in $\rho$ has certain length (depending on $n$) then there are $n$ modal nodes on $\alpha$. For that aim we define $M$ to be the maximal length of a path in $\pi$ without a modal node. Notably, $M$ is smaller than the number of nodes in $\pi$. 
	We write $l(\alpha)$ for the length of a path $\alpha$.
	
	\claim{2}
	Let $\alpha$ be a path in $\rho$ starting from a node $a_0$. Let $s = |\Pi^0|+ |\Tau^0|$. If $l(\alpha) \geq s \cdot 2^{M+1}$, then there is a modal node $a_k$ on $\alpha$. Moreover $|\Pi^k| + |\Tau^k| \leq s\cdot 2^{M+1}$.
	
	\claimproof{2}
	Let $\alpha = a_0a_1...a_k$ be a path without a modal node and let $s_j = |\Pi^j| + |\Tau^j|$ for $j = 0,...,k$. Let $\alpha_i$ be the $i$-th projection of $\alpha$ in $\hat{\pi}$ for $i = 1,...,|\Pi^k|$. Then the first node of $\alpha_i$ is $\oriNode(\pi)$ for some $\pi \in \Pi^0$. Thus the paths $\alpha_i$ form a forest $F_l$ consisting of $|\Pi^0|$ many trees with roots $\oriNode(\pi)$ for $\pi \in \Pi^0$. Analogously, the paths $\beta_j$ in $\hat{\tau}$ form a forest $F_r$. Let $F = F_l \cup F_r$, then $F$ consists of $s$ trees. 
	All rules in the \Focus system have at most two premises. Therefore in all reduction steps of the  \traversedTrans, at most two new proofs are added; hence every node in $F$ has at most two children.\footnote{Note that this would not be possible if we would allow contraction rules in $\pi$, as a reduction with a contraction would potentially double the size of the multicut.} If the modal rule is never applied in $\alpha$, the lengths of all branches in $F$ are bound by $M$. Thus, every tree in $F$ consists of at most $2^M$ nodes and therefore $|F| \leq s \cdot 2^M$.
	
	If a traversed leaf is transformed in the construction, meaning that a child is added, then also one proof of $\Pi$ or $\Tau$ is transformed. After that, we may add \RuDischarge{} rules. But as we reuse \RuDischarge{} rules for all leaves labeled with the same sequent up to $\seteq$, there are at most as many nodes labeled with \RuDischarge{} as other nodes. Let $k$ be the length of $\alpha$, then $s + k/2 \leq |F|$.  Hence $k \leq 2|F| - 2s \leq s \cdot 2^{M+1}$, meaning that after at most $s\cdot 2^{M+1}$ transformations all proofs in $\Pi$ and $\Tau$ must have a modal node at the root.
	In every step of the construction the size of $\Pi \cup \Tau$ is increased by at most one, hence $s_{j+1} \leq s_j +1$ and therefore $s_k \leq s + k \leq 2|F| \leq s \cdot 2^{M+1}$.
	\claimproofend

	\claim{3}
	Let $\alpha$ be a path in $\rho$ starting from the root. For every $n \in \Nat$, if $l(\alpha) \geq 2^{(M+1)\cdot (n +2) + 1}$, then there are at least $n$ modal rules on $\alpha$.
	
	\claimproof{3}
	For the root $r$ of $\rho$ we have $s= |\{\hat{\pi}\}| + |\{\hat{\tau}\}| = 2$.
	We can find modal nodes $b_1,...,b_n$ on $\alpha$ using Claim 2. Doing so the length of the path from $r$ to $b_n$ can be bound by $\sum_{j=1}^n s_{j} \cdot 2^{M+1} = \sum_{j=1}^n 2 \cdot 2^{(M+1)\cdot j} \cdot  2^{M+1} = 2 \cdot \sum_{j=2}^{n+1} 2^{(M+1)\cdot j} \leq 2 \cdot 2^{(M+1) \cdot (n + 2)}$, where $s_j = |\Pi^j| + |\Tau^j|$ corresponds to the number of proofs in the traversed leaf at $b_j$ for $j = 1,...,n$.
	\claimproofend

	Note that in the construction of $\rho$, a modal node is added only if the root of every proof in $\Pi$ is a modal node as well. Hence there are also $n$ modal nodes on every projection $\alpha_i$ for $i = 1,...,m$.
	
	For later use we define the function $f_M(n) = 2^{(M+1)\cdot (n +2) + 1}$
	
	\bigskip
	\new{
		Recall that we defined the depth $\depthT(v)$ of a traversed leaf $v$ as \newline$\depthT(v)= \max\{\depth(\oriNode(\pi)) \| \pi \in \Pi\}$, where $v$ is labeled with 
	\begin{align*}
		\begin{prooftree}
			\hypo{\merge{\Pi}{\Psi}{\Tau}}
			\infer[no rule]1{\Gamma_1,...,\Gamma_m, \Delta_1,...,\Delta_n}
		\end{prooftree}.
	\end{align*}
	If $\alpha = a_0\ldots a_g$ is a path in $\rho$, then $\depthT$ is not increasing on $\alpha$. That is, for $i<j<g$ we have $\depthT(a_i) \geq \depthT(a_j)$. Next, we want to show that if $\alpha$ is of a certain length then $\depthT$ is at some point strictly decreasing.
	
	Recall that $m_l$ is the number of nodes in $\hat{\pi}$.
	}
	
	\claim{4}
	Let $a$ be a node in $\rho$ with $\depth_t(a) = d$. Then between $a$ and every traversed leaf $v$ with $\depth(v) = d$ there are at most $m_l + 2^{n_l^2\cdot n_r}$ many modal nodes.
	
	\claimproof{4}
	Suppose that $v$ is a traversed leaf and $\alpha = a_0a_1\cdots$ is the path from $a = a_0$ to $v$ with more than $m_l + 2^{n_l^2\cdot n_r}$ many modal nodes on $\alpha$. Let $b$ be the lowest node on $\alpha$, such that there are $m_l$ modal rules between $a$ and $b$ and let $\beta$ be the subpath of $\alpha$ from $b$ to $v$. 
	
	\new{
	Let $w_1,...,w_k$ be a path in $\hat{\pi}$, where $\depth(w_j) = d$ for all $j= 1,...,k$. If $k \geq m_l$, then $w_k$ is in a proper cluster. Hence, for all $a_j \in \beta$ and $\pi \in \Pi^j$ we have that if $\depth(\oriNode(\pi))= d$ then $\oriNode(\pi)$ is in a proper cluster. In proper clusters, no \RuF rules are applied. In the construction an \RuF rule is only added if for some $i$ the root of $\pi_i$ is labeled with $\RuF$ and it holds that $\depth(\oriNode(\pi_i)) = \depthT(a_j)$. For nodes in $\beta$ this is not possible, as long the depth of $a_j$ is $d$. Moreover, for every $a_j$ in $\beta$, there is a formula in focus, as there is a $\pi \in \Pi^j$ such that $\oriNode(\pi)$ is in a proper cluster of depth $d$ and the same formulas in focus are added to $\rho$ in the reductions for \RuF and \RuDischarge[]. This is the case as all formulas $\psi \in \Psi$ are out of focus in the proofs $\pi \in \Pi_j$. 
	}
	
	There are more than $2^{n_l^2\cdot n_r}$ modal nodes on $\beta$. \new{By Claim 1} these modal nodes are labeled with at most $2^{n_l^2 \cdot n_r}$ many sequents up to $=_{Set}$. Hence there are modal nodes $c$ and $w$ in $\beta$, such that $\sfS(c) =_{Set} \sfS(w)$. On the path from $c$ to $w$ there is a modal rule applied, no \RuF rules are applied and all sequents have a formula in focus. Hence the path from $c$ to $w$ is successful and the node $w$ would get discharged in the construction. This contradicts the fact that the path $\alpha$ has more than $m_l + 2^{n_l^2\cdot n_r}$ modal nodes. 
	\claimproofend

	Let $d = \depth(\hat{\pi})$. Iterating Claim 4 we obtain that for every traversed leaf $v$, the path $\alpha$ from the root of $\rho$ to $v$ has at most $(d+1) \cdot (m_l + 2^{n_l^2\cdot n_r})$ many modal nodes.
	
	Combining this with Claim 3, we obtain that the height of traversed leaves is bound by $f_M((d+1) \cdot (m_l + 2^{n_l^2\cdot n_r}))$. In conclusion, as every constructed tree is finitely branching, after finitely many steps a traversed proof $\rho_T$ without traversed leaves -- a \Focus-proof -- is constructed.		
\end{proof}

\subsection{Example}

Let $\phi, \psi, \chi$ and $\delta$ be the following formulas, with their intuitive meaning written on the right: 
\begin{align*}
	\phi \isdef& \nu x. \lbox x \land \mu y. \ldia y \lor \mybar{p}, \qquad &&\text{``everywhere $\mybar{p}$ is reachable''}\\
	\psi \isdef& \mu x. \ldia x \lor p, \qquad &&\text{``$p$ is reachable''}\\
	\chi \isdef& \mu x. \ldia x \lor q, \qquad &&\text{``$q$ is reachable''}\\
	\delta \isdef& \mu x. \ldia x \lor ({p} \land \mybar{q}), \qquad &&\text{``$p \land \mybar{q}$ is reachable''}\\
\end{align*}
Note that ``$p$ is reachable'' means that there is a finite path to a state where $p$ holds. The negation $\mybar{\delta}$ of $\delta$ translates to $\nu x. \lbox x \land (\mybar{p} \lor q)$ which intuitively means ``everywhere $p$ implies $q$''. 
The negation $\mybar{\phi}$ of $\phi$ is $\mu x. \ldia x \lor \nu y. \lbox y \land p$ and reads as ``there is a reachable state, where everywhere $p$ holds''. 
Thus $\mybar{\phi}$ and $\mybar{\delta}$ imply $\chi$, in other words the sequent $\phi,\delta, \chi$ is valid. We give a \Focus-proof using an important cut with $\psi$:
\begin{align*}
	\begin{prooftree}
		\hypo{\hat{\pi}}
		\infer[no rule]1{{\phi},\psi}
		\hypo{\hat{\tau}}
		\infer[no rule]1{\mybar{\psi},{\delta},\chi}
		\infer2[\RuCut]{{\phi},{\delta},\chi}
	\end{prooftree}.
\end{align*}
where the proofs $\hat{\pi}$ and  $\hat{\tau}$  are given as follows. We let $\gamma \isdef \mu y. \ldia y \lor \mybar{p}$ and mention that $\mybar{\psi} = \nu x. \lbox x \land \mybar{p}$. Note that in $\hat{\tau}$ the cut formula $\mybar{\psi}$ is the only formula containing a $\nu$-operator and is therefore crucial in the successful repeat. In this example we omit annotations of $u$ for readability.
\begin{align*}
	\begin{minipage}{0.47\textwidth}
		\begin{prooftree}
			\hypo{[\phi^f, \psi]^\dx}
			\infer1[\RuBox]{\lbox \phi^f, \ldia \psi}
			\infer1[\RuWeak]{\lbox \phi^f, \ldia \psi, p}
			\hypo{}
			\infer1[\AxLit]{\mybar{p},p}
			\infer1[\RuWeak]{\ldia \gamma, \mybar{p}, \ldia \psi, p}
			\infer1[\RuOr]{\ldia \gamma \lor \mybar{p}^u, \ldia \psi, p}
			\infer1[\RuMu]{\gamma^f, \ldia \psi, p}
			\infer2[\RuAnd]{\lbox \phi \land \gamma^f, \ldia \psi, p}
			\infer1[\RuNu]{\mathllap{b:\quad} \phi^f, \ldia \psi, p}
			\infer1[\RuOr]{\phi^f, \ldia \psi \lor p}
			\infer1[\RuMu]{\mathllap{a:\quad}\phi^f, \psi}
			\infer1[\RuDischarge[\dx]]{\phi^f, \psi}
			\infer1[\RuF]{\phi^u, \psi}
		\end{prooftree}
	\end{minipage}
	\begin{minipage}{0.48\textwidth}
		\begin{prooftree}
			\hypo{[\mybar{\psi}^f,\delta, \chi]^{\dy}}
			\infer1[\RuBox]{\lbox \mybar{\psi}^f,\ldia \delta, \ldia \chi}
			\infer1[\RuWeak]{\mathllap{s:\quad}\lbox \mybar{\psi}^f,\ldia \delta, p \land \mybar{q},\ldia \chi, q}
			\hypo{}
			\infer1[\AxLit]{\mybar{p}^f, p, q}
			\hypo{}
			\infer1[\AxLit]{\mybar{p}^f,\mybar{q}, q}
			\infer2[\RuAnd]{\mybar{p}^f, p \land \mybar{q}, q}
			\infer1[\RuWeak]{\mathllap{t:\quad}\mybar{p}^f,\ldia \delta, p \land \mybar{q},\ldia \chi, q}
			\infer2[\RuAnd]{\lbox \mybar{\psi} \land \mybar{p}^f,\ldia \delta, p \land \mybar{q},\ldia \chi, q}
			\infer1[\RuNu]{\mathllap{r:\quad}\mybar{\psi}^f,\ldia \delta, p \land \mybar{q},\ldia \chi, q}
			\infer1[\RuOr]{\mybar{\psi}^f,\ldia \delta, p \land \mybar{q},\ldia \chi \lor q}
			\infer1[\RuMu]{\mybar{\psi}^f,\ldia \delta, p \land \mybar{q},\chi}
			\infer1[\RuOr]{\mybar{\psi}^f,\ldia \delta \lor (p \land \mybar{q}),\chi}
			\infer1[\RuMu]{\mathllap{w:\quad} \mybar{\psi}^f,{\delta},\chi}
			\infer1[\RuDischarge[\dy]]{\mybar{\psi}^f,{\delta},\chi}
			\infer1[\RuF]{\mybar{\psi}^u,{\delta},\chi}
		\end{prooftree}
	\end{minipage}
\end{align*}

We want to eliminate the important cut as in the construction given in Subsection \ref{subsec.constrImportant}. We start by defining the traversed proof $\rho_I$ as above by
\[\begin{prooftree}
	\hypo{\merge{\hat{\pi}}{\psi}{\hat{\tau}}}
	\infer[no rule]1{\phi,\delta,\chi}
\end{prooftree}.\]
We proceed by reducing $\hat{\pi}$. The last applied rule in $\hat{\pi}$ is \RuF and \new{$\depth(\origin(\hat{\pi}))$} is maximal (there is only one proof). We therefore add \RuF to $\rho_I$. Afterwards the proof is unfolded and then $\psi$ is principal. 
On the right hand side in $\hat{\tau}$ the \RuF rule is ignored and then the proof is unfolded. The following rules \RuMu and \RuOr are non-principal and the cut will be pushed upwards. This yields the following traversed proof. Note that $\hat{\pi}_a$ denotes the subproof of $\hat{\pi}$ rooted at the node $a$.
\[\begin{prooftree}
	\hypo{\merge{\hat{\pi}_a}{\psi}{\hat{\tau}_r}}
	\infer[no rule]1{{\phi}^f,\ldia \delta, p \land \mybar{q},\ldia \chi, q}
	\infer1[\RuOr]{{\phi}^f,\ldia \delta, p \land \mybar{q},\ldia \chi \lor q}
	\infer1[\RuMu]{{\phi}^f,\ldia \delta, p \land \mybar{q},\chi}
	\infer1[\RuOr]{{\phi}^f,\ldia \delta \lor (p \land \mybar{q}),\chi}
	\infer1[\RuMu]{{\phi}^f,{\delta},\chi}
	\infer1[\RuF]{{\phi}^u,{\delta},\chi}
\end{prooftree}\]
Now $\psi$ is principal on both sides and gets reduced. First the reduction for \RuMu is applied and then for \RuOr, giving the following traversed proof.
\[\begin{prooftree}
	\hypo{\merge{\hat{\pi}_b}{\ldia\psi,p}{\hat{\tau}_s,\hat{\tau}_t}}
	\infer[no rule]1{{\phi}^f,\ldia \delta, p \land \mybar{q},\ldia \chi, q,\ldia \delta, p \land \mybar{q},\ldia \chi, q}
	\infer1[\RuContr]{{\phi}^f,\ldia \delta, p \land \mybar{q},\ldia \chi, q}
	\infer1[\RuOr]{{\phi}^f,\ldia \delta, p \land \mybar{q},\ldia \chi \lor q}
	\infer1[\RuMu]{{\phi}^f,\ldia \delta, p \land \mybar{q},\chi}
	\infer1[\RuOr]{{\phi}^f,\ldia \delta \lor (p \land \mybar{q}),\chi}
	\infer1[\RuMu]{{\phi}^f,{\delta},\chi}
	\infer1[\RuF]{{\phi}^u,{\delta},\chi}
\end{prooftree}\]
This traversed proof is not tidy, as $p \not\equic \psi$. We transform it into a tidy traversed proof by adding a cut of lower rank.
\[\begin{prooftree}
	\hypo{\merge{\hat{\pi}_b}{\ldia\psi}{\hat{\tau}_s}}
	\infer[no rule]1{{\phi}^f,\ldia \delta, p \land \mybar{q},\ldia \chi, q,p}
	\hypo{\hat{\tau}_t}
	\infer[no rule]1{\mybar{p},\ldia \delta, p \land \mybar{q},\ldia \chi, q}
	\infer2[\RuCut]{{\phi}^f,\ldia \delta, p \land \mybar{q},\ldia \chi, q,\ldia \delta, p \land \mybar{q},\ldia \chi, q}
	\infer1[\RuContr]{{\phi}^f,\ldia \delta, p \land \mybar{q},\ldia \chi, q}
	\infer1[\RuOr]{{\phi}^f,\ldia \delta, p \land \mybar{q},\ldia \chi \lor q}
	\infer1[\RuMu]{{\phi}^f,\ldia \delta, p \land \mybar{q},\chi}
	\infer1[\RuOr]{{\phi}^f,\ldia \delta \lor (p \land \mybar{q}),\chi}
	\infer1[\RuMu]{{\phi}^f,{\delta},\chi}
	\infer1[\RuF]{{\phi}^u,{\delta},\chi}
\end{prooftree}\]

We continue reducing non-principal rules, until a \RuBox rule is applied on the left branch and the cut formula gets weakened on the right branch.

\[\begin{prooftree}
	\hypo{\merge{\hat{\pi}_a}{\psi}{\hat{\tau}_w}}
	\infer[no rule]1{{\mathllap{v:\quad}\phi}^f,\delta,\chi}
	\infer1[\RuBox]{\lbox{\phi}^f,\ldia \delta,\ldia \chi}
	\infer1[\RuWeak]{\lbox{\phi}^f,\ldia \delta,\ldia \chi, p}
	\hypo{}
	\infer1[\AxLit]{\mybar{p}, p}
	\infer1[\RuWeak]{\ldia \gamma, \mybar{p},\ldia \delta,\ldia \chi, p}
	\infer1[\RuOr]{\ldia \gamma \lor \mybar{p}^u,\ldia \delta,\ldia \chi, p}
	\infer1[\RuMu]{\gamma^f,\ldia \delta,\ldia \chi, p}
	\infer2[\RuAnd]{\lbox{\phi}\land \gamma^f,\ldia \delta,\ldia \chi, p}
	\infer1[\RuNu]{{\phi}^f,\ldia \delta,\ldia \chi, p}
	\infer1[\RuWeak]{{\phi}^f,\ldia \delta, p \land \mybar{q},\ldia \chi, q,p}
	\hypo{\hat{\tau}_t}
	\infer[no rule]1{\mybar{p},\ldia \delta, p \land \mybar{q},\ldia \chi, q}
	\infer2[\RuCut]{{\phi}^f,\ldia \delta, p \land \mybar{q},\ldia \chi, q,\ldia \delta, p \land \mybar{q},\ldia \chi, q}
	\infer1[\RuContr]{{\phi}^f,\ldia \delta, p \land \mybar{q},\ldia \chi, q}
	\infer1[\RuOr]{{\phi}^f,\ldia \delta, p \land \mybar{q},\ldia \chi \lor q}
	\infer1[\RuMu]{{\phi}^f,\ldia \delta, p \land \mybar{q},\chi}
	\infer1[\RuOr]{{\phi}^f,\ldia \delta \lor (p \land \mybar{q}),\chi}
	\infer1[\RuMu]{\mathllap{c:\quad}{\phi}^f,{\delta},\chi}
	\infer1[\RuF]{{\phi}^u,{\delta},\chi}
\end{prooftree}\]
Now the traversed leaf $v$ is labeled with the same sequent as its ancestor $c$ and the path from $c$ to $v$ is successful. We can therefore insert a $\RuDischarge[\dz]$ rule at $c$ and discharge $v$ by $\dz$. This yields a \Focus-proof of $\phi,\delta,\chi$, where the only cut is of lower rank. Note that in the construction of Definition \ref{def.traversedConstruction} this check is only carried out when a \RuBox rule would be applied. Thus the proof would get transformed further until we reach  a node labeled with $\lbox{\phi}^f,\ldia \delta,\ldia \chi$ again and only then discharge the leaf.

%% file: sec.cutUnimportant.tex
\section{Elimination of unimportant cuts}\label{sec.unimportantCuts}

In this section we reduce unimportant cuts to important cuts of the same rank. The strategy is as follows.
Given a proper cluster $S$ with unimportant cuts, we inductively unfold leaves in $S$ and push cuts in $S$ upwards using the cut reductions in Appendix \ref{app.CutReductions}, until we find repeats below all cuts. In $S$ there are no occurrences of \RuF rules and cut reductions for all other rules do not affect formulas in focus. Therefore, if we can find a repeat path below all cuts, this will constitute a successful path -- assuming it contains a modal rule.

In this process, all cuts that are pushed outside of $S$ become important cuts. Due to the presence of contractions we have to work with a generalisation of the \RuCut rule, the \RuMix rule, which allows to introduce the cut formula multiple times and can therefore be seen as a combination of cut and contractions. 

\bigskip\noindent
The \RuMix rule is the following rule:
\[\begin{prooftree}
	\hypo{\Sigma_0,\phi^u,...,\phi^u}
	\hypo{\mybar{\phi}^u,...,\mybar{\phi}^u,\Sigma_1}
	\infer2[\RuMix]{\Sigma_0,\Sigma_1}
\end{prooftree}\]
where $\phi^u$ does not occur in $\Sigma_0$ and $\mybar{\phi}^u$ does not occur in $\Sigma_1$. Note that there are finitely many occurrences of $\phi^u$ in the left premise of \RuMix and there are finitely many occurrences of $\mybar{\phi}^u$ in the right premise and that the amount of occurrences of $\phi^u$ in the left premise might differ from the amount of occurrences of $\mybar{\phi}^u$ in the right premise.

We use the same terminology for \RuMix as we did for \RuCut. For instance, we say that $\phi$ is the mix-formula of the \RuMix rule depicted above and we define the rank of a \RuMix rule as the rank of its mix-formula. 

We let \FocusS be the proof system obtained from \Focus by replacing the cut rule with the mix rule. As \RuMix is a generalisation of \RuCut, every \Focus-proof may be seen as a \FocusS-proof by simply replacing \RuCut rules by \RuMix rules. Conversely, every \FocusS-proof can be translated to a \Focus-proof by replacing \RuMix rules by \RuCut rules and contractions. Importantly, the rank of cut/mix formulas is not affected.

We call a sequent $\Sigma$ \emph{modal}, if all formulas in $\Sigma$ are modal formulas. 
We call a \FocusS-derivation $\pi$ \emph{local}, if $\pi$ does not contain the rules \RuBox, \RuF and \RuDischarge[].

The following lemma deals with the finitary part of the mix-elimination. We can push mixes upwards until all premises of a mix are modal sequents. First we need to define proofs with assumptions.
Recall that a closed leaf is either discharged or labeled with an axiom.

\begin{definition}
	Let $\calA$ be a set of sequents. A \emph{\FocusS-proof with assumptions $\calA$} is a finite \FocusS-derivation $\pi$, where every leaf of $\pi$ is either closed or labeled with a sequent in $\calA$.
	
	A proof $\pi$ with assumptions $\calA$ is called \emph{focused}, if for every assumption $\Gamma$ in $\calA$ that contains a formula in focus, every node on the path from the root of $\pi$ to any occurrence of $\Gamma$ in $\pi$ contains a formula in focus. 
\end{definition}

\new{
In the \FocusS-system we replace the \RuCut rule with a \RuMix rule. To eliminate \RuMix rules we employ \RuMix reduction; these reductions can be defined as the \RuCut reductions in Appendix \ref{app.CutReductions}, where \RuCut rules are replaced by \RuMix rules. If we employ a \RuMix reduction to a rule different from \RuF, then formulas do not lose focus. This is the content of the following lemma and can be observed by looking at the \RuMix reductions.
\begin{proposition}\label{prop.focusedMixRed}
	Let $\pi$ be a focused \FocusS-proof with assumptions $\calA$. Let $\pi'$ be obtained from $\pi$ by applying a mix reduction to $\pi$. Then $\pi'$ is focused.
	\end{proposition}
}

\begin{lemma}\label{lem.unimportantFinitaryCE}
	Let $\calA$ be a set of modal sequents. Let $\pi$ be a local \FocusS-proof with assumptions $\calA$ and only one \RuMix rule of rank $n$ at the root of $\pi$. Then $\pi$ can be transformed to a local \FocusS-proof $\pi'$ with assumptions $\calA$ of the same sequent, where the premises of all \RuMix rules are open assumptions in $\calA$ and all mixes have rank $\leq n$. Additionally, if $\pi$ is focused,  then $\pi'$ is focused as well.
\end{lemma}
\begin{proof}[Proof (Sketch)]
	Note that $\pi$ does not contain \RuDischarge[] rules. Therefore $\pi$ is a finitary proof without cycles and we may employ cut elimination (more precisely: mix elimination) for finitary proofs, see for example \cite{Takeuti1987}. 
	The mix-reductions that are used resemble the cut reductions in Appendix \ref{app.CutReductions}, but then for the more general \RuMix rule. As the focus of this paper is not on cut-elimination for finitary proofs, we omit the details.
	The overall strategy is to inductively ``push the mix upwards'' in $\pi$ until one of its premises is an axiom and the mix can be omitted. In our situation we have to consider the additional case where one of the premises of the mix is an assumption in $\calA$. In this situation, as $\calA$ consists of modal sequents, the mix formula is a modal formula. Then the mix formula is never principal in $\pi$: it does not contain modal rules. Therefore we can push the mix upwards even further until both premises of the mix are open assumptions. 
	\new{Proposition \ref{prop.focusedMixRed} implies that $\pi'$ is focused if $\pi$ is focused. }
\end{proof}

We will now use Lemma \ref{lem.unimportantFinitaryCE} to inductively push mixes upwards until there are enough modal rules below every mix, which guarantees that we find a successful repeat below every mix. First we need the following definition. 

\begin{definition}
	Let $\pi$ be a \FocusS-derivation and $v$ be a node in $\pi$. The \emph{infinite unfolding of $\comp(v)$} in $\pi$, written $\unf[v]{\pi}$, is obtained from $\pi$ by recursively replacing discharged leaves $l$ that are component descendants of $v$ with $\pi_{c(l)}$, and removing nodes labeled with \RuDischarge[\dx] whenever no discharged leaf is labeled with $\dx$.
\end{definition}

\begin{lemma}\label{lem.cutsUnimportant}
	Let $\pi$ be a \Focus-proof of cut-rank $n$ such that all cuts of rank $n$ are unimportant and in the root-cluster. Then we can transform $\pi$ into a \Focus-proof $\pi'$ of the same sequent with cut-rank $\leq n$, where all cuts are important.
\end{lemma}
\begin{proof}
	\new{
	First we need to introduce some notions. Let $\rho$ be a (possibly infinite) \Focus-derivation. 
	The \emph{focused trunk} $\FT(\rho)$ of $\rho$ is the subderivation of $\rho$ with the same root up to the first occurrences of unfocused sequents. Note that the focused trunk of $\rho$ may be infinite.
	Without loss of generality we may assume that between any node in the focused trunk of a derivation $\rho$ and any \RuDischarge{} rule there is a modal node, otherwise we can unfold \RuDischarge{} rules.
	We let the \emph{extended focused trunk} $\RC(\rho)$ be the subderivation of $\rho$ with the same root up to the first occurrences of modal nodes outside of the focused trunk of $\rho$. 
	We define the \emph{$k$-fragment of $\rho$} to be the subderivation of $\RC(\rho)$ with the same root up to the $k$-th application of a modal rule. 
	}
	
	\medskip
	By replacing cuts with mixes we let $\pi$ be a \FocusS-proof without renaming it.
	Let $\Gamma$ be the sequent at the root $r$ of $\pi$ and let $\unf[r]{\pi}$ be the infinite unfolding of $\comp(r)$ of $\pi$. 
	
	\new{We want to push the cuts (also the ones with cut rank $< n$) occurring in $\RC(\unf[r]{\pi})$ upwards until the mix-free subproof of $\RC(\unf[r]{\pi})$ is big enough. This is formalized in the following claim.}

	\claim{1}
	For every $k$ we can construct a \FocusS-derivation $\pi_k$ of $\Gamma$ without open assumptions, where all mixes have mix-rank $\leq n$ and are outside of the $k$-fragment of $\pi_k$. 
	\new{
	Additionally, all \RuMix rules outside of the focused trunk of $\pi_k$ are important.}
	
	\claimproof{1}
	We prove the claim by induction on $k$. For $k=0$ the derivation $\unf[r]{\pi}$ satisfies the requirements. Let $\pi_k$ be a derivation satisfying the requirements of the claim for $k \geq 0$. We construct the desired derivation $\pi_{k+1}$ by an inner induction on the number $l$ of \RuMix rules in the $(k+1)$-fragment of $\pi_k$. If $l=0$, then $\pi_k$ already satisfies the requirements for $k+1$ and we are done. 
	If $l > 0$ let $\sfC$ be an occurrence of a \RuMix rule in the $(k+1)$-fragment of $\pi_k$ such that there is no \RuMix rule above $\sfC$ in the $(k+1)$-fragment of $\pi_k$. Let $\rho$ be the sub-derivation of the $(k+1)$-fragment of $\pi_k$ rooted at the conclusion of $\sfC$ and let $\calA$ be the set of assumptions of $\rho$. Then $\rho$ satisfies the assumption of Lemma \ref{lem.unimportantFinitaryCE} and applying the lemma yields a focused proof $\rho'$, where the premises of all \RuMix rules are open assumptions in $\calA$. We can replace $\rho$ with $\rho'$ in $\pi_k$ and apply the following mix-reduction for all \RuMix rules in $\rho'$:
	\[
	\begin{prooftree}
		\hypo{\Sigma,\phi}
		\infer1[\RuBox]{\ldia \Sigma, \lbox \phi}
		\hypo{\mybar{\phi},...,\mybar{\phi}, \gamma, \Sigma}
		\infer1[\RuBox]{\ldia \mybar{\phi},...,\ldia \mybar{\phi},\lbox \gamma, \ldia \Sigma}
		\infer2[\RuMix]{\lbox \gamma, \ldia\Sigma}
	\end{prooftree}	 
	\qquad \longrightarrow \qquad
	\begin{prooftree}
		\hypo{\Sigma,\phi}
		\hypo{\mybar{\phi},...,\mybar{\phi}, \gamma,  \Sigma}
		\infer2[\RuMix]{\gamma, \Sigma}
		\infer1[\RuBox]{\lbox \gamma, \ldia \Sigma}
	\end{prooftree}
	\]
	This results in a derivation as desired with $l-1$ many occurrences of \RuMix rules in its $(k+1)$-fragment. We can thus apply the inner induction hypothesis to obtain a proof $\pi_{k+1}$ of $\Gamma$ without open assumptions, where all mixes have mix-rank $\leq n$ \new{and all \RuMix rules are outside of the $(k+1)$-fragment of $\pi_{k+1}$.}
	
	\new{All \RuMix rules that were pushed outside of the focused trunk of $\pi_{k+1}$ are important: no cut reduction puts formulas in focus and therefore the conclusion of those \RuMix rules are unfocused, implying that they are important. 
	}
	\claimproofend

	Let $m$ be the number of modal formulas in $\Clos(\Gamma)$ and let $k \isdef 4^m +1$. Let $\pi_k$ be given as in Claim 1. The $k$-fragment of $\pi_k$ is mix-free, therefore all conclusions of modal rules in the $k$-fragment of $\pi_k$ consist of modal formulas in $\Clos(\Gamma)$, where every formula could occur in focus or out of focus.
	Thus, conclusions of such modal rules are labeled with at most $4^m$ many distinct sequents up to $\seteq$. 
	
	\new{
	Each branch $\beta = v_0...v_n$ in the $k$-fragment of $\pi_k$, where at least one of the children of $v_n$ is in $\RC(\pi_k)$, contains $4^m+1$ many modal nodes. Hence, on such a branch $\beta$ there are nodes $v$ and $l$ such that $v$ and $l$ are labeled with the same sequent up to $\seteq$ and such that on the path from $v$ to $l$ a modal rule is applied. As all nodes in the $k$-fragment of $\pi_k$ contain a formula in focus, this implies that the path from $v$ to $l$ is successful.
	}
	
	For each such branch choose the root-most such nodes $v$ and $l$, insert a \RuDischarge[\dx] rule at $v$ with fresh discharge token $\dx$ and replace $l$ with
	\begin{align*}
		\begin{prooftree}
			\hypo{[\sfS_v]^{\dx}}
			\infer1[\RuWeak,\RuContr]{\sfS_l}
		\end{prooftree}.
	\end{align*}
	By König's Lemma this results in a finite \FocusS-proof $\pi'$.  All \RuMix-rules in $\pi'$ are outside of the \new{$k$-fragment of $\pi_k$} and thus are important. Hence, the proof $\pi'$ has mix-rank $\leq n$, where all mixes are important. By replacing all mixes in $\pi'$ with cuts and contractions we obtain the desired \Focus-proof.
\end{proof}

%% file: sec.contractions.tex
\section{Elimination of contractions}\label{sec.contractions}

It is well-known that contractions pose one of the major difficulties to cut elimination. In our case, in the elimination of important cuts, cut-reductions of the multicut with contractions may double the size of the multicut. This ruins our termination proof as we rely on a bound on the size of the multicut. We thus first opt to eliminate contractions from cut-free proofs and aim to prove the following lemma.

\begin{lemma}\label{lem.contrElim}
	Let $\pi$ be a cut-free \Focus-proof. Then there is a cut-free, contraction-free \Focus-proof $\pi'$ of the same sequent.
\end{lemma}

The elimination of contractions shares similarities with the elimination of cuts, in the sense that we treat contractions in trivial clusters differently than contractions in proper clusters. 
In the first step of our procedure we push all contractions in trivial clusters upwards until all contractions are in proper clusters. For this to work we need to be able to swap occurrences of contractions with the rules \RuOr, \RuAnd and \RuEta; in order to do so we first show that those rules are invertible in Subsection \ref{subsec.invertibleRules}. 
Contractions in proper clusters are eliminated in a similar way as unimportant cuts: we push the contraction upwards until we can find successful repeats below them. The proof of termination of this process is more complicated, as we need to find repeats without introducing new contractions. For this purpose we refer to the results on well-quasi-orders from Section \ref{sec.sub.wqo}.

Recall that the depth of a node $v$ in a proof $\pi$ is the maximal number of proper clusters on a path starting from $v$.
\begin{definition}
	We define the \emph{shallow-depth} of a node $v$ in a proof $\pi$ as the maximal length of a path in $\pi$ starting at $v$ and not containing nodes in proper clusters, where the shallow-depth of $v$ equals $0$ if $v$ is in a proper cluster. 
	The \emph{\RuContr-free shallow depth} of $v$ is defined as the shallow-depth without counting nodes labeled with \RuContr.
	
	Let $\sfC$ be an occurrence of a \RuContr rule with conclusion $v$ in a \Focus-proof $\pi$. The \emph{depth} and \emph{shallow-depth} of $\sfC$ are defined as the depth and shallow-depth of $v$, respectively. The \emph{contraction-depth} of a proof $\pi$ is defined as the maximal depth of  an occurrence of a \RuContr rule in $\pi$.
\end{definition}

\subsection{Strongly invertible rules}\label{subsec.invertibleRules}

\begin{definition}
	Let $\pi$ be a \Focus-proof. We call $\pi$ a \FocusM-proof if all occurrences of contractions in $\pi$ are in \emph{proper clusters}.
\end{definition} 
This definition might look unusual at first glance, but recall that our aim is to push contractions in trivial clusters upwards. This is only possible if we can invert the rules \RuOr, \RuAnd and $\RuEta$. However, their invertibility relies on the absence of contractions in trivial clusters higher up in the proof tree. Thus we do not allow such contractions in \FocusM-proofs.

\new{
We say that a \Focus-proof $\pi$ is \emph{$k$-focused} if every node of depth $\geq k$ has a formula in focus.
}
\begin{definition}
	Let $\begin{prooftree}
		\hypo{\Sigma_1}
		\hypo{\cdots}
		\hypo{\Sigma_n}
		\infer3[\Ru]{\Sigma}
	\end{prooftree}$ be a rule in Figure \ref{fig.RulesFocus}. We call \Ru \emph{strongly invertible} in \FocusM, if every \FocusM-proof $\pi$ of $\Sigma$ can be transformed, for every $i = 1,...,n$, to a \FocusM-proof $\pi_i$ of $\Sigma_i$ with the same depth, shallow depth \new{
	and such that for every $k$, if $\pi$ is $k$-focused, then $\pi_i$ is $k$-focused as well.
	}
\end{definition}

\begin{lemma}\label{lem.InvertibleRules}
	The rules \RuOr, \RuAnd and \RuEta are strongly invertible in \FocusM.
\end{lemma}
\begin{proof}
	We only prove that \RuAnd is strongly invertible, the proofs for the other rules are similar.
	Let $\pi$ be a proof of $\phi \land \psi^a, \Sigma$ with depth $m$ and shallow-depth $k$.
	The proof goes by induction on $m$ with an inner induction on $k$.  
	
	If $k>0$ we proceed with a case distinction on the applied rule \Ru at the root of $\pi$. Note that $\Ru \neq \RuContr$ because $\pi$ is a \FocusM-proof. If \Ru = \RuAnd with principal formula $\phi \land \psi^a$, then the proofs rooted at the premises of \Ru are the desired proofs. If \Ru = \RuWeak the transformation is obvious. If any other rule is applied, we transform the proof as follows, where $\pi_1',...,\pi_n'$ are obtained from respectively $\pi_1,...,\pi_n$ by applying the induction hypothesis.
	\[\begin{prooftree}
		\hypo{\pi_1}
		\infer[no rule]1{\phi \land \psi^b, \Sigma_1}
		\hypo{\cdots}
		\hypo{\pi_n}
		\infer[no rule]1{\phi \land \psi^b, \Sigma_n}
		\infer3[\Ru]{\phi \land \psi^a, \Sigma}
	\end{prooftree}
	\qquad	\longrightarrow \qquad
	\begin{prooftree}
		\hypo{\pi_1'}
		\infer[no rule]1{\phi^b, \psi^b, \Sigma_1}
		\hypo{\cdots}
		\hypo{\pi_n'}
		\infer[no rule]1{\phi^b, \psi^b, \Sigma_n}
		\infer3[\Ru]{\phi^a,\psi^a, \Sigma}
	\end{prooftree}
\]

Now assume that $k = 0$, meaning that the root-cluster is proper. Then $\pi$ is the following proof on the left, where $\pi_0$ is the sub-derivation of $\pi$ up to (i) \RuAnd rules with $\phi \land \psi^a$ principal and (ii) nodes outside the root-cluster. We transform $\pi$ to obtain a proof of $\phi^a, \Sigma$ as follows, where $\pi_0^{\phi}$ is obtained from $\pi_0$ by replacing $\phi \land \psi^a$ with $\phi^a$ at every node, analogously for $\pi_0^{\psi}$.
	\[\begin{prooftree}
		\hypo{[\phi \land \psi^a, \Sigma]^\dx}
		\infer[no rule]1{\vdots}
		\infer[no rule]1{\pi_l}
		\infer[no rule]1{\vdots}
		\infer[no rule]1{\phi^b, \Delta}
		\hypo{[\phi \land \psi^a, \Sigma]^\dx}
		\infer[no rule]1{\vdots}
		\infer[no rule]1{\pi_r}
		\infer[no rule]1{\vdots}
		\infer[no rule]1{\psi^b, \Delta}
		\infer2[\RuAnd]{\phi \land \psi^b, \Delta}
		\infer[no rule]1{\vdots}
		\infer[no rule]1{\pi_0}
		\infer[no rule]1{\vdots}
		\infer[no rule]1{\phi \land \psi^a, \Sigma}
		\infer1[\RuDischarge]{\phi \land \psi^a, \Sigma}
	\end{prooftree}
	\qquad	\longrightarrow \qquad
	\begin{prooftree}
		\hypo{[\phi^a, \Sigma]^\dx}
		\hypo{[\phi^a, \Sigma]^\dx}
		\hypo{[\psi^a, \Sigma]^\dy}
		\infer2[\RuAnd]{\phi \land \psi^a, \Sigma}
		\infer[no rule]1{\vdots}
		\infer[no rule]1{\pi_r}
		\infer[no rule]1{\vdots}
		\infer[no rule]1{\psi^b, \Delta}
		\infer[no rule]1{\vdots}
		\infer[no rule]1{\pi_0^{\psi}}
		\infer[no rule]1{\vdots}
		\infer[no rule]1{\psi^a, \Sigma}
		\infer[]1[\RuDischarge[\dy]]{\psi^a, \Sigma}
		\infer2[\RuAnd]{\phi \land \psi^a, \Sigma}
		\infer[no rule]1{\vdots}
		\infer[no rule]1{\pi_l}
		\infer[no rule]1{\vdots}
		\infer[no rule]1{\phi^b, \Delta}
		\infer[no rule]1{\vdots}
		\infer[no rule]1{\pi_0^{\phi}}
		\infer[no rule]1{\vdots}
		\infer[no rule]1{\phi^a, \Sigma}
		\infer1[\RuDischarge]{\phi^a, \Sigma}
	\end{prooftree}
\]
Note that for any node outside the root-cluster labeled with $\phi \land \psi^a, \Gamma$, we inductively obtain proofs of $\phi^a,\Gamma$ and of $\psi^a, \Gamma$ of the same depth. Therefore the above transformation yields a proof of $\phi^a, \Sigma$ of depth $m$ and shallow-depth $0$. An analogous transformation gives a proof of $\psi^a, \Sigma$.

It is clear that in all cases \new{the resulting proof is $k$-focused if $\pi$ is $k$-focused.} We thus have shown that \RuAnd is strongly invertible.
\end{proof}

\subsection{Contractions in trivial clusters}

\begin{lemma}\label{lem.contrTrivialCluster}
	Let $\pi$ be a cut-free \Focus-proof of contraction-depth $m$. Then $\pi$ can be transformed to a cut-free \Focus-proof $\pi'$ of contraction-depth $\leq m$ of the same sequent, where all contractions are in proper clusters.
\end{lemma}
\begin{proof}
	Let $\sfC_1,...,\sfC_n$ be the occurrences of contraction rules in trivial clusters in $\pi$ with respective \RuContr-free shallow depths $d_1,...,d_n$. We prove the lemma by induction on the Dershowitz–Manna ordering on the multiset $\{d_1,...,d_n\}$ induced by the natural order on $\Nat$.
	
	Let $\sfC$ be an occurrence of a contraction rule $$\begin{prooftree}
		\hypo{\phi^a,\phi^a, \Sigma}
		\infer1[\RuContr]{\phi^a, \Sigma}
	\end{prooftree}$$ in a trivial cluster with \RuContr-free shallow depth $d$, such that there is no contraction rule in a trivial cluster in $\pi$ above $\sfC$. Note that the subproof of $\pi$ rooted at the premise of $\sfC$ is a \FocusM-proof. Let \Ru be the rule applied at the premise of $\sfC$. We proceed with a case distinction based on the shape of \Ru. 

	\new{If \Ru = \RuF then we perform the following transformation, where we assume without loss of generality that either $a=b$ or $b = u$.
	\[\begin{prooftree}
		\hypo{[\phi^a, \phi^b,\Sigma']^\dx}
		\infer[no rule]1{\vdots}
		\infer[no rule]1{\phi^a, \phi^b,\Sigma'}
		\infer1[\RuDischarge]{\phi^a, \phi^b,\Sigma'}
		\infer1[\RuF]{\phi^u, \phi^u,\Sigma}
		\infer1[\RuContr]{\phi^u,\Sigma}
	\end{prooftree}
	\qquad	\longrightarrow \qquad
	\begin{prooftree}
		\hypo{[\phi^a,\Sigma']^\dx}
		\infer1[\RuWeak]{\phi^a, \phi^b,\Sigma'}
		\infer[no rule]1{\vdots}
		\infer[no rule]1{\phi^a,\phi^b,\Sigma'}
		\infer1[\RuU]{\phi^a,\phi^a, \Sigma'}
		\infer1[\RuContr]{\phi^a,\Sigma'}
		\infer1[\RuDischarge]{\phi^a,\Sigma'}
		\infer1[\RuF]{\phi^u,\Sigma}
	\end{prooftree}\]}%
	This results in a \Focus-proof with one less contraction rule in a proper cluster and we can therefore apply the induction hypothesis.

	If $\Ru \neq \RuDischarge[]$ and $\phi^a$ is not principal in \Ru, we can exchange the order in which the rules \Ru and \RuContr are applied and thus reduce $d$. 
	
	Otherwise assume that $\phi^a$ is principal in the rule $\Ru = \RuOr$.
	We transform the proof $\pi$ as follows:
	\[
	\begin{prooftree}
		\hypo{\pi'}
		\infer[no rule]1{\phi, \psi,\phi \lor \psi,\Sigma}
		\infer[]1[\RuOr]{\phi \lor \psi,\phi \lor \psi,\Sigma}
		\infer[]1[\RuContr]{\phi \lor \psi,\Sigma}
	\end{prooftree}	
	\qquad \longrightarrow \qquad
	\begin{prooftree}
		\hypo{\pi'}
		\infer[no rule]1{\phi, \psi,\phi \lor \psi,\Sigma}
		\infer[double]1[$\RuOr^I$]{\phi, \phi, \psi, \psi, \Sigma}
		\infer[]1[\RuContr]{\phi, \phi, \psi, \Sigma}
		\infer[]1[\RuContr]{\phi, \psi, \Sigma}
		\infer[]1[\RuOr]{\phi \lor \psi,\Sigma}
	\end{prooftree}
	\] 
	where $\RuOr^I$ describes an application of the invertibility of \RuOr (Lemma \ref{lem.InvertibleRules}). Both introduced contraction rules have \RuContr-free shallow depth $d-1$, thus we may apply the induction hypothesis. 
	
	If $\phi^a$ is principal in a different rule, we can perform similar transformations using the invertibility results shown in Lemma \ref{lem.InvertibleRules}. Note that in all those transformations the depth of $\pi$ remained the same.
\end{proof}

\subsection{Contractions in proper clusters}
	The idea to reduce the depth of contractions in proper clusters is to push contractions upwards until we find successful repeats below all contractions. This resembles the elimination of unimportant cuts in Lemma \ref{lem.cutsUnimportant}. Here we have to be a bit more careful, as in the reductions of the contraction rule we will use the invertibility of \RuOr, \RuAnd and \RuEta -- which only holds for \FocusM-proofs. We have to make sure that we apply reductions only at those nodes $v$, where no contraction rules appear in trivial clusters above $v$. We thus opt to only unfold leaves in the root-component when needed, compared to the proof of Lemma \ref{lem.cutsUnimportant}, where we already started the process with the infinite unfolding of the root-component. 
	
	In Lemma \ref{lem.cutsUnimportant} the algorithm stops when for every path $\tau$ we found a pair of nodes $v,l$ such that $v$ is a descendant of $l$, the path from $v$ to $l$ is successful and $S_v \seteq S_l$. Then we could apply weakenings and contractions at $l$ to obtain a successful repeat. 
	Now we do not want to introduce \RuContr rules and we therefore only demand that $S_v \subseteq S_l$; In this case we only need to apply weakenings to obtain a successful repeat.
	
	In the proof of termination finding such nodes $v,l$ becomes more tricky. Our solution is to use results on well-quasi-orders. Let $\calM_X$ be the set of sequents occurring in a cut-free proof. In Section \ref{sec.sub.wqo} we saw that $(\calM_X,\subseteq)$ forms a well-quasi-order and so we can find a bound $N$, such that on all paths longer than $N$ we can find such nodes $v,l$ as desired.
	
	To guarantee that on every repeat path there is a modal node we need the following technical lemma. It states that in a cut-free, contraction-free proof all repeat paths contain a modal node.
	
\begin{lemma}\label{lem.repeatPathContainsMod}
	Let $\tau$ be a repeat path in a \Focus-derivation that does not contain nodes labeled with \RuCut, \RuContr and \RuF. Then $\tau$ contains a node labeled with \RuBox. 
\end{lemma}
\begin{proof}
	Let $\phi$ and $\psi$ be formulas. We let $\phi \tracestep^- \psi$ if $\phi \tracestep \psi$ and $\phi$ is not a modal formula. The relation $\trace^-$ is defined as the reflexive and transitive closure of $\tracestep^-$. 
	Note that all formulas are assumed to be guarded. Therefore for formulas $\phi$ and $\psi$ with $\phi \neq \psi$ it cannot be the case that both $\phi \trace^- \psi$ and $\psi \trace^- \phi$ hold.
	
	We let $\Clos^-(\phi)$ be the least superset of $\{\phi\}$ that is closed under $\trace^-$ .
	We define $\nmf(\phi) \isdef |\Clos^-(\phi)|$ to be the number of non-modal formulas in $\Clos^-(\phi)$.
	For a sequent $\Sigma$ we define $\nmf(\Sigma)$ to be the \emph{multiset} $\{\nmf(\phi) \| \phi \in \Sigma^-\}$. We let $\lDM$ be the Dershowitz-Manna ordering on multisets of natural numbers induced by the natural ordering on $\Nat$. Let $\Sigma$ be a premise and $\Sigma'$ be the conclusion of a rule \Ru. Then,
	\begin{enumerate}
		\item if \Ru = \RuOr, \RuAnd, \RuEta or \RuWeak then $\nmf(\Sigma) \lDM \nmf(\Sigma')$,
		\item if \Ru = \RuBox then $\nmf(\Sigma) \geDM \nmf(\Sigma')$ and
		\item if \Ru = \RuU then $\nmf(\Sigma) = \nmf(\Sigma')$.
	\end{enumerate}
	This can easily be verified. For instance, for the rule \RuAnd this holds as $\nmf(\phi) < \nmf(\phi \land \psi)$ because of $\phi \not\trace^- \phi \land \psi$.
	Now let $\tau$ be a repeat path where all nodes on $\tau$ are labeled with the rules \RuOr, \RuAnd, \RuBox, \RuEta, \RuU or \RuWeak. First note that $\tau$ cannot only consist of nodes labeled with \RuU. All other rules apart from \RuBox increase $\nmf(\Sigma)$ and the only rule that reduces $\nmf(\Sigma)$ is $\RuBox$. Hence, there has to be a node labeled with \RuBox on $\tau$.	
\end{proof}

\new{
Let the root $r$ of a \Focus-proof $\pi$ be labeled with $\RuDischarge[\disA]$. Recall that the \emph{unfolding of $\pi$} is obtained from $\pi$ by replacing every discharged leaf labeled with $\disA$ with $\pi$, and removing the node $r$.

Note that the unfolding $\pi'$ of a proof $\pi$ may have a greater depth than $\pi$. However, the depth of nodes without formulas in focus does not increase.
Recall that a \Focus-proof $\pi$ is \emph{$k$-focused} if every node of depth $\geq k$ has a formula in focus.
\begin{lemma}\label{lem.ceUnfoldingKeepsFocus}
	Let $k \in \Nat$ and let $\pi$ be a $k$-focused \Focus-proof.
	Let the root of $\pi$ be labeled with $\RuDischarge[\disA]$ and let $\pi'$ be the unfolding of $\pi$.
	Then $\pi'$ is $k$-focused.
\end{lemma}
\begin{proof}
	Every node $v$ in $\pi'$ of depth $\geq k$ is a copy of a node $u$ in $\pi$ of depth $\geq k$. 
\end{proof}

\begin{lemma}\label{lem.ceContrRedKeepsFocus}
	Let $k \in \Nat$ and let $\pi$ be a $k$-focused \Focus-proof.
	Let $\pi'$ be obtained from $\pi$ by applying a contraction reduction from Appendix \ref{app.contrReductions} to the root of $\pi$. 
	Then $\pi'$ is $k$-focused.
\end{lemma}
\begin{proof}
	The case of a \RuDischarge{} rule follows from Lemma \ref{lem.ceUnfoldingKeepsFocus}. In the other cases we use the fact that whenever the strong invertibility of a rule is applied, $k$-focused proofs are transformed to $k$-focused proofs, see Lemma \ref{lem.InvertibleRules}. The lemma then follows straightforwardly.
	\end{proof}
}

\begin{lemma}\label{lem.contrRootCluster}
	Let $\pi$ be a cut-free \Focus-proof with contraction-depth $m$ where contractions only occur in proper clusters and such that the root-cluster of $\pi$ is proper. Then $\pi$ can be transformed to a cut-free \Focus-proof $\pi'$ of the same sequent, where all contractions have depth $<m$. 
\end{lemma}
\begin{proof}
	\new{
		We start by defining some notions. Given a \Focus-proof $\rho$, let $\uptoM{\rho}$ be the set of nodes in $\rho$ of depth $\geq m$. If the depth of $\rho$ is $\geq m$, then $\uptoM{\rho}$ is non-empty and forms a subtree of $\rho$ containing the root of $\rho$.
	Let $\uptoMC{\rho}$ be the maximal connected subset of $\uptoM{\rho}$ with the same root up to (and including) the first occurrences of contraction or discharge rules.

	We call a maximal path $\tau = v_0...v_m$ in $\uptoMC{\rho}$ \emph{critical}, if at least one of the children of $v_m$ is in $\uptoM{\rho}$. A critical path $\tau$ is called \emph{tamed} if there are nodes $v$ and $l$ on $\tau$ such that $v$ is a proper ancestor of $l$ and $\sfS_v \subseteq \sfS_l$ and is called \emph{untamed} otherwise.
	
	Now consider the \Focus-proof $\pi$. Let $\Sigma$ be the sequent at the root $r$ of $\pi$.  Then $\uptoM{\pi}$ is simply the root cluster of $\pi$ and $\uptoMC{\pi}$ consists of the single node $r$. 
	
	We transform $\pi$ using the following algorithm\footnote{Note that $\pi$ may change in the process and consequently $\uptoM{\pi}$ and $\uptoMC{\pi}$ may change as well. In particular, the root cluster may become trivial, and thus $\uptoM{\pi}$ may consist of multiple clusters.}:
	\begin{enumerate}
		\item If all critical paths in $\uptoMC{\pi}$ are tamed, then stop.
		\item Else if there is a node $v$ in a trivial cluster in $\pi$ labeled with an occurrence $\sfC$ of a \RuContr rule such that no contraction rule is applied in a trivial cluster above $v$, then apply a reduction from Appendix \ref{app.contrReductions} to $\sfC$.
		\item Else take a root-most node $v$ in $\uptoM{\pi}$ labeled with $\RuDischarge[\dx]$ and unfold it, meaning that every discharged leaf $l$ labeled with $\dx$ is replaced with $\pi_v$ and the node $v$ is removed.
	\end{enumerate}
	As $\pi$ is a proof, at some point a principal reduction to a contraction rule is applied and therefore at some point the length of all critical paths in $\uptoMC{\pi}$ increases. 
	To show termination it therefore suffices to show that every critical path of a certain length is tamed.
}

	Every node in \new{$\uptoMC{\pi}$} is labeled with a sequent consisting of formulas in $\Clos(\Sigma)$, hence by a multiset over the finite set $X \isdef \Clos(\Sigma)$. 
	As shown in Section \ref{sec.sub.wqo} we have that $\sfM_X = (\calM_X,\subseteq,\norm_{\infty})$ is a normed well-quasi-order. Any untamed critical path in $\uptoMC{\pi}$ corresponds to a bad sequence over $\sfM_X$. We can therefore use the bounds on controlled bad sequences over $\sfM_X$ to obtain a bound on the length of critical paths \new{in $\uptoMC{\pi}$}. It remains to find a control function and a starting value.
	
	Given a premise $\Delta$ and the conclusion $\Delta'$ of a rule \Ru we have that $\norm[\Delta]_{\infty} \leq \norm[\Delta']_{\infty} + 2$. Thus we can choose the control function $f: n \mapsto n+2$, let $t \isdef \norm[\Sigma]_{\infty}$ be the starting value and let $N \isdef \lenf{\sfM_X}{f}(t)$. Any untamed critical path in \new{$\uptoMC{\pi}$} corresponds to an $(f,t)$-controlled bad sequence over $\sfM_X$. But the length of $(f,t)$-controlled bad sequences over $\sfM_X$ is bound by $N$ and therefore the length of untamed critical paths in \new{$\uptoMC{\pi}$} is bound by $N$ as well. This suffices to show termination.
	
	Let $\tilde{\pi}$ be the proof obtained by this algorithm. For any critical path $\tau$ in \new{$\uptoMC{\tilde{\pi}}$ }let $v$ and $l$ be the root-most nodes such that $v$ is a proper ancestor of $l$ and $\sfS_v \subseteq \sfS_l$. We add a node labeled with $\RuDischarge[\dx]$ at $v$ and replace $l$ with
	\[\begin{prooftree}
		\hypo{[\sfS_v]^\dx}
		\infer1[\RuWeak]{\sfS_l}
	\end{prooftree}.\]
	This results in a \Focus-derivation $\rho$. All remaining nodes labeled with contractions were pushed outside of \new{$\uptoM{\rho}$} and therefore have depth $<m$. It remains to show that all repeat leaves are discharged. In all reductions the resulting sequents still have formulas in focus, therefore all sequents in \new{$\uptoMC{\rho}$} have a formula in focus. Clearly no \RuF rules were introduced, hence no node in \new{$\uptoMC{\rho}$} is labeled with \RuF. All newly introduced repeat paths $\tau_l$ do not contain nodes labeled with \RuCut or \RuContr, therefore Lemma \ref{lem.repeatPathContainsMod} implies that there is a modal node on $\tau_l$. Hence, all repeat paths are successful and we obtain a cut-free proof of the same sequent, where all contractions have depth $<m$. 
\end{proof}

We can now combine the Lemmas \ref{lem.contrTrivialCluster} and \ref{lem.contrRootCluster} and prove the elimination of contractions.

\begin{proof}[Proof of Lemma \ref{lem.contrElim}]
	We prove the Lemma by induction on the contraction-depth $m$ of $\pi$. By Lemma \ref{lem.contrTrivialCluster} we can transform $\pi$ to a proof $\pi_0$ with contraction-depth $m$, where all contractions are in proper clusters. 
	We can apply Lemma \ref{lem.contrRootCluster} to every subproof of $\pi_0$ rooted at a proper cluster containing contractions of depth $m$. This yields a cut-free \Focus-proof $\pi'$ of the same sequent with contraction-depth $<m$. The statement then follows by the induction hypothesis.
\end{proof}

%% file: sec.cutElim.tex
\section{Cut elimination theorem}\label{sec.CEtheorem}

We can now put together the elimination of important and unimportant cuts and obtain cut elimination for the \Focus system.
There is one extra step that we have to  carry out, namely to push important cuts upwards until the cut formula is a fixpoint formula.

\begin{definition}\label{def.cutsCritical}
	Let $\pi$ be a \Focus-proof and $\sfC$ be an important cut in $\pi$. We call $\sfC$ \emph{essential} if the cut formula $\psi$ is a fixpoint-formula.
\end{definition}

\begin{lemma}\label{lem.cutsEssential}
	Let $\pi$ be a contraction-free \Focus-proof with only one cut that is important, of rank $n$ and at the root. Then there is a \Focus-proof $\pi'$ of the same sequent with cut-rank $\leq n$, where all cuts are essential.
\end{lemma}
\begin{proof}
	Using the cut reductions from Appendix \ref{app.CutReductions} we can push the cut upwards until all cut formulas are fixpoint-formulas. All resulting cuts are important as well. All cut reductions apart from the one for \RuEta do not increase the syntactic size of the cut formula and in the cut reduction for \RuBox the syntactic size of the cut-formula decreases. As on every repeat path there is an application of \RuBox, the syntactic size of cut formulas decreases until all cut-formulas of rank $n$ are fixpoint-formulas.
\end{proof}

\begin{theorem}[Cut elimination]\label{thm.cutElimination}
	We can transform every \Focus-proof $\pi$ into a cut-free \Focus-proof $\pi'$ of the same sequent.
\end{theorem}
\begin{proof}
	Let $P_1,...,P_k$ be the proper clusters in $\pi$ that do contain cut rules, where $n^u_j$ is the maximal rank of a cut in $P_j$ for $j = 1,..,k$.
	Let $S_1,...,S_m$ be the trivial clusters in $\pi$ that do contain an essential cut rule, where $S_i$ contains a cut of rank $n^e_j$ for $j = 1,...,m$. 
	Let $T_1,...,T_l$ be the trivial clusters in $\pi$ that do contain an important, but not essential cut rule, where $T_i$ contains a cut of rank $n^i_j$ for $j = 1,...,l$. 
	
	We define the \emph{cut-order} $o(\pi)$  of $\pi$ as the multiset $$\{3 \cdot n^u_1 + 2,...,3 \cdot n^u_k +2 , 3 \cdot n^i_1 + 1,..., 3 \cdot n^i_l + 1, 3\cdot n^e_1,...,3\cdot n^e_m\}.$$
	Let $\lDM$ be the Dershowitz-Manna ordering on multisets of natural numbers induced by the natural ordering on $\Nat$.
	 We prove the lemma by $\lDM$-induction on $o(\pi)$. The definition of $o(\pi)$ guarantees that $o(\pi)$ becomes $\lDM$-smaller if either 
	 \begin{enumerate}
	 	\item[(i)] one proper cluster with unimportant cuts of rank $n$ is replaced with multiple important cuts in trivial clusters with rank $\leq n$, or
	 	\item[(ii)] one non-essential, important cut of rank $n$ in a trivial cluster is replaced with multiple essential cuts of rank $n$, or 
	 	\item[(iii)] one essential cut of rank $n$ in a trivial cluster is replaced with multiple cuts of rank $<n$.
	 \end{enumerate}
	
	Let $\pi_0$ be a subproof of $\pi$, where all cuts are in the root-cluster of $\pi_0$ and let $n$ be the cut-rank of $\pi_0$. If the root-cluster is proper then all cuts in the root-cluster of $\pi_0$ are unimportant. Otherwise there is one important cut at the root of $\pi_0$. 
	
	In the first case, Lemma \ref{lem.cutsUnimportant} yields a proof $\pi_1$ with cut-rank $n$, where all cuts of rank $n$ are important. 
	In the second case, $\pi_0$ consists of one cut at the root with two cut-free subproofs rooted at the premises of the cut. We can apply Lemma \ref{lem.contrElim} to these subproofs and transform $\pi_0$ to $\pi_0'$, where $\pi_0'$  does not contain contractions and has one important cut with rank $n$ at the root. 
	If the cut is not essential, then Lemma \ref{lem.cutsEssential} yields a proof $\pi_1$ with cut-rank $n$, where all cuts are essential. Otherwise the cut is essential and Lemma \ref{lem.cutsImportant} yields a proof $\pi_1$ with cut-rank $<n$.
	
	In all cases, we substitute $\pi_0$ with $\pi_1$ in $\pi$ and obtain a proof $\pi'$, where $o(\pi') \lDM o(\pi)$. We can apply the induction hypothesis to obtain a cut-free proof.
\end{proof}

\begin{corollary}
	We can transform every \Focus-proof $\pi$ into a cut-free and contraction-free \Focus-proof $\pi'$ of the same sequent.
\end{corollary}
\begin{proof}
	Combine Theorem \ref{thm.cutElimination} and Lemma \ref{lem.contrElim}.
\end{proof}

%% file: sec.conc.tex
\section{Conclusion}\label{sec.conclusion}

We have presented a syntactic cut-elimination procedure for a cyclic proof system for the alternation-free modal $\mu$-calculus.
Several possible extensions and adaptations of the presented approach are worth mentioning.

First, the result can be readily extended to the \emph{polymodal case}, where a set of modalities is considered.
Perhaps most interesting is the applicability to temporal and dynamic logics -- such as \PDL, \LTL, and \CTL{} -- since these can be viewed as {fragments of the alternation-free $\mu$-calculus}.
Although our cut-elimination result does not apply to them directly, a similar method can be used.
To illustrate this, consider \PDL. As shown in \cite{Carreiro2014}, \PDL corresponds to the completely additive fragment $\muMLca$ of the modal $\mu$-calculus, with translations provided between \PDL and $\muMLca$.
Since $\muMLca$ is a fragment of the alternation-free $\mu$-calculus, our cut-elimination result transfers directly to the \Focus system when restricted to sequents of $\muMLca$-formulas.
Using the translations from \cite{Carreiro2014}, we can define a \Focus-style system \PDLf for \PDL.
	Moreover, translations between \Focus-proofs of $\muMLca$-sequents and \PDLf proofs can be given in both directions. 
However, these translations may introduce cuts, preventing a direct transfer of our cut-elimination result. 
Nonetheless, since annotations and the global soundness condition in \PDLf are simpler than in \Focus, it should be possible to adapt our cut-elimination method and apply it directly to \PDLf without difficulty.
	
Regarding the {extension to more expressive logics}, it is worth investigating whether our technique can be generalised to the full modal $\mu$-calculus. Candidate proof systems include those proposed by Jungteerapanich and Stirling \cite{Jungteerapanich2010,Stirling2014}. 
Our construction relies on a key property of $\AFMC$-formulas $\phi$: either $\phi$ or its negation $\mybar{\phi}$ is not contained in the closure of a $\nu$-formula of the same rank.
Since such formulas can never be in focus, descendants of such a formula are not crucial for the success-condition of repeat paths. 
For general $\muML$ formulas, this property need not hold, and a more sophisticated method would be required to handle the resulting complexity of annotations.

Also of interest is to determine the precise \emph{complexity} of our cut-elimination procedure. As we currently rely on results concerning well-quasi-orders, we can only establish an upper bound comparable to the Ackermann function. Whether the termination argument can be simplified to yield a tighter bound remains an open question. 
	
Finally, given the technical nature of our result, it would be worthwhile to pursue a formal verification of the procedure using an interactive theorem prover, following recent examples such as~\cite{Shillito2023} and~\cite{Borzechowski2025}.

\section*{Acknowledgments}
We would like to thank Yde Venema and the anonymous reviewers for their valuable comments on previous versions of this paper. 

%% file: app.constrImp.tex
\section{Elimination of important cuts}\label{app.constrImportant}

These are the omitted cases from Definition \ref{def.traversedConstruction}. Recall that we pick $i \in \{1,...,m\}$ or $j \in \{1,...n\}$ and reduce $\pi_i$ or $\tau_j$. 
Let $\psi \in \Psi$ be a formula and let $\pi_i$ and $\tau_j$ be cut-connected proofs via $\psi$ in respectively $\Pi$ and $\Tau$. We let $\Psi = \Psi',\psi$; $\Psi_i = \Psi_i',\psi$ and $\Phi_j = \Phi_j', \psi$ as well as $\Pi = \Pi',\pi_i$ and $\Tau = \Tau',\tau_j$.

\begin{itemize}		
	\item \textbf{\RuWeak rule.} If there is an $i$ such that the last applied rule in $\pi_i$ is \RuWeak, where the principal formula is $\psi \in \Psi_i$, then $\pi_i$ is of the form 
	\begin{align*}
		\begin{prooftree}
			\hypo{\pi_i'}
			\infer[no rule]1{\Gamma_i', \Psi_i'}
			\infer1[\RuWeak]{\Gamma_i', \Psi_i', \psi}
		\end{prooftree}
	\end{align*}
	Let $\calM$ be the multicut at $v$ and let $\del{\calM}{\pi_i}{\psi}$ be the multicut obtained from $\calM$ by removing an edge $E_\psi(\pi_i,\tau)$ for some $\tau$. Let $\del{\calM}{\pi_i'}{\psi}$ be the multicut obtained from  $\del{\calM}{\pi_i}{\psi}$ by replacing $\pi_i$ with $\pi_i'$ and let $\calS$ be its conclusion. Then we replace $v$ with
	\begin{align*}
		\begin{prooftree}
			\hypo{\del{\calM}{\pi_i'}{\psi}}
			\infer[no rule]1{\calS}
			\infer1[\RuWeak]{\Gamma_1,...,,\Gamma_m, \Delta_1,...,\Delta_n}
		\end{prooftree}
	\end{align*}
	\new{We define $\oriPath({\pi}_i')$ as the path $\oriPath(\pi_i)$ extended with the child of $\oriNode(\pi_i)$.}
	Analogously if there is a $j$ such that the last applied rule in $\tau_j$ is \RuWeak, where the principal formula is $\psi \in \Phi_j$.

	\item \textbf{Non-principal rule.} If there is an $i$ such that the last applied rule in $\pi_i$ is a rule with principal formula in $\Gamma_i$, then we ``push the cut upwards''. Then $\pi_i$ has the form
	\begin{align*}
		\begin{prooftree}
			\hypo{\pi_i^1}
			\infer[no rule]1{\Gamma_i^1, \Psi_i}
			\hypo{\hdots}				
			\hypo{\pi_i^n}
			\infer[no rule]1{\Gamma_i^n, \Psi_i}
			\infer3[\Ru]{\Gamma_i,\Psi_i}
		\end{prooftree}
	\end{align*}
	We let $\Gamma_1,...,\Gamma_n = \calS', \Gamma_i$. The leaf $v$ is replaced with
	\begin{align*}
		\begin{prooftree}
			\hypo{\merge{\Pi',\pi_i^1}{\Psi}{\Tau}}
			\infer[no rule]1{\calS', \Gamma_i^1, \calD}
			\hypo{\hdots}				
			\hypo{\merge{\Pi',\pi_i^n}{\Psi}{\Tau}}
			\infer[no rule]1{\calS', \Gamma_i^n, \calD}
			\infer3[\Ru]{\calS', \Gamma_i, \calD}
		\end{prooftree}
	\end{align*}
	\new{In this case $\oriNode(\pi_i)$ is also labeled with $\Ru$ and we define $\oriPath(\pi_i^k)$ as the path $\oriPath(\pi_i)$ extended with the child of $\oriNode(\pi_i)$ containing the formulas in $\Gamma_i^k$ for $k = 1,...,n$.}
	Analogously if there is a $j$ such that the last applied rule in $\tau_j$ is a rule with principal formula in $\Delta_j$.
	
\end{itemize}

In the remaining cases a non-modal formula $\psi \in \Psi$ is principal on both sides. 
Assume that $\psi$ is the principal formula in the last applied rule in $\pi_i$ and that $\mybar{\psi}$ is the principal formula in the last applied rule in $\tau_j$. 

\begin{itemize}
	
	\item \textbf{\RuOr rule.} If $\psi \equiv \psi_0 \lor \psi_1$, then $\pi_i$ has the form
	\begin{align*}
		\begin{prooftree}
			\hypo{\pi_i'}
			\infer[no rule]1{\Gamma_i, \Psi_i', \psi_0^u,\psi_1^u}
			\infer1[\RuOr]{\Gamma_i,\Psi_i',\psi_0 \lor \psi_1^u}
		\end{prooftree}
	\end{align*}
	and  $\tau_j$ has the form 
	\begin{align*}
		\begin{prooftree}
			\hypo{\tau_j^0}
			\infer[no rule]1{\Delta_j', \Phi_j', \mybar{\psi_0}^a}
			\hypo{\tau_j^1}
			\infer[no rule]1{\Delta_j', \Phi_j', \mybar{\psi_1}^a}
			\infer2[\RuAnd]{\Delta_j', \Phi_j', \mybar{\psi_0} \land \mybar{\psi_1}^a}
		\end{prooftree}
	\end{align*}
	Then we replace $v$ with 
	\begin{align*}
		\begin{prooftree}
			\hypo{\merge{\Pi',\pi_i'}{\Psi', \psi_0,\psi_1}{\Tau',\tau_j^0,\tau_j^1}}
			\infer[no rule]1{\Gamma_1,...,\Gamma_m, \Delta_1,...,\Delta_j,\Delta_j,...,\Delta_n}
			\infer1[\RuContr]{\Gamma_1,...,\Gamma_m, \Delta_1,...,\Delta_j,...,\Delta_n}
		\end{prooftree}
	\end{align*}
	where $\psi_0$ is cut-connected to $\pi_i'$ and $\tau_j^0$; and $\psi_1$ is cut-connected to $\pi_i'$ and $\tau_j^1$.
	\new{ We define $\oriPath({\pi}_i')$ as the path $\oriPath(\pi_i)$ extended with the child of $\oriNode(\pi_i)$. The path $\oriPath(\tau_j^k)$ is defined as $\oriPath(\tau_j)$ extended with the child of $\oriNode(\tau_j)$ containing the auxiliary formula $\mybar{\psi_k}^a$ for $k=0,1$.}
	
	\item \textbf{\RuAnd rule.} The case $\RuAnd$ is dual to $\RuOr$.
	
	\item \textbf{\RuMu rule.} If $\psi \equiv \mu x. \chi$ then $\pi_i$ has the form
	\begin{align*}
		\begin{prooftree}
			\hypo{\pi_i'}
			\infer[no rule]1{\Gamma_i, \Psi_i', \chi\subst{x}{\mu x.\chi}^u}
			\infer1[\RuMu]{\Gamma_i,\Psi_i',\mu x.\chi^u}
		\end{prooftree}
	\end{align*}
	and $\tau_j$ has the form 
	\begin{align*}
		\begin{prooftree}
			\hypo{\tau_j'}
			\infer[no rule]1{\Delta_j', \Phi_j', \mybar{\chi}\subst{x}{\nu x. \mybar{\chi}}^a}
			\infer1[\RuNu]{\Delta_j', \Phi_j', \nu x. \mybar{\chi}^a}
		\end{prooftree}
	\end{align*}
	Then we replace $v$ with 
	\begin{align*}
		\begin{prooftree}
			\hypo{\merge{\Pi',\pi_i'}{\Psi', \chi\subst{x}{\mu x.\chi}}{\Tau', \tau_j'}}
			\infer[no rule]1{\Gamma_1,...,\Gamma_m, \Delta_1,...,\Delta_n}
		\end{prooftree}
	\end{align*}
	\new{	We define $\oriPath({\pi}_i')$ as the path $\oriPath(\pi_i)$ extended with the child of $\oriNode(\pi_i)$,  and define $\oriPath(\tau_j')$ as the path $\oriPath(\tau_j)$ extended with the child of $\oriNode(\tau_j)$. }
	
	\item \textbf{\RuNu rule.} As $\rho$ is tidy, $\psi \equic \phi$ for every formula $\psi \in \Psi$. Therefore $\psi$ is magenta and due to Proposition \ref{prop.aFree} this means that $\psi$ is never a $\nu$-formula. 
	
	\item \textbf{Axioms.} As $\psi$ is magenta it is never of the form  $p$ or $\mybar{p}$. Hence the last applied rule in $\pi_i$ or $\tau_j$ is not an axiom.
	
\end{itemize}

%% file: app.cutRed.tex
\section{Cut Reductions}\label{app.CutReductions}
In this appendix we present the cut reductions used in this paper. For readability we state the reductions for a simplified \RuCut rule, where $\Sigma_l = \Sigma_r$ and \RuContr rules are applied implicitly; this can be generalised in the obvious way. We also omit the annotations, whenever they are not affected by the cut reductions. Note that the cut formula is always out of focus.

\subsection{Principal cut reductions}

\[
\begin{prooftree}
	\hypo{\pi_0}
	\infer[no rule]1{\Sigma, \phi}
	\hypo{\pi_1}
	\infer[no rule]1{\Sigma, \psi}
	\infer2[ \RuAnd]{\Sigma, \phi \land \psi}
	\hypo{\pi_2}
	\infer[no rule]1{\mybar{\phi}, \mybar{\psi}, \Sigma}
	\infer1[ \RuOr]{\mybar{\phi} \lor \mybar{\psi},\Sigma}
	\infer2[ \RuCut]{\Sigma}
\end{prooftree}	
\qquad	\longrightarrow \qquad
\begin{prooftree}
	\hypo{\pi_0}
	\infer[no rule]1{\Sigma, \phi}
	\hypo{\pi_1}
	\infer[no rule]1{\Sigma, \psi}
	\hypo{\pi_2}
	\infer[no rule]1{\mybar{\phi}, \mybar{\psi}, \Sigma}
	\infer2[ \RuCut]{\mybar{\phi}, \Sigma}
	\infer2[ \RuCut]{\Sigma}
\end{prooftree}
\]

\[
\begin{prooftree}
	\hypo{\pi_0}
	\infer[no rule]1{\Sigma,\phi[\mu x . \phi / x]}
	\infer1[\RuMu]{\Sigma, \mu x . \phi}
	\hypo{\pi_1}
	\infer[no rule]1{\mybar{\phi}[\nu x . \mybar{\phi} / x], \Sigma}
	\infer1[\RuNu]{\nu x . \mybar{\phi}, \Sigma}
	\infer2[ \RuCut]{\Sigma}
\end{prooftree}	
\qquad \longrightarrow \qquad
\begin{prooftree}
	\hypo{\pi_0}
	\infer[no rule]1{\Sigma, \phi[\mu x . \phi / x]}
	\hypo{\pi_1}
	\infer[no rule]1{\mybar{\phi}[\nu x . \mybar{\phi} / x], \Sigma}
	\infer2[ \RuCut]{\Sigma}
\end{prooftree}
\]

\[
\begin{prooftree}
	\hypo{\pi_0}
	\infer[no rule]1{\Sigma,\phi}
	\infer1[ \RuBox]{\ldia \Sigma, \lbox \phi}
	\hypo{\pi_1}
	\infer[no rule]1{\mybar{\phi}, \gamma,  \Sigma}
	\infer1[ \RuBox]{\ldia \mybar{\phi},\lbox \gamma, \ldia \Sigma}
	\infer2[ \RuCut]{\lbox \gamma, \ldia\Sigma}
\end{prooftree}	 
\qquad \longrightarrow \qquad
\begin{prooftree}
	\hypo{\pi_0}
	\infer[no rule]1{\Sigma,\phi}
	\hypo{\pi_1}
	\infer[no rule]1{\mybar{\phi}, \gamma,  \Sigma}
	\infer2[ \RuCut]{\gamma, \Sigma}
	\infer1[ \RuBox]{\lbox \gamma, \ldia \Sigma}
\end{prooftree}
\]

\[
\begin{prooftree}
	\hypo{\pi_0}
	\infer[no rule]1{\Sigma, p}
	\hypo{}
	\infer1[\AxLit]{\mybar{p}, p}
	\infer2[ \RuCut]{\Sigma, p}
\end{prooftree}
\qquad \longrightarrow \qquad
\begin{prooftree}
	\hypo{\pi_0}
	\infer[no rule]1{\Sigma,p}
\end{prooftree}
\]

\[
\begin{prooftree}
	\hypo{\pi_0}
	\infer[no rule]1{\Sigma, \mybar{p}}
	\hypo{}
	\infer1[\AxLit]{p,\mybar{p}}
	\infer2[ \RuCut]{\Sigma, \mybar{p}}
\end{prooftree}
\qquad \longrightarrow \qquad
\begin{prooftree}
	\hypo{\pi_0}
	\infer[no rule]1{\Sigma,\mybar{p}}
\end{prooftree}
\]

\[
\begin{prooftree}
	\hypo{\pi_0}
	\infer[no rule]1{\Sigma}
	\infer1[ \RuWeak]{\Sigma, \phi}
	\hypo{\pi_1}
	\infer[no rule]1{\mybar{\phi},\Sigma}
	\infer2[ \RuCut]{\Sigma}
\end{prooftree} \qquad \longrightarrow \qquad
\begin{prooftree}
	\hypo{\pi_0}
	\infer[no rule]1{\Sigma}
\end{prooftree}
\]

\subsection{Non-principal cut reductions}
We push \RuF and \RuU rules upwards away from the root and unfold \RuDischarge[] rules. The presented reductions are analogous, if the right premise of the cut is labeled with \RuDischarge[], \RuF  or \RuU. Note that we assume that all proofs are minimally focused.

\begin{description}
	
	\item[Case $\mathsf{D}$ ] We unfold companion nodes:
\[
\begin{prooftree}
	\hypo{\pi_0}
	\infer[no rule]1{\Sigma, \phi}
	\infer1[ \RuDischarge]{\mathllap{v:\qquad\quad}\Sigma, \phi}
	\hypo{\pi_1}
	\infer[no rule]1{\mybar{\phi}, \Sigma}
	\infer2[ \RuCut]{\Sigma}
\end{prooftree}
\qquad \longrightarrow \qquad
\begin{prooftree}
	\hypo{\pi_0'}
	\infer[no rule]1{\Sigma, \phi}
	\hypo{\pi_1}
	\infer[no rule]1{\mybar{\phi}, \Sigma}
	\infer2[ \RuCut]{\Sigma}
\end{prooftree}
\]
where $\pi_0'$ is obtained from $\pi_0$ by replacing every discharged leaf labeled with $\dx$ with $\pi_v$, where $v$ is the left premise of the \RuCut rule.\footnote{Here and in the following cut reductions we replace discharge tokens $\dy$ with fresh discharge tokens, whenever a \RuDischarge[\dy] rule is duplicated.}

 \item[Case $\RuF$] \new{ Because we assume that $\pi$ is minimally focused, the premise of an \RuF rule is labeled with $\RuDischarge[]$. We transform hose proofs as follows:}
\[
\begin{prooftree}
	\hypo{\pi_0}
	\infer[no rule]1{\Sigma', \phi^a}
	\infer1[ \RuDischarge]{\Sigma', \phi^a}
	\infer1[ \RuF]{\mathllap{v:\qquad\quad}\Sigma, \phi^u}
	\hypo{\pi_1}
	\infer[no rule]1{\mybar{\phi}^u, \Sigma}
	\infer2[ \RuCut]{\Sigma}
\end{prooftree}
\qquad \longrightarrow \qquad
\begin{prooftree}
	\hypo{\pi_0'}
	\infer[no rule]1{\Sigma, \phi^u}
	\hypo{\pi_1}
	\infer[no rule]1{\mybar{\phi}^u, \Sigma}
	\infer2[ \RuCut]{\Sigma}
\end{prooftree}
\]
where $\pi_0'$ is obtained from $\pi_0$ by (i) unfocusing sequents up to \RuDischarge[] rules and leaves labeled with $\dx$ and (ii) replacing every discharged leaf labeled with $\dx$ with the subproof $\pi_v$, where $v$ is the left premise of the \RuCut rule.

 \item[Case \RuU] 
\new{ Minimally focused proofs only contain two types of \RuU rules.
 First, \RuU rules in trivial clusters, where all formulas in its premise are out of focus. 
 Second, \RuU rules in proper clusters, where a single formula $\psi$ is put out of focus and the rank of $\psi$ is lower than the maximal rank of formulas in focus in the sequent. These rules are applied as soon as possible. In \RuCut rules, the formulas in focus in the premises are the same as the formulas in focus in its conclusion. Therefore, there are no \RuU rules of the second type occurring in the premises of \RuCut rules (as they already could have been applied at its conclusion).
 Therefore, in minimally focused proofs, if a \RuU rule is labeling the premise of a \RuCut, then all the formulas in its premise are out of focus.
 We proceed with an inner case distinction on the rule applied at the premise of \RuU.
 }
 \begin{description}	
 	\item[\it Subcase \RuF] 
 	If the premise of \RuU on is labeled with \RuF, we do the following:	
\[
\begin{prooftree}
	\hypo{\pi_0}
	\infer[no rule]1{\Sigma' , \phi^a}
	\infer1[ \RuDischarge]{\Sigma' , \phi^a}
	\infer1[ \RuF]{\Sigma^u, \phi^u}
	\infer1[\RuU]{\Sigma, \phi^u}
	\hypo{\pi_1}
	\infer[no rule]1{\mybar{\phi}^u, \Sigma}
	\infer2[ \RuCut]{\Sigma}
\end{prooftree}
\qquad \longrightarrow \qquad
\begin{prooftree}
	\hypo{\pi_0'}
	\infer[no rule]1{\Sigma^u,\phi^u}
	\infer1[ \RuU]{\Sigma, \phi^u}
	\hypo{\pi_1}
	\infer[no rule]1{\phi^u, \Sigma}
	\infer2[ \RuCut]{\Sigma}
\end{prooftree}
\]
	where $\pi'_0$ is defined as in the case for \RuF. That is, $\pi_0'$ is obtained from $\pi_0$ by (i) unfocusing sequents up to \RuDischarge[] rules and leaves labeled with $\disA$ and (ii) replacing every discharged leaf labeled with $\disA$ with the subproof $\pi_v$ from $\pi$ rooted at $v$, where $v$ is the premise of the \RuU rule.

\item[\it Subcase \Ru] 
Otherwise the premise of \RuU is labeled with a rule \Ru different from \RuF, \RuDischarge[] and \RuU. We proceed as follows:
\[
	\begin{prooftree}
		\hypo{\pi_1}
		\infer[no rule]1{\Sigma_1}
		\hypo{\cdots}
		\hypo{\pi_n}
		\infer[no rule]1{\Sigma_n}
		\infer3[ \Ru]{\Sigma^u, \phi^u}
		\infer1[ \RuU]{\Sigma, \phi^u}
		\hypo{\pi_0}
		\infer[no rule]1{\mybar{\phi}^u, \Sigma}
		\infer2[ \RuCut]{\Sigma}
	\end{prooftree} 
	\qquad \longrightarrow \qquad
	\begin{prooftree}
		\hypo{\pi_1}
		\infer[no rule]1{\Sigma_1}
		\infer1[ \RuU]{\Sigma_1'}
		\hypo{\cdots}
		\hypo{\pi_n}
		\infer[no rule]1{\Sigma_n}
		\infer1[ \RuU]{\Sigma_n'}
		\infer3[ \Ru]{\Sigma, \phi^u}
		\hypo{\pi_0}
		\infer[no rule]1{\mybar{\phi}^u, \Sigma}
		\infer2[ \RuCut]{\Sigma}
	\end{prooftree}
\]
\end{description}
 
  \item[Case \Ru]
  Let $\Ru$ be a rule different from $\RuBox$, \RuF, \RuU and \RuDischarge[]. Then we transform the proof as follows:
  \begin{multline*}
  	\begin{prooftree}
  		\hypo{\pi_1}
  		\infer[no rule]1{\Sigma_1, \phi^u}
  		\hypo{\cdots}
  		\hypo{\pi_n}
  		\infer[no rule]1{\Sigma_n, \phi^u}
  		\infer3[ \Ru]{\Sigma, \phi^u}
  		\hypo{\pi_0}
  		\infer[no rule]1{\mybar{\phi}^u, \Sigma}
  		\infer2[ \RuCut]{\Sigma}
  	\end{prooftree}	\\
  	\longrightarrow \qquad
  	\begin{prooftree}
  		\hypo{\pi_1}
  		\infer[no rule]1{\Sigma_1, \phi^u}
  		\hypo{\pi_0}
  		\infer[no rule]1{\mybar{\phi}^u, \Sigma}
  		\infer2[ \RuCut]{\Sigma_1}
  		\hypo{\cdots}
  		\hypo{\pi_n}
  		\infer[no rule]1{\Sigma_n, \phi^u}
  		\hypo{\pi_0}
  		\infer[no rule]1{\mybar{\phi}^u, \Sigma}
  		\infer2[ \RuCut]{\Sigma_n}
  		\infer3[ \Ru]{\Sigma}
  	\end{prooftree}
  \end{multline*}
  
\end{description}

%% file: app.contrRed.tex
\section{Reduction of contractions}\label{app.contrReductions}

In this appendix, we introduce the reductions used to ``push up" contractions. Whenever the strong invertibility of a rule $\Ru$ is employed, we denote it by a doubled line an the rule name $\Ru^I$. For readability, we omit annotations whenever they are not affected by the reductions.

\subsection{Principal reductions}
%

\[
\begin{prooftree}
	\hypo{\pi'}
	\infer[no rule]1{\phi, \psi,\phi \lor \psi,\Sigma}
	\infer1[\RuOr]{\phi \lor \psi,\phi \lor \psi,\Sigma}
	\infer1[\RuContr]{\phi \lor \psi,\Sigma}
\end{prooftree}	
\qquad \longrightarrow \qquad
\begin{prooftree}
	\hypo{\pi'}
	\infer[no rule]1{\phi, \psi,\phi \lor \psi,\Sigma}
	\infer[double]1[$\RuOr^I$]{\phi, \phi, \psi, \psi, \Sigma}
	\infer1[\RuContr]{\phi, \phi, \psi, \Sigma}
	\infer1[\RuContr]{\phi, \psi, \Sigma}
	\infer1[\RuOr]{\phi \lor \psi,\Sigma}
\end{prooftree}
\]

\[
\begin{prooftree}
	\hypo{\pi_0}
	\infer[no rule]1{\phi, \phi \land \psi,\Sigma}
	\hypo{\pi_1}
	\infer[no rule]1{ \psi, \phi \land \psi,\Sigma}
	\infer2[\RuAnd] {\phi \land \psi, \phi \land \psi,\Sigma}
	\infer1[\RuContr]{\phi \land \psi,\Sigma}
\end{prooftree}	
\qquad	\longrightarrow \qquad
\begin{prooftree}
	\hypo{\pi_0}
	\infer[no rule]1{\phi, \phi \land \psi,\Sigma}
	\infer[double]1[$\RuAnd^I$]{\phi, \phi,\Sigma}
	\infer1[\RuContr]{\phi,\Sigma}
	\hypo{\pi_1}
	\infer[no rule]1{\psi, \phi \land \psi,\Sigma}
	\infer[double]1[$\RuAnd^I$]{\psi, \psi,\Sigma}
	\infer1[\RuContr]{\psi,\Sigma}
	\infer2[\RuAnd]{\phi \land \psi,\Sigma}
\end{prooftree}
\]

\[
\begin{prooftree}
	\hypo{\pi'}
	\infer[no rule]1{\phi[\eta x . \phi / x], \eta x . \phi,\Sigma}
	\infer1[\RuEta]{\eta x . \phi, \eta x . \phi,\Sigma}
	\infer1[\RuContr]{\eta x . \phi,\Sigma}
\end{prooftree}	
\qquad \longrightarrow \qquad
\begin{prooftree}
	\hypo{\pi'}
	\infer[no rule]1{\phi[\eta x . \phi / x], \eta x . \phi, \Sigma}
	\infer[double]1[$\RuEta^I$]{\phi[\eta x . \phi / x], \phi[\eta x . \phi / x], \Sigma}
	\infer1[\RuContr]{\phi[\eta x . \phi / x], \Sigma}
	\infer1[\RuEta]{\eta x . \phi,\Sigma}
\end{prooftree}
\]

\[
\begin{prooftree}
	\hypo{\pi'}
	\infer[no rule]1{\phi, \psi, \psi,\Sigma}
	\infer1[\RuBox]{\lbox \phi, \ldia \psi, \ldia \psi, \ldia \Sigma}
	\infer1[\RuContr]{\lbox \phi, \ldia \psi, \ldia\Sigma}
\end{prooftree}	 
\qquad \longrightarrow \qquad
\begin{prooftree}
	\hypo{\pi'}
	\infer[no rule]1{\phi, \psi, \psi,\Sigma}
	\infer1[\RuContr]{\phi, \psi, \Sigma}
	\infer1[\RuBox]{\lbox \phi, \ldia \psi, \ldia\Sigma}
\end{prooftree}
\]

\[
\begin{prooftree}
	\hypo{\pi'}
	\infer[no rule]1{\phi,\Sigma}
	\infer1[\RuWeak]{\phi,\phi,\Sigma}
	\infer1[\RuContr]{\phi,\Sigma}
\end{prooftree} \qquad \longrightarrow \qquad
\begin{prooftree}
	\hypo{\pi'}
	\infer[no rule]1{\phi,\Sigma}
\end{prooftree}
\]

\subsection{Non-principal reductions}

\begin{description}
	
\item[Case $\mathsf{D}$ ] 
We unfold \RuDischarge{} rules in the same way that we did in the cut reductions.
\[
\begin{prooftree}
	\hypo{\pi_0}
	\infer[no rule]1{\phi^a,\phi^a,\Sigma}
	\infer1[ \RuDischarge]{\phi^a,\phi^a,\Sigma}
	\infer1[\RuContr]{\phi^a,\Sigma}
\end{prooftree}
\qquad \longrightarrow \qquad
\begin{prooftree}
	\hypo{\pi_0'}
	\infer[no rule]1{\phi^a,\phi^a, \Sigma}
	\infer1[\RuContr]{\phi^a,\Sigma}
\end{prooftree}
\]
where $\pi_0'$ is obtained from $\pi_0$ by replacing every discharged leaf labeled with $\disA$ with $\pi_v$, where $v$ is the premise of the \RuContr rule.\footnote{As in the cut reductions, we replace discharge tokens $\disB$ with fresh discharge tokens, whenever a \RuDischarge[\disB] rule is duplicated.}

\item[Case \RuF] 
Because proofs are minimally focused, the premise of an \RuF rule is labeled with \RuDischarge{}. \new{We reduce those contractions as follows, where we assume without loss of generality that either $a=b$ or $b = u$.
\[\begin{prooftree}
	\hypo{[\phi^a, \phi^b,\Sigma']^\dx}
	\infer[no rule]1{\vdots}
	\infer[no rule]1{\phi^a, \phi^b,\Sigma'}
	\infer1[\RuDischarge]{\phi^a, \phi^b,\Sigma'}
	\infer1[\RuF]{\phi^u, \phi^u,\Sigma}
	\infer1[\RuContr]{\phi^u,\Sigma}
\end{prooftree}
\qquad	\longrightarrow \qquad
\begin{prooftree}
	\hypo{[\phi^a,\Sigma']^\dx}
	\infer1[\RuWeak]{\phi^a, \phi^b,\Sigma'}
	\infer[no rule]1{\vdots}
	\infer[no rule]1{\phi^a,\phi^b,\Sigma'}
	\infer1[\RuU]{\phi^a,\phi^a, \Sigma'}
	\infer1[\RuContr]{\phi^a,\Sigma'}
	\infer1[\RuDischarge]{\phi^a,\Sigma'}
	\infer1[\RuF]{\phi^u,\Sigma}
\end{prooftree}\]}%
Note that we treat \RuDischarge{} rules that are premises of \RuF rules differently than those \RuDischarge{} rules that are not. The reason for that is, that the former occur at the root of proper clusters and the latter occur inside proper clusters. We deal with \RuDischarge{} rules at the root of proper clusters when reducing contractions in trivial clusters, and with \RuDischarge{} rules inside proper clusters when reducing contractions in proper clusters.

\item[Case \RuU] 

We only consider minimally focused proofs. This implies that premises of \RuU rules in trivial clusters are out of focus.
In proper clusters, a formula $\phi$ is put out of focus iff $\phi$ is of a non-maximal rank. Let $v$ be a node in a proper cluster labeled with a contraction rule with principal formula $\phi$. We may assume that the formula $\phi$ is not put out of focus at the premise of the contraction rule --  if $\phi$ is of non-maximal rank it would already be put out of focus at $v$.
Therefore the annotations of both occurrences of $\phi$ in the premise of a \RuU rule are the same. We reduce those \RuU rules as follows.
\[
\begin{prooftree}
	\hypo{\pi'}
	\infer[no rule]1{\phi^b,\phi^b,\Sigma'}
	\infer1[ \RuU]{\phi^a,\phi^a,\Sigma}
	\infer1[\RuContr]{\phi^a, \Sigma}
\end{prooftree}	
\qquad \longrightarrow \qquad
\begin{prooftree}
	\hypo{\pi'}
	\infer[no rule]1{\phi^b,\phi^{b},\Sigma'}
	\infer1[\RuContr]{\phi^{b},\Sigma'}
	\infer1[ \RuU]{\phi^a,\Sigma}
\end{prooftree}
\]

\item[Case \Ru] 
Let $\Ru$ be a rule different from $\RuBox$, \RuU, \RuF and \RuDischarge[] with a principal formula different than $\phi$. Then we reduce $\Ru$ as follows.
\[
	\begin{prooftree}
		\hypo{\pi_1}
		\infer[no rule]1{\phi^a,\phi^a,\Sigma_1}
		\hypo{\cdots}
		\hypo{\pi_n}
		\infer[no rule]1{\phi^a,\phi^a,\Sigma_n}
		\infer3[\Ru]{\phi^a,\phi^a,\Sigma}
		\infer1[\RuContr]{\phi^a, \Sigma}
	\end{prooftree}	
	\qquad \longrightarrow \qquad
	\begin{prooftree}
		\hypo{\pi_1}
		\infer[no rule]1{\phi^a,\phi^a,\Sigma_1}
		\infer1[\RuContr]{\phi^a,\Sigma_1}
		\hypo{\cdots}
		\hypo{\pi_n}
		\infer[no rule]1{\phi^a,\phi^a,\Sigma_n}
		\infer1[\RuContr]{\phi^a,\Sigma_n}
		\infer3[ \Ru]{\phi^a,\Sigma}
	\end{prooftree}
\]

\end{description}